\documentclass[10pt,reqno]{amsart}
\usepackage{latexsym,amsmath}
\usepackage{times}
\usepackage{amsfonts}
\usepackage{amssymb}
\usepackage{latexsym}

\usepackage{bbm}

\newcommand\atom{\delta}
\newcommand\GW{\mathrm{GW}}
\newcommand\thet{\vartheta}

\newcommand{\measurei}{p_{i}(\vec \gamma) \dd \pi_{i, \vec \gamma}(\mu_{\vec \gamma}) }
\newcommand{\measureii}{p_{i,\ell}(\vec \hgamma) \dd \pi_{i, \ell, \vec \hgamma}(\mu_{\vec \hgamma}) }
\newcommand{\hgamma}{\widehat{\gamma}}
\newcommand{\ogamma}{\overline{\gamma}}

\newcommand{\beq}{\begin{equation}} \newcommand{\eeq}{\end{equation}}

\newcommand\dd{{\mathrm d}}
\newcommand\ism{\cong}

\newcommand\G{\vec G}
\newcommand\T{\vec T}

\newcommand\good{tame}

\newcommand\gnp{G(n,p)}
\newcommand\GWP{G_{\mathrm{WP}}}
\newcommand\GGWP{\G_{\mathrm{WP}}}

\newcommand\gnm{G(n,m)}

\numberwithin{equation}{section}

\newcommand\br[1]{\left(#1\right)}

\def\vec#1{\mathchoice{\mbox{\boldmath$\displaystyle#1$}}
{\mbox{\boldmath$\textstyle#1$}}
{\mbox{\boldmath$\scriptstyle#1$}}
{\mbox{\boldmath$\scriptscriptstyle#1$}}}

\newcommand{\Zkc}{Z_{k}}
\newcommand{\Zkb}{Z_{k,\mathrm{bal}}}

\newcommand{\Zkg}{Z_{k,\mathrm{\good}}}

\newcommand{\dc}{d_{k,\mathrm{cond}}}
\newcommand{\dcrit}{d_{k,\mathrm{crit}}}
\newcommand{\dk}{d_{k-\mathrm{col}}}

\DeclareMathOperator{\pr}{\mathbb P}

\newcommand\SIGMA{\vec\sigma}

\newcommand\TAU{\vec\tau}

\newtheorem{definition}{Definition}[section]

\newtheorem{remark}[definition]{Remark}
\newtheorem{theorem}[definition]{Theorem}
\newtheorem{lemma}[definition]{Lemma}
\newtheorem{proposition}[definition]{Proposition}
\newtheorem{corollary}[definition]{Corollary}

\newtheorem{fact}[definition]{Fact}

\usepackage[dvips]{epsfig}
\graphicspath{{./Fig_eps/}{./Eps/}}
\DeclareGraphicsExtensions{.ps,.eps}
\usepackage{color}
\usepackage{fullpage}
\newcommand\core{\mathrm{core}}

\newcommand\Forb{\mathrm{Forb}}
\newcommand\Bal{\mathrm{Bal}}

\newcommand\cA{\mathcal{A}}
\newcommand\cB{\mathcal{B}}
\newcommand\cC{\mathcal{C}}

\newcommand\cF{\mathcal{F}}
\newcommand\cG{\mathcal{G}}
\newcommand\cE{\mathcal{E}}

\newcommand\cN{\mathcal{N}}
\newcommand\cQ{\mathcal{Q}}
\newcommand\cH{\mathcal{H}}
\newcommand\cS{\mathcal{S}}
\newcommand\cT{\mathcal{T}}

\newcommand\cL{\mathcal{L}}
\newcommand\cM{\mathcal{M}}

\newcommand\cP{\mathcal{P}}
\newcommand\cX{\mathcal{X}}
\newcommand\cY{\mathcal{Y}}
\newcommand\cV{\mathcal{V}}
\newcommand\cW{\mathcal{W}}
\newcommand\cZ{\mathcal{Z}}

\def\cR{{\mathcal R}}
\def\cC{{\mathcal C}}
\def\cE{{\mathcal E}}

\newcommand\eul{\mathrm{e}}
\newcommand\eps{\varepsilon}
\newcommand\ZZ{\mathbf{Z}}

\newcommand\Erw{\mathbb{E}}

\newcommand{\vecone}{\vec{1}}

\newcommand{\Vol}{\mathrm{Vol}}

\newcommand{\Po}{{\rm Po}}
\newcommand{\Bin}{{\rm Bin}}
\newcommand{\Be}{{\rm Be}}

\newcommand{\bink}[2] {{{#1}\choose {#2}}}

\newcommand\ra{\rightarrow}

\newcommand\bc[1]{\left({#1}\right)}
\newcommand\cbc[1]{\left\{{#1}\right\}}
\newcommand\bcfr[2]{\bc{\frac{#1}{#2}}}

\newcommand\brk[1]{\left\lbrack{#1}\right\rbrack}

\newcommand\norm[1]{\left\|{#1}\right\|}
\newcommand\abs[1]{\left|{#1}\right|}

\newcommand\RR{\mathbb{R}}
\newcommand\FF{\phi}
\newcommand\RRpos{\RR_{\geq0}}
\newcommand{\Whp}{W.h.p.}
\newcommand{\whp}{w.h.p.}

\newcommand{\tensor}{\otimes}

\newcommand{\Erdos}{Erd\H{o}s}
\newcommand{\Renyi}{R\'enyi}

\newcommand\Lem{Lemma}
\newcommand\Prop{Proposition}
\newcommand\Thm{Theorem}
\newcommand\Cor{Corollary}
\newcommand\Sec{Section}

\begin{document}

\title{The condensation phase transition in random graph coloring$^\star$}

\author[Bapst et al.]{Victor Bapst, Amin Coja-Oghlan, Samuel Hetterich, Felicia Ra\ss mann and Dan Vilenchik}
\thanks{$^\star$ The research leading to these results has received funding from the European Research Council under the European Union's Seventh Framework
			Programme (FP/2007-2013) / ERC Grant Agreement n.\ 278857--PTCC}

\address{Victor Bapst, {\tt bapst@math.uni-frankfurt.de}, Goethe University, Mathematics Institute, 10 Robert Mayer St, Frankfurt 60325, Germany.}

\address{Amin Coja-Oghlan, {\tt acoghlan@math.uni-frankfurt.de}, Goethe University, Mathematics Institute, 10 Robert Mayer St, Frankfurt 60325, Germany.}

\address{Samuel Hetterich, {\tt hetteric@math.uni-frankfurt.de}, Goethe University, Mathematics Institute, 10 Robert Mayer St, Frankfurt 60325, Germany.}

\address{Felicia Ra\ss mann, {\tt rassmann@math.uni-frankfurt.de}, Goethe University, Mathematics Institute, 10 Robert Mayer St, Frankfurt 60325, Germany.}

\address{Dan Vilenchik, {\tt dan.vilenchik@weizmann.ac.il}, Faculty of Mathematics \&\ Computer Science, The Weizmann Institute,  Rehovot, Israel.}

\address{Dan Vilenchik, {\tt dan.vilenchik@weizmann.ac.il}, Facutly of Mathematics \&\ Computer Science, The Weizamnn Institute,  Rehovot, Israel.}

\maketitle

\begin{abstract}
\noindent
Based on a non-rigorous formalism called the ``cavity method'', 
physicists have put forward intriguing predictions on phase transitions in discrete structures.
One of the most remarkable ones is that in problems such as random $k$-SAT or random graph $k$-coloring,
very shortly before the threshold for the existence of solutions there occurs another phase transition called {\em condensation}
	[Krzakala et al., PNAS 2007].
The existence of this phase transition appears to be intimately related
to the difficulty of proving precise results on, e.g., the $k$-colorability threshold as well as to the performance of message passing algorithms.
In random graph $k$-coloring, there is a precise conjecture as to the location of the condensation phase transition in terms
of a distributional fixed point problem.
In this paper we prove this conjecture for $k$ exceeding a certain constant $k_0$.

\end{abstract}

\section{Introduction}

\noindent
{\em Let $\gnp$ denote the random graph on the vertex set $V=\cbc{1,\ldots,n}$ obtained by connecting any two vertices with probability $p\in[0,1]$ independently.
Throughout the paper, 
we are concerned with the setting that $p=d/n$ for a number $d>0$ that remains fixed as $n\ra\infty$.
We say that $G(n,d/n)$ has a property {\em with high probability} (`\whp') if its probability converges to $1$ as $n\ra\infty$.}

\medskip\noindent
The study of random constraint satisfaction problems started with experimental work in the 1990s,
which led to two hypotheses~\cite{Cheeseman,MitchellSelmanLevesque}.
First, that in problems such as random $k$-SAT or random graph coloring
there is a {\em satisfiability threshold}, i.e., a critical ``constraint density'' below which the instance admits
a solution and above which it does not \whp\
Second, that this threshold is associated with the algorithmic ``difficulty'' of actually computing a solution, where
	``difficulty'' has been quantified in various ways, albeit not in the formal sense of computational complexity.
These findings have led to a belief that random instances of $k$-SAT or graph $k$-colorability near the threshold for the existence of solutions
are challenging algorithmic benchmarks, at the very least.

These two hypotheses have inspired theoretical work.
Short of establishing the existence of an actual satisfiability threshold, Friedgut~\cite{Ehud} and Achlioptas and Friedgut~\cite{AchFried} proved
that in random $k$-SAT and random graph $k$-coloring there exists a {\em sharp threshold sequence}.
For instance, in the graph $k$-coloring problem,
this is a sequence $\dk(n)$ that marks the point where the probability of being $k$-colorable drops from $1$ to $0$.%
	\footnote{Formally, for any $k\geq3$ there is a sequence $(\dk(n))_{n}$ such that
		for any fixed $\eps>0$, $G(n,p)$ is $k$-colorable \whp\ if $p<(1-\eps)\dk(n)/n$, while 
		$G(n,p)$ fails to be $k$-colorable \whp\ if $p>(1+\eps)\dk(n)/n$.}
The dependence on $n$ 
allows for the possibility that this point might vary with the number of vertices, although this is broadly conjectured not to be the case.
In fact, proving that $(\dk(n))_{n\geq1}$ converges to a single number $\dk$ is a well-known open problem.
So is determining the location of $\dk(n)$ (or its limit), as~\cite{AchFried} is a pure existence result.

In addition, inspired by predictions from statistical physics, the geometry of the set of solutions of random $k$-SAT or $k$-colorability instances 
has been investigated \cite{Barriers,Molloy}.
The result is that at a certain point well before the satisfiability threshold the set of solutions shatters into a multitude of well-separated ``clusters''.
Inside each cluster, all solutions agree on most of the variables/vertices, the so-called ``frozen" ones.
The average degree $d$ at which these ``frozen clusters'' arise (roughly) matches the point up to which efficient algorithms provably find solutions.
Hence, on the one hand it is tempting to think that there is a connection between clustering and the computational ``difficulty'' of finding a solution~\cite{Barriers,Molloy,LenkaPhD}.
On the other hand, physicists have suggested new {\em message passing algorithms} specifically to cope with a clustered geometry~\cite{BMPWZ,MPZ}.
A satisfactory analysis of these algorithms remains elusive.

The physics predictions are not merely circumstantial or experimental findings.
They derive from a non-rigorous but systematic formalism called the {\em cavity method}~\cite{MM}.
This technique yields, among other things, a prediction as to the precise location of the $k$-SAT or $k$-colorability threshold.
But perhaps even more remarkably, according to the cavity method shortly before
the threshold for the existence of solutions there occurs another phase transition called {\em condensation}~\cite{pnas}.
This phase transition marks a further change in the geometry of the solution space.
While prior to the condensation phase transition each cluster contains
only an exponentially small fraction of all solutions, thereafter a sub-exponential number
of clusters contain a constant fraction of the entire set of solutions.
As we will see in \Sec~\ref{Sec_related} below, the condensation phenomenon seems to hold the key to
a variety of problems, including that of finding the $k$-colorability threshold and
of analyzing message passing algorithms rigorously.
More generally, the physicists' cavity method is extremely versatile.
It has been used to put forward tantalizing conjectures in a variety of areas, including
coding theory, probabilistic combinatorics, compressive sensing and, of course, mathematical physics
	(see~\cite{MM} for an overview).
Hence the importance of providing a rigorous foundation for this technique.

\section{Results}\label{Sec_results}

\noindent
In this paper we prove that, indeed, a condensation phase transition occurs in random graph coloring,
	and that it occurs at the {\em precise} location predicted by the cavity method. 
This is the first rigorous result to determine the exact location of the condensation transition in a model of this kind.
Additionally, the proof yields a direct combinatorial explanation of how this phase transition comes about.

\subsection{Catching a sharp threshold}
To state the result, let us denote by $Z_k(G)$ the number of $k$-colorings of a graph $G$. 
We would like to study the ``typical value'' of $Z_k(G(n,d/n))$ in the limit as $n\ra\infty$.
As it turns out, the correct scaling of this quantity (to obtain a finite limit) is%
		\footnote{In the physics literature, one typically considers $n^{-1}\ln Z$ instead of $Z^{1/n}$, where $Z$ is the so-called ``partition function''.
			We work with 
				the $n$th root because our ``partition function'' $Z_k$ may be equal to $0$.}
	$$\Phi_k(d)\equiv \lim_{n\ra\infty}\Erw[Z_k(G(n,d/n))^{1/n}].$$
In physics terminology, a ``phase transition'' is a point $d_0$ where the function $d\mapsto\Phi_k(d)$ is non-analytic.
However,  the limit $\Phi_k(d)$ is not currently known to exists for all $d,k$.%
		\footnote{It seems natural to conjecture that
			the limit $\Phi_k(d)$ exists for all $d,k$, but proving this might be difficult.
			In fact, the existence of the limit for all $d,k$ would imply that $\dk(n)$ converges.} 
Hence, we need to tread carefully.
For a given $k\geq3$ we call $d_0\in(0,\infty)$ \emph{smooth} if there exists $\eps>0$ such that
\begin{itemize}
\item for any $d\in(d_0-\eps,d_0+\eps)$ the limit $\Phi_k(d)$ exists, and
\item the map $d\in(d_0-\eps,d_0+\eps)\mapsto\Phi_k(d)$ 
		has an expansion as an absolutely convergent power series around $d_0$.
\end{itemize}
If $d_0$ fails to be smooth, we say that a  \emph{phase transition} occurs at $d_0$.

For a smooth $d_0$ the sequence of random variables $(Z_k(G(n,d_0/n))^{1/n})_{n}$ converges to $\Phi_k(d_0)$ in probability.
This follows from a concentration result for the number of $k$-colorings from~\cite{Barriers}.
Hence,  $\Phi_k(d)$ really captures the ``typical'' value of $Z_k(G(n,d/n)$ (up to a sub-exponential factor).

The above notion of ``phase transition'' is in line with the intuition held in  combinatorics.
For instance, the classical result of \Erdos\ and \Renyi~\cite{ER} implies that
the function that maps $d$ to the limit as $n\ra\infty$ of the expected fraction of vertices 
that belong to the largest component of $G(n,d/n)$  
	is non-analytic at $d=1$.
Similarly, if there actually is a sharp threshold $\dk$ for $k$-colorability, then $\dk$ is a phase transition in the above sense.
	\footnote{For $d<\dk$, $G(n,d/n)$ has a $k$-coloring \whp,
		and thus the number of $k$-colorings is, in fact, exponentially large in $n$ as there are $\Omega(n)$ isolated vertices \whp\
		Hence, if $\Phi_k(d)$ exists for $d<\dk$, then $\Phi_k(d)>0$.
		By contrast, for $d>\dk$ the random graph $G(n,d/n)$ fails to be $k$-colorable \whp, and therefore $\Phi_k(d)=0$.
		Thus, $\Phi_k(d)$ cannot be analytic at $\dk$.}

As a next step,
we state (an equivalent but slightly streamlined version of) the physics prediction from~\cite{LenkaFlorent} as to the location of the condensation phase transition.
As most predictions based on the ``cavity method'',
this one comes in terms of a distributional fixed point problem.
To be specific, let $\Omega$ be the set of probability measures on the set $\brk k=\cbc{1,\ldots,k}$.
We identify $\Omega$ with the $k$-simplex, i.e., the set of maps $\mu:\brk k\ra\brk{0,1}$ such that
$\sum_{h=1}^k\mu(h)=1$, equipped 
with the topology and Borel algebra induced by $\RR^k$.
Moreover, we define a map 
	$\cB:\bigcup_{\gamma=1}^\infty\Omega^\gamma\ra\Omega$, $(\mu_1,\ldots,\mu_\gamma)\mapsto\cB[\mu_1,\ldots,\mu_\gamma]$
by letting
	\begin{equation}\label{eqBPOperator}
	\cB[\mu_1,\ldots,\mu_\gamma](i)=\begin{cases}
		\qquad\qquad1/k&\mbox{ if }
			\sum_{h\in\brk k}\prod_{j=1}^\gamma1-\mu_j(h)=0,\\
		\frac{\prod_{j=1}^\gamma1-\mu_j(i)}{\sum_{h\in\brk k}\prod_{j=1}^\gamma1-\mu_j(h)}&\mbox{ otherwise,}
		\end{cases}\qquad\mbox{for any }i\in\brk k.
	\end{equation}

Further, let $\cP$ be the set of all probability measures on $\Omega$.
For each $\mu\in\Omega$ let $\atom_\mu\in\cP$ denote the Dirac measure that puts mass one on the single point $\mu$.
In particular, $\atom_{k^{-1}\vecone}\in\cP$ signifies the measure that puts mass one on the uniform distribution $k^{-1}\vecone=(1/k,\ldots,1/k)$.
For $\pi\in\cP$ and $\gamma\geq0$ let
	\begin{equation}\label{eqZgamma}
	Z_\gamma(\pi)=\sum_{h=1}^k\bc{1-\int_\Omega\mu(h)\dd\pi(\mu)}^\gamma.
	\end{equation}
Further, define a map 
$\cF_{d,k}:\cP\ra\cP$, $\pi\mapsto\cF_{d,k}[\pi]$ by letting
	\begin{align}
	\cF_{d,k}[\pi]&=
		\exp(-d)\cdot\atom_{k^{-1}\vecone}+
		\sum_{\gamma=1}^\infty\frac{\gamma^d\exp(-d)}{\gamma!\cdot Z_\gamma(\pi)}
		\int_{\Omega^\gamma}
			\brk{\sum_{h=1}^k \prod_{j=1}^{\gamma} 1 - \mu_j(h) }
				\cdot\atom_{\cB[\mu_1, \dots, \mu_\gamma]}
			\bigotimes_{j=1}^\gamma  \dd \pi(\mu_j).\label{eqFixedPoint2}
	\end{align}
Thus, in~(\ref{eqFixedPoint2}) we integrate a function with values in $\cP$, viewed as
a subset of the Banach space%
		\footnote{
		To be completely explicit, the probability mass that a measurable set $A\subset\Omega$ carries under $\cF_{d,k}\brk\pi$ is
		$$\cF_{d,k}[\pi](A)=\exp(-d)\cdot\vecone_{\frac1k\vecone\in A}+
		\sum_{\gamma\geq1}\frac{\gamma^d\exp(-d)}{\gamma!\cdot Z_\gamma(\pi)}
		\int
			[\sum_{h=1}^k \prod_{j=1}^{\gamma} 1 - \mu_j(h)]
				\cdot\vecone_{\cB[\mu_1, \dots, \mu_\gamma]\in A}
			\bigotimes_{j=1}^\gamma  \dd \pi(\mu_j),$$
		where $\vecone_{\nu\in A}=1$ if $\nu\in A$ and $\vecone_{\nu\in A}=0$ otherwise.
		}
 of signed measures on $\Omega$.
The normalising term $Z_\gamma(\pi)$ ensures that $\cF_{d,k}[\pi]$ 
really is a probability measure on $\Omega$.

The main theorem is in terms of a fixed point of the map $\cF_{d,k}$, i.e., a point $\pi^*\in\cP$ such that $\cF_{d,k}[\pi^*]=\pi^*$.
In general, the map $\cF_{d,k}$ has several fixed points.
Hence, we need to single out the correct one.
For $h\in\brk k$ let $\atom_h\in\Omega$ denote the 
vector whose $h$th coordinate is one and whose other coordinates are $0$
	(i.e., the  Dirac measure on $h$). 
We call a measure $\pi\in\cP$ {\em frozen} if
$\pi(\cbc{\delta_{1},\ldots,\delta_{k}})\geq2/3$;
in words, the total probability mass concentrated on the $k$ vertices
of the simplex $\Omega$ is at least $2/3$.

\begin{figure}
\footnotesize
\begin{align}
	\nonumber
\FF_{d,k}(\pi) &=\FF_{d,k}^e(\pi)+
	\frac{1}{k} \sum_{i \in [k]}\sum_{\gamma_1, \dots, \gamma_k=0}^\infty
		\FF_{d,k}^v(\pi;{i};\gamma_1,\ldots,\gamma_k)
		 \prod_{h\in [k]}\bcfr{d}{k-1}^{\gamma_h}\frac{\exp(-d/(k-1))}{\gamma_h!},&\mbox{where}\\
\FF_{d,k}^e(\pi)&=- \frac{d}{2k(k-1)} \sum_{h_1=1}^k\sum_{h_2\in\brk k\setminus\cbc{h_1}}
	\int_{\Omega^2} \ln \left[  1 - \sum_{h \in [k]} \mu_{1}(h) \mu_{2}(h) \right] \bigotimes_{i=1}^2\dd\pi_{h_i}(\mu_{i}),\label{eqBethee}\\
\FF_{d,k}^v(\pi;i;\gamma_1,\ldots,\gamma_k)&=
	\left\{\begin{array}{cl}
	\ln k&\mbox{ if }\sum_{i=1}^k\gamma_i=0,\\
	\displaystyle\int_{\Omega^{\gamma_1+\cdots+\gamma_k}}
		\ln \left[ \sum_{h =1}^k \prod_{h' \in [k] \setminus \{{i}\}} \prod_{j=1}^{\gamma_{h'}}
	1 - \mu_{h'}^{(j)}(h)  \right] \bigotimes_{h' \in [k]} \bigotimes_{j=1}^{\gamma_{h'}} \dd \pi_{h'}(\mu_{h'}^{(j)})
	&\mbox{ if }\sum_{i=1}^k\gamma_i>0.
		\end{array}
		\right.
	\label{eqBethev}
\end{align}
\caption{The function $\FF_{d,k}$}\label{Fig_FFdk}
\end{figure}

As a final ingredient, we need a function $\FF_{d,k}:\cP\ra\RR$.
To streamline the notation, for $\pi\in\cP$ and $h\in\brk k$ we write $\pi_h$ for the measure $\dd\pi_h(\mu)=k\mu(h)\dd\pi(\mu)$.
With this notation, $\FF_{d,k}$ is  defined in Figure~\ref{Fig_FFdk}.
The integrals in~(\ref{eqBethee}) and~(\ref{eqBethev}) are well-defined because the 
set where the argument of the logarithm vanishes has measure zero.

\begin{theorem}\label{Thm_cond}
There exists a constant $k_0\geq3$ such that for any $k\geq k_0$ the following holds.
If $d\geq(2k-1)\ln k-2$, then $\cF_{d,k}$ has precisely one frozen fixed point $\pi^*_{d,k}$.
Further, the function
	\begin{equation}\label{eqThm_condPhi}
	\Sigma_k:d\mapsto\ln k+\frac d2\ln(1-1/k)-\FF_{d,k}(\pi^*_{d,k})
	\end{equation}
has a unique zero $\dc$ in the interval $[(2k-1)\ln k-2,(2k-1)\ln k-1]$.
For this number $\dc$ the following three statments hold.
\begin{enumerate}
\item[(i)] Any $0<d<\dc$ is smooth and
		$\Phi_k(d)=k\cdot (1-1/k)^{d/2}.$
\item[(ii)] There occurs a phase transition at $\dc$.
\item[(iii)] If $d>\dc$, then
			$$\limsup_{n\ra\infty}\Erw[Z_k(G(n,d/n))^{1/n}]<k\cdot (1-1/k)^{d/2}.$$
		Thus, if $d$ is smooth, then $\Phi_k(d)<k\cdot (1-1/k)^{d/2}.$
\end{enumerate}
\end{theorem}

The key strength of \Thm~\ref{Thm_cond} and the main achievement of this work is that we identify the {\em precise} location of the phase transition.
In particular, the result $\dc$ is one 
number rather than a ``sharp threshold sequence'' that might vary with $n$.
Admittedly, this precise answer is not exactly a simple one.
But that seems unsurprising, given the intricate combinatorics of the random graph coloring problem.
That said, the proof of \Thm~\ref{Thm_cond} will illuminate matters.
For instance, the fixed point $\pi^*_{d,k}$ turns out to have a nice combinatorial interpretation and, 
perhaps surprisingly, $\pi^*_{d,k}$ emerges to be a {\em discrete} probability distribution.

The above formulas are derived systematically via the cavity method~\cite{MM}.
For instance, the functional $\phi_{d,k}$ is a special case of a general formula, the so-called ``Bethe free entropy''.
Moreover, the map $\cF_{d,k}$ is the distributional version of the ``Belief Propagation'' operator.
In effect, the predictions as to the condensation phase transitions in other problems look very similar to the above.
Consequently, it can be expected that the proof technique developed in the present work carries over to many other problems.

While the main point of \Thm~\ref{Thm_cond} is that it gives an exact answer,
it is not difficult to obtain a simple asymptotic expansion of $\dc$ in the limit of large $k$. Namely, %
	$\dc=(2k-1)\ln k-2\ln 2+\eps_k,$ 
		where $\eps_k\ra 0$ as $k\ra\infty$. 
This asymptotic formula 
	was obtained in~\cite{Danny} by means of a {\em much}
		simpler argument than the one developed in the present paper.
		However, this simpler argument does not quite get
		to the bottom of the combinatorics behind the condensation phase transition.

\subsection{The cluster size}
The proof of \Thm~\ref{Thm_cond} 
allows us to formalise the physicists' notion that as $d$ tends to $\dc$, 
the cluster size approaches the total number of $k$-colorings.
Of course, we need to formalise what we mean by ``clusters'' first.
Thus, let $G$ be a graph on $n$ vertices.
If $\sigma,\tau$ are $k$-colorings of $G$, we define their {\em overlap} as
the $k\times k$-matrix
$\rho(\sigma,\tau)=(\rho_{ij}(\sigma,\tau))_{i,j\in\brk k}$ with entries
	$$\rho_{ij}(\sigma,\tau)=\frac{|\sigma^{-1}(i)\cap\tau^{-1}(j)|}{n},$$
i.e., $\rho_{ij}(\sigma,\tau)$ is the fraction of vertices colored $i$ under $\sigma$ and $j$ under $\tau$.
Now, define the {\em cluster} of $\sigma$ in $G$ as
	\begin{equation}\label{eqCluster}
	\cC(G,\sigma)=\cbc{\tau:\mbox{$\tau$ is a $k$-coloring of $G$ and $\rho_{ii}(\sigma,\tau)\geq0.51/k$ for all $i\in\brk k$}}.
	\end{equation}
Suppose that $\sigma,\tau$ are such that $|\sigma^{-1}(i)|,|\tau^{-1}(i)|\sim n/k$ for all $i\in\brk k$;
most $k$-colorings of $G(n,d/n)$ have this property \whp~\cite{AchFried,Covers}.
Then $\tau\in\cC(G,\sigma)$ means that a little over $50\%$ of the vertices with color $i$ under
$\sigma$ also have color $i$ under $\tau$.
To this extent, $\cC(G,\sigma)$ comprises of colorings ``similar'' to $\sigma$.
In fact, for the range of $d$ that we are interested in,
this definition coincides \whp\ with that
from~\cite{Molloy} (``colorings that can be reached from $\sigma$ by iteratively altering the colors of $o(n)$ vertices at  time'').

\begin{corollary}\label{Cor_entropyCrisis}
With the notation and assumptions of \Thm~\ref{Thm_cond}, the function
	$\Sigma_k$ 
is continuous, strictly positive and monotonically decreasing 
on $((2k-1)\ln k-2,\dc)$, and $\lim_{d\ra\dc}\Sigma_k(d)=0$.
Further, given that $G(n,d/n)$ is $k$-colorable, let $\TAU$ be a uniformly random $k$-coloring of this random graph.
Then for any $d\in((2k-1)\ln k-2,\dc)$,
	\begin{eqnarray*}
	\lim_{\eps\searrow 0}\lim_{n\ra\infty}&\pr&\brk{\frac1n\ln\frac{|\cC(G(n,d/n),\TAU)|}{Z_k(G(n,d/n))}\leq \Sigma_k(d)+\eps\,
			\big|\,\chi(G(n,d/n))\leq k}=1,
		\quad\mbox{and}\\
	\lim_{\eps\searrow 0}\limsup_{n\ra\infty}&\pr&\brk{\frac1n\ln\frac{|\cC(G(n,d/n),\TAU)|}{Z_k(G(n,d/n))}\geq \Sigma_k(d)-\eps
		\,\big|\,\chi(G(n,d/n))\leq k}>0.
	\end{eqnarray*}
\end{corollary}

\noindent
We observe that our conditioning on the chromatic number $\chi(G(n,d/n))$ being at most $k$ is necessary to speak of a random $k$-coloring $\TAU$ but otherwise harmless.
For the first part of \Thm~\ref{Thm_cond} implies that $G(n,d/n)$ is $k$-colorable \whp\ for any $d<\dc$.
Indeed, if $d<\dc$, then $\Phi_k(d)=k(1-1/k)^{d/2}>0$ and thus $\Zkc(G(n,d/n))^{1/n}>0$ \whp\
because $(\Zkc(G(n,d/n))^{1/n})$ converges to $\Phi_k(d)$ in probability.

In words, \Cor~\ref{Cor_entropyCrisis} states that there is a certain function $\Sigma_k>0$
such that the total number of $k$-colorings
exceeds the number of $k$-colorings in the cluster of a randomly chosen $k$-coloring by at least a factor of $\exp[n(\Sigma_k(d)+o(1))]$ \whp\
However, as $d$ approaches $\dc$, $\Sigma_k(d)$ tends to $0$, and
with a non-vanishing probability the gap between the total number of $k$-colorings and the size
of a single cluster is upper-bounded by $\exp[n(\Sigma_k(d)+o(1))]$.

\section{Discussion and related work}\label{Sec_related}
\noindent
In this section we discuss some relevant related work and also explain the impact of \Thm~\ref{Thm_cond}
on some questions that have come up in the literature.

\subsection{The $k$-colorability threshold} 
The problem of determining the chromatic number of random graphs has attracted a great deal of attention since it was first posed by \Erdos\ and \Renyi~\cite{ER}
	(see~\cite{JLR} for a comprehensive overview).
In the case that $p=d/n$ for a fixed real $d>0$, 
the problem amounts to calculating the threshold sequence $\dk(n)$.
The best current bounds are
	\begin{equation}\label{eqThrLoc}
	(2k-1)\ln k-2\ln2+\eps_k
	\leq\liminf_{n\ra\infty}\dk(n) \leq\limsup_{n\ra\infty}\dk(n)\leq(2k-1)\ln k-1+\delta_k,
	\end{equation}
where $\eps_k,\delta_k\ra0$ as $k\ra\infty$.
The upper bound is by the ``first moment'' method~\cite{Covers}.
The lower bound rests on a ``second moment'' argument~\cite{Danny}, which improves a
landmark result of Achlioptas and Naor~\cite{AchNaor}.

While \Thm~\ref{Thm_cond} allows for the possibility that $\dc$ is equal to the
$k$-colorability threshold $\dk$ (if it exists), the physics prediction is that these two are different.
More specifically, the cavity method yields a prediction as to the precise value of $\dk$ in terms
of another distributional fixed point problem.
An asymptotic expansion in terms of $k$ leads to the conjecture
$\dk=(2k-1)\ln k-1+\eta_k$ with $\eta_k\ra0$ as $k\ra\infty$.
Thus, the upper bound in~(\ref{eqThrLoc}) is conjectured to be asymptotically tight in the limit $k\ra\infty$.

The present work builds upon the second moment argument from~\cite{Danny}.
Conversely, \Thm~\ref{Thm_cond} yields a small improvement over the lower bound from~\cite{Danny}.
Indeed, as we saw above \Thm~\ref{Thm_cond} implies that $\liminf_{n\ra\infty}\dk(n)\geq\dc$, thereby
determining the precise ``error term'' $\eps_k$ in the lower bound~(\ref{eqThrLoc}).

In fact, $\dc$ is the best-possible lower bound that can be obtained
via a certain ``natural'' type of second moment argument. 
Assume that $Z\geq0$ is a random variable such that $\ln\Erw[Z(G(n,d/n))]\sim\ln\Erw[\Zkc(G(n,d/n))]$;
	think of $Z$ as a random variable that counts $k$-colorings, perhaps excluding some ``pathological cases''.
Then for any $d$ such that the second moment method ``works'', i.e.,
	$$\Erw[Z(G(n,d/n))^2]\leq O(\Erw[Z(G(n,d/n))])^2,$$
a concentration result from~\cite{Barriers} implies that $\Phi_k(d)=k(1-1/k)^{d/2}$.
Consequently, $d\leq\dc$.

\subsection{``Quiet planting?''}
The notion that for $d$ close to the (hypothetical) $k$-colorability threshold $\dk$
it seems difficult to find a $k$-coloring of $G(n,d/n)$ algorithmically could be used to construct a candidate one-way function~\cite{Barriers}
	(see also~\cite{Goldreich}).
This function maps a $k$-coloring $\sigma$ to a random graph $G(n,p',\sigma)$ by linking any two vertices $v,w$ with $\sigma(v)\neq\sigma(w)$
with some $p'$ independently.
The edge probability $p'$ could be chosen such that the average degree of the resulting graph is close to the $k$-colorability threshold.
This distribution on graphs is the so-called \emph{planted model}. 

If the planted distribution is close to $G(n,d/n)$, one might think that
the function $\sigma\mapsto G(n,p',\sigma)$ is difficult to invert.
Indeed, it should be difficult to find {\em any} $k$-coloring of $G(n,p',\sigma)$, not to mention the planted coloring $\sigma$.
As shown in~\cite{Barriers}, the planted distribution and $G(n,d/n)$ are interchangeable (in a certain precise sense) iff $\Phi_k(d)=k(1-1/k)^{d/2}$.
Hence, $\dc$ marks the point where these two distributions start to differ.
In particular, \Thm~\ref{Thm_cond} shows that at the $k$-colorability threshold, the two distributions are {\em not} interchangeable.
In effect, experimental evidence that coloring $G(n,d/n)$ is ``difficult'' at or near $\dk$ is inconclusive with respect to the
problem of finding a $k$-coloring in the planted model (which may, of course, well be difficult for some other reason).

\subsection{Message passing algorithms}
The cavity method has inspired new ``message passing'' algorithms by the name of Belief/Survey Propagation Guided Decimation~\cite{MPZ}.
Experiments on random graph $k$-coloring instances for small values of $k$ indicate an excellent performance of these algorithms~\cite{BMPWZ,LenkaPhD,LenkaFlorent}.
However, whether these experimental results are reliable and/or extend to larger $k$ remains shrouded in mystery.

For instance, Belief Propagation Guided Decimation can most easily be described in terms of list colorings.
Suppose that $G$ is a given input graph.
Initially, the list of colors available to each vertex is the full set $\brk k$.
The algorithm chooses a color for one vertex at a time as follows.
First, it performs a certain fixed point iteration 
	to approximate for each vertex the marginal probability of taking some color $i$ in a randomly chosen proper list coloring of $G$.
Then, a vertex $v$ is chosen, say, uniformly at random and a random color $i$ is chosen from the
	(supposed) approximation to its marginal distribution.
The color list of $v$ is reduced to the singleton $\cbc i$, color $i$ gets removed from the lists of all the neighbors of $v$, and we repeat.
The algorithm terminates when either for each vertex a color has been chosen (``success'') or
	the list of some vertex becomes empty (``failure'').
Ideally, if at each step the algorithm manages to compute precisely the correct marginal distribution,
the result would be a uniformly random $k$-coloring of the input graph.
Of course, generating such a random $k$-coloring is $\#P$-hard in the worst case, and the crux is that
the aforementioned fixed point iteration may or may not produce a good approximation to the actual marginal distribution.

Perhaps the most plausible stab at understanding Belief Propagation Guided Decimation is the non-rigorous contribution~\cite{RTS}.
Roughly speaking, the result of the Belief Propagation fixed point iteration after $t$ iterations can be expected to yield a good approximation to the actual marginal distribution
iff there is no condensation among the remaining list colorings. 
If so, one should expect that the algorithm actually finds a $k$-coloring 
if condensation does not occur at any step $0\leq t\leq n$.
Thus, we look at a two-dimensional ``phase diagram'' parametrised by the average degree $d$ and the time $t/n$.
We need to identify the line that marks the (suitably defined) condensation phase transition in this diagram.
\Thm~\ref{Thm_cond} deals with the case $t=0$, and it would be most interesting to see if the present techniques
extend to $t\in(0,1)$.
Attempts at (rigorously) analysing message passing algorithms  along these lines have been made for random $k$-SAT,
but the current results are far from precise~\cite{BP,Angelica}.

\subsection{The physics perspective}

In physics terminology the random graph coloring problem is an example of a ``diluted mean-field model of a disordered system''.
The term ``mean-field'' refers to the fact that there is no underlying lattice geometry, while ``diluted'' indicates that
the average degree in the underlying graph is bounded.
Moreover, ``disordered systems'' reflects that the model involves randomness
	(i.e., the random graph). 
Diluted mean-field models are considered a better approximation to ``real'' disordered systems (such as glasses) than models 
where the underlying graph is complete, such as 
the Sherrington-Kirkpatrick model~\cite{MM}.
From the viewpoint of physics, the question of whether ``disordered systems'' exhibit a condensation phase transition can be traced back to
Kauzmann's experiments in the  1940s~\cite{Kauzmann48}.
In models where the underlying graph is complete,
physicsts predicted an affirmative answer in the 1980s~\cite{KiTh87},
and this has long been confirmed rigorously~\cite{TalagrandBook}.

With respect to ``diluted'' models,
Coja-Oghlan and Zdeborova~\cite{Lenka} showed that a condensation phase transition {\em exists} in random $r$-uniform hypergraph $2$-coloring.
Furthermore, \cite{Lenka} determines the location of the condensation phase transition up to an error
$\eps_r$ that tends to zero as the uniformity $r$ of the hypergraph becomes large.
By contrast, \Thm~\ref{Thm_cond} is the first result that pins down the {\em exact} condensation phase transition in a diluted mean-field model.

Technically, we build upon some of the techniques that have been developed to study the ``geometry''
of the set of $k$-colorings of the random graph and add to this machinery.
Among the techniques that we harness is the ``planting trick'' from~\cite{Barriers} (which, in a sense, we are going to ``put into reverse''), the notion
of a core~\cite{Barriers,Danny,Molloy}, techniques for proving the existence of ``frozen variables''~\cite{Molloy}, and a concentration argument from~\cite{Lenka}.
Additionally, our proof directly incorporates some of the physics calculations from~\cite[Appendix~C]{LenkaFlorent}.
That said, the cornerstone of the present work is 
a novel argument that allows us to connect the distributional fixed point problem from~\cite{LenkaFlorent}
rigorously with the  geometry of the set of $k$-colorings.

\bigskip
\noindent {\em From here on we tacitly assume that $k\geq k_0$ for some large enough constant $k_0$ and that $n$ is sufficiently large.
We use the standard $O$-notation when referring to the limit $n\ra\infty$.
Thus, $f(n)=O(g(n))$ means that there exist $C>0$, $n_0>0$ such that for all $n>n_0$ we have $|f(n)|\leq C\cdot|g(n)|$.
In addition, we use the standard symbols $o(\cdot),\Omega(\cdot),\Theta(\cdot)$.
In particular, $o(1)$ stands for a term that tends to $0$ as $n\ra\infty$.

Additionally, we use asymptotic notation with respect to the limit of large $k$.
To make this explicit, we insert $k$ as an index.
Thus, $f(k)=O_k(g(k))$ means that there exist $C>0$, $k_0>0$ such that for all $k>k_0$ we have $|f(k)|\leq C\cdot|g(k)|$.
Further, we write $f(k)=\tilde O_k(g(k))$ to indicate that there exist $C>0$, $k_0>0$ such that 
for all $k>k_0$ we have $|f(k)|\leq (\ln k)^C\cdot|g(k)|$.}

\section{Outline}\label{sec:outline}

\medskip
\noindent
The proof of Theorem \ref{Thm_cond} is composed of two parallel threads.
The first thread is to identify an ``obvious'' point where a phase transition occurs or, more specifically,
a critical degree $\dcrit$ where statements (i)-(iii) of the theorem are met.
The second thread is to identify the frozen fixed point $\pi^*_{d,k}$ of $\cF_{d,k}$ and to interpret it combinatorially.
Finally, the two threads intertwine to show that $\dcrit=\dc$, i.e. that the ``obvious'' phase transition
$\dcrit$ is  indeed  the unique zero of equation \eqref{eqThm_condPhi}.
The first thread is an extension of ideas developed in~\cite{Lenka} for random hypergraph $2$-coloring
to the (technically more involved) random graph coloring problem.
The second thread and the intertwining of the two require novel arguments.

\subsection{The first thread}
Because the $n$th root sits inside the expectation, the quantity $$\Phi_k(d)=\lim_{n\ra\infty}\Erw[Z_k(G(n,d/n))^{1/n}]$$
is difficult to calculate for general values of $d$.
However for $d\in [0,1)$, $\Phi_k(d)$ is easily understood.
In fact, the celebrated result of \Erdos\ and \Renyi~\cite{ER} implies that for $d\in[0,1)$ the random graph $G(n,d/n)$ is basically a forest.
Moreover, the number of $k$-colorings of a forest with $n$ vertices and $m$ edges is well-known to be $k^{n}(1-1/k)^{m}$.
Since $G(n,d/n)$ has $m\sim dn/2$ edges \whp, we obtain
	\begin{equation}\label{eqTreeHug1}
	Z_k(G(n,d/n))^{1/n}\sim k(1-1/k)^{d/2}\qquad\mbox{for }d<1.
	\end{equation}
As $Z_k(G)^{1/n}\leq k$ for any graph on $n$ vertices, (\ref{eqTreeHug1}) implies that
	\begin{equation}\label{eqTreeHug2}
	\Phi_k(d)=\lim_{n\ra\infty}\Erw[Z_k(G(n,d/n))^{1/n}]=k(1-1/k)^{d/2}\qquad\mbox{for }d<1.
	\end{equation}
Clearly, the function $d\mapsto k(1-1/k)^{d/2}$ is analytic on all of $(0,\infty)$. 
Therefore, the uniqueness of analytic continuations implies that the least $d>0$ where the limit $\Phi_k(d)$  either fails to exist or
strays away from $k(1-1/k)^{d/2}$ is going to be a phase transition.
Hence, we let
	\begin{equation}\label{eqdcrit}
	\dcrit=\sup\cbc{d\geq0:
		\mbox{the limit $\Phi_k(d)$ exists and }
			\Phi_k(d)=k(1-1/k)^{d/2}}.
	\end{equation}

\begin{fact}
We have $\dcrit\leq(2k-1)\ln k$.
\end{fact}
\begin{proof}
The upper bound~(\ref{eqThrLoc}) on the $k$-colorability threshold implies that
for $d>(2k-1)\ln k$, $G(n,d/n)$ fails to be $k$-colorable \whp\
Hence, for such $d$ we have $\Zkc(G(n,d/n))=0$ \whp, and thus $\Phi_k(d)=0$.
By contrast, $k(1-1/k)^{d/2}>0$ for any $d>0$. 
\end{proof}

Thus, $\dcrit$ is a well-defined finite number, and there occurs a phase transition at $\dcrit$.
Moreover, the following proposition 
yields a lower bound on $\dcrit$ and
implies that $\dcrit$ satisfies the first condition in \Thm~\ref{Thm_cond},
see \Sec~\ref{Sec_dcrit} for the proof.

\begin{proposition}\label{Lemma_dcrit}
For any $d>0$ we have $\limsup_{n\ra\infty}\Erw[Z_k(G(n,d/n))^{1/n}]\leq k(1-1/k)^{d/2}$.
Moreover, the number $\dcrit$ satisfies
	\begin{equation}\label{eqdcrit1}
	\dcrit=\sup\cbc{d\geq0:
		\liminf_{n\ra\infty}\Erw[Z_k(G(n,d/n))^{1/n}]\geq k(1-1/k)^{d/2}}\geq(2k-1)\ln k-2.
	\end{equation}
\end{proposition}

Thus, we know that there {\em exists} a number $\dcrit$ that satisfies conditions (i)--(ii) in \Thm~\ref{Thm_cond}.
Of course, to actually calculate this number we need to unearth its combinatorial ``meaning''.
As we saw in \Sec~\ref{Sec_results}, if $\dcrit$ really is the condensation phase transition,
then the combinatorial interpretation should be as follows.
For $d<\dcrit$, the size of the cluster that a randomly chosen $k$-coloring $\vec\tau$ belongs to
is smaller than $\Zkc(G(n,d/n))$ by an exponential factor $\exp(\Omega(n))$ \whp\
But as $d$ approaches $\dcrit$, the gap between the cluster size and $\Zkc(G(n,d/n))$ diminishes.
Hence, $\dcrit$ should mark the point where the cluster size has the same order of magnitude as $\Zkc(G(n,d/n))$.

But how can we possibly get a handle on the size of the cluster that a randomly chosen $k$-coloring $\vec\tau$ of $G(n,d/n)$ belongs to?
No ``constructive'' argument (or efficient algorithm) is known for obtaining a single $k$-coloring of $G(n,d/n)$ for $d$ anywhere close to $\dk$,
let alone for sampling one uniformly at random.
Nevertheless, as observed in~\cite{Barriers}, in the case that $\Phi_k(d)=k(1-1/k)^{d/2}$, i.e., for $d<\dcrit$,
it is possible to capture the experiment of first choosing the random graph $G(n,d/n)$ and then sampling a $k$-coloring $\vec\tau$
uniformly at random by means of a different, much more innocent experiment.

In this latter experiment, we {\em first} choose a map $\SIGMA:\brk n\ra\brk k$ uniformly at random.
Then, we generate a graph $G(n,p',\SIGMA)$ on $\brk n$ by connecting any two vertices $v,w\in\brk n$ such that $\SIGMA(v)\neq\SIGMA(w)$ with probability $p'$ independently.
If $p'=dk/(k-1)$ is chosen so that the expected number of edges is the same as in $G(n,d/n)$
	and if $\Phi_k(d)=k(1-1/k)^{d/2}$, then this so-called {\em planted model} is a good approximation to the 
``difficult'' experiment of first choosing $G(n,d/n)$ and then picking a random $k$-coloring.
In particular, we expect that 
	$$\Erw[|\cC(G(n,p',\SIGMA),\SIGMA)|^{1/n}]\sim\Erw[|\cC(G(n,d/n),\TAU)|^{1/n}],$$
i.e., that the suitably scaled cluster size in the planted model is about the same as the cluster size in $G(n,d/n)$.
Hence, $\dcrit$ should mark the point where $\Erw[|\cC(G(n,p',\SIGMA),\SIGMA)|^{1/n}]$ equals $k(1-1/k)^{d/2}$.
The following proposition verifies that this is indeed so.
Let us write $\G=G(n,p',\SIGMA)$ for the sake of brevity.

\begin{proposition}\label{Lemma_plantedCluster}
Assume that $(2k-1)\ln k-2\leq d\leq(2k-1)\ln k$ and set
	\begin{equation}\label{eqLemma_plantedCluster1}
	p'=d'/n\quad\mbox{with }d'=\frac{dk}{k-1}.
	\end{equation}
\begin{enumerate}
\item
If
		\begin{equation}\label{eqLemma_plantedCluster2}
		\lim_{\eps\searrow0}\liminf_{n\ra\infty}\pr\brk{|\cC(\G,\vec\sigma)|^{1/n}\leq k(1-1/k)^{d/2}-\eps}=1,
		\end{equation}
	then $d\leq\dcrit$.
\item 
Conversely, if
		\begin{equation}\label{eqLemma_plantedCluster3}
		\lim_{\eps\searrow0}\liminf_{n\ra\infty}\pr\brk{|\cC(\G,\vec\sigma)|^{1/n}\geq k(1-1/k)^{d/2}+\eps}=1,
		\end{equation}
then $\limsup_{n\ra\infty}\Erw[Z_k(G(n,d/n))^{1/n}]<k(1-1/k)^{d/2}$.
In particular, $d\geq\dcrit$.
\end{enumerate}
\end{proposition}

\noindent
The proof of \Prop~\ref{Lemma_plantedCluster} is given in \Sec~\ref{Sec_plantedCluster}.

\subsection{The second thread.}

Our next aim is to ``solve'' the fixed point problem for $\cF_{d,k}$ to an extent that 
gives the fixed point an explicit combinatorial interpretation.
This combinatorial interpretation is  in terms of a certain random tree process, associated with a concept of ``legal colorings''.
Specifically, we consider a multi-type Galton-Watson branching process.
Its set of types is
	$$\cT=\cbc{(i,\ell):i\in\brk k,\, \ell\subset\brk k,\, i\in\ell}.$$
The intuition is that $i$ is a ``distinguished color'' and that $\ell$ is a set of ``available colors''.
The branching process is further parameterized by a vector $\vec q=(q_1,\ldots,q_k)\in[0,1]^k$ such
that $q_1+\cdots+q_k\leq 1$.
Let $d'=dk/(k-1)$ and
	$$q_{i,\ell}=\frac1k\prod_{j\in\ell\setminus\cbc i}\exp(-q_jd')\prod_{j\in\brk k\setminus\ell}
		1-\exp(-q_jd')\qquad\mbox{for }(i,\ell)\in\cT.$$
Then
	$$\sum_{(i,\ell)\in\cT}q_{i,\ell}=1.$$
Further, for each $(i,\ell)\in\cT$ such that $|\ell|>1$ we define $\cT_{i,\ell}$ as the set of all $(i',\ell')\in\cT$ such that
$\ell\cap\ell'\neq\emptyset$ and $|\ell'|>1$.
In addition, for $(i,\ell)\in\cT$ such that $|\ell|=1$ we set $\cT_{i,\ell}=\emptyset$.

The branching process $\GW(d,k,\vec q)$ starts with a single individual, whose type $(i,\ell)\in\cT$ is chosen from
the probability distribution $(q_{i,\ell})_{(i,\ell)\in\cT}$.
In the course of the process, each individual of type $(i,\ell)\in\cT$ spawns a Poisson number $\Po(d'q_{i',\ell'})$
of offspring of type $(i',\ell')$ for each $(i',\ell')\in\cT_{i,\ell}$.
In particular, only the initial individual may have a type $(i,\ell)$ with $|\ell|=1$, in which case it does not have any offspring.
Let $1\leq \cN\leq\infty$ be the progeny of the process (i.e., the total number of individuals created).

We  are going to view $\GW(d,k,\vec q)$ as a distribution over trees endowed with some extra information.
Let us define a {\em decorated graph} as a graph $T=(V,E)$ together with 
a map $\vartheta:V\ra\cT$ such that for each edge $e=\cbc{v,w}\in E$ we have $\vartheta(w)\in\cT_{\vartheta(v)}$.
Moreover, a {\em rooted decorated graph} is a decorated graph $(T,\vartheta)$ together with a distinguished vertex $v_0$, the {\em root}.
Further, an {\em isomorphism} between two rooted decorated graphs $T$ and $T'$ 
is an isomorphism of the underlying graphs that preserves the root and the  types of the vertices. 

Given that $\cN<\infty$,
 the branching process $\GW(d,k,\vec q)$ canonically induces a probability distribution over
isomorphism classes of rooted decorated trees.
Indeed, we obtain a tree whose vertices are all the individuals created in the course of the branching process and
	where there is an edge
	between each individual and its offspring.
The individual from which the process starts is the root.
Moreover, by construction each individual $v$ comes with a type $\thet(v)$.
We denote the (random) isomorphism class of this tree by $\T_{d,k,\vec q}$.
(It is natural to view the branching process as a probability distribution over {\em isomorphism classes}
as the process does not specify the order in which offspring is created.)

To proceed, we define a {\em legal coloring} of a decorated graph $(G,\vartheta)$ as a map $\tau:V(G)\ra\brk k$
such that $\tau$ is a $k$-coloring of $G$ and such that for any type $(i,\ell)\in\cT$ and
for any vertex $v$ with $\vartheta(v)=(i,\ell)$ we have $\tau(v)\in \ell$.
Let $\cZ(G,\vartheta)$ denote the number of legal colorings.

Since $\cZ(G,\vartheta)$ is isomorphism-invariant, we obtain the integer-valued random variable $\cZ(\T_{d,k,\vec q})$.
We have $\cZ(\T_{d,k,\vec q})\geq1$ with certainty
because a legal coloring $\tau$ can be constructed by coloring each vertex with its distinguished color (i.e., setting $\tau(v)=i$ if $v$ has type $(i,\ell)$).
Hence, $\ln\cZ(\T_{d,k,\vec q})$ is a well-defined non-negative random variable.
Additionally, we write $|\T_{d,k,\vec q}|$ for the number of vertices in $\T_{d,k,\vec q}$.

Finally, consider a rooted, decorated tree $(T,\thet,v_0)$ and let
 $\vec\tau$ be a legal coloring of $(T,\thet,v_0)$ chosen uniformly at random.
Then the color $\vec\tau(v_0)$ of the root is a random variable with values in $\brk k$.
Let $\mu_{T,\thet,v_0}\in\Omega$ denote its distribution.
Clearly, $\mu_{T,\thet,v_0}$ is invariant under isomorphisms.
Consequently, the distribution $\mu_{\T_{d,k,\vec q}}$ of the color of the root of a tree in the random isomorphism class 
$\T_{d,k,\vec q}$ is a well-defined $\Omega$-valued random variable.
Let $\pi_{d,k,\vec q}\in\cP$ denote its distribution.
Then we can characterise the frozen fixed point of $\cF_{d,k}$ as follows.

\begin{proposition}\label{Prop_fix}
Suppose that $d\geq(2k-1)\ln k-2$.
\begin{enumerate}
\item  The function
		\begin{equation}\label{eqSimpleFix}
		q\in[0,1]\mapsto(1-\exp(-dq/(k-1)))^{k-1}
		\end{equation}
	has a unique fixed point $q^*$ in the interval $[2/3,1]$.
	Moreover, with
		\begin{equation}\label{eqq*}
		\vec q^*=k^{-1}(q^*,\ldots,q^*)\in[0,1]^k
		\end{equation}
	the branching process $\GW(d,k,\vec q^*)$ is sub-critical.
	Thus, $\pr[\cN<\infty]=1$.
\item 	 The map $\cF_{d,k}$ has precisely one frozen fixed point, namely $\pi_{d,k,\vec q^*}$.
\item We have
	$\FF_{d,k}(\pi_{d,k,\vec q^*})=\Erw\brk{\frac{\ln\cZ(\T_{d,k,\vec q^*})}{|\T_{d,k,\vec q^*}|}}.$
\item The function $\Sigma_k$ from~(\ref{eqThm_condPhi}) 
		is strictly decreasing and continuous on $[(2k-1)\ln k-2,(2k-1)\ln k-1]$ and has
		a unique zero $\dc$ in this interval.
\end{enumerate}
\end{proposition}

\noindent
The function~(\ref{eqSimpleFix}) and its fixed point also occur in the physics work~\cite{LenkaFlorent}.
The proof of \Prop~\ref{Prop_fix} can be found in \Sec~\ref{Sec_fix}.

\subsection{Tying up the threads}
To prove that $\dc=\dcrit$, we establish a connection between
the random tree $\T_{d,k,\vec q^*}$ and the random graph $\G$ with planted coloring $\SIGMA$.
We start by giving a recipe for computing the cluster size $|\cC(\G,\SIGMA)|$, and then show that the random tree process ``cooks" it.

Computing the cluster size hinges on a close understanding of its combinatorial structure.
As hypothesised in physics work~\cite{MM} and established rigorously in~\cite{Barriers,Covers,Molloy},
typically many vertices $v$ are ``frozen'' in $\cC(\G,\SIGMA)$, i.e., $\tau(v)=\tau'(v)$
for any two colorings $\tau,\tau' \in \cC(\G,\SIGMA)$.
More generally, we consider for each vertex $v$ the set 
	$$\ell(v)=\cbc{\tau(v):\tau\in\cC(\G,\vec\sigma)}$$
of colors that $v$ may take in colorings $\tau$ that belong to the cluster.
Together with the ``planted'' color $\SIGMA(v)$, we can thus assign each vertex $v$ a type $\thet(v)=(\SIGMA(v),\ell(v))$.
This turns $\G$ into a decorated graph $(\G,\thet)$.

By construction, each coloring $\tau\in \cC(\G,\SIGMA)$ is a legal coloring of the decorated graph $\G$.
Conversely, we will see that \whp\ any legal coloring of $(\G,\thet)$ belongs to the cluster $\cC(\G,\SIGMA)$.
Hence, computing the cluster size $|\cC(\G,\SIGMA)|$ amounts to calculating the number $\cZ(\G,\thet)$ of legal colorings of $\G,\thet$.

This calculation is facilitated by the following observation.
Let $\widetilde\G$ be the graph obtained from $\G$ by deleting
all edges $e=\cbc{v,w}$ that join two vertices such that $\ell(v)\cap\ell(w)=\emptyset$.
Then any legal coloring $\tau$ of $\widetilde\G$ is a legal coloring of $\G$,
	because $\tau(v)\in\ell(v)$ for any vertex $v$.
Hence, $\cZ(\G,\thet)=\cZ(\widetilde\G,\thet)$.

Thus, we just need to compute $\cZ(\widetilde\G,\thet)$.
This task is much easier than computing $\cZ(\G,\thet)$ directly
because $\widetilde\G$ turns out to have {\em significantly} fewer edges than $\G$ \whp\
More precisely, \whp\ $\widetilde\G$ (mostly) consists of connected components that are trees of bounded size.
In fact, in a certain sense the distribution of the tree components
converges to that of the decorated random tree $\T_{d,k,\vec q_*}$.
In effect, we obtain

\begin{proposition}\label{Prop_cluster}
Suppose that $d\geq(2k-1)\ln k-2$ and let $p'$ be as in~(\ref{eqLemma_plantedCluster1}).
Let $\vec q^*$ be as in~(\ref{eqq*}).
Then the sequence 
	$\{\frac1n\ln|\cC(\G,\SIGMA)|\}_n$ 
converges  to $\Erw \left[\frac{\ln\cZ(\T_{d,k,\vec q^*})}{|\T_{d,k,\vec q^*}|}\right]$ in probability.
\end{proposition}

The proof of \Prop~\ref{Prop_cluster}, which can be found in \Sec~\ref{Sec_cluster}, is 
based on the precise analysis of a further message-passing algorithm called {\em Warning Propagation}.
Combining Propositions \ref{Lemma_plantedCluster} and \ref{Prop_cluster}, we see that $\dcrit$ is equal to $\dc$ given by Proposition \ref{Prop_fix}. Theorem~\ref{Thm_cond} then follows from Proposition~\ref{Lemma_dcrit}.

\section{Groundwork: the first and the second moment method
	}\label{Sec_dcrit}

\noindent
In this section we prove \Prop~\ref{Lemma_dcrit} and also lay the foundations for the proof of \Prop~\ref{Lemma_plantedCluster}.
Throughout this section, we always set $m=\lceil dn/2\rceil$ and we let $\gnm$ denote a random graph with vertex
set $V=[n]=\cbc{1,\ldots,n}$ and with precisely $m$ edges chosen uniformly at random.

\subsection{The first moment upper bound}
We start by deriving an upper bound on $\Phi_k(d)$ 
by computing the expected number of $k$-colorings.
To avoid fluctuations of the total number of edges, we work with the $\gnm$ model.

\begin{lemma}\label{Lemma_firstMoment}
We have $\Erw[\Zkc(\gnm)]=\Theta(k^n(1-1/k)^{m})$.
\end{lemma}

\Lem~\ref{Lemma_firstMoment} is folklore.
We carry the proof out regardless to make a few observations that will be important later.
For a map $\sigma:\brk n\ra\brk k$ let
	\begin{equation}\label{eqForb}
	\Forb(\sigma)=\sum_{i=1}^k\bink{|\sigma^{-1}(i)|}{2}
	\end{equation}
be the number of ``forbidden pairs'' of vertices that are colored the same under $\sigma$.
By convexity,
	\begin{equation}\label{eqForb2}
	 \Forb(\sigma)\geq (1-1/k)N,\qquad\mbox{with $N=\bink n2$.}
	 \end{equation}
Hence,  using Stirling's formula, we find
	\begin{equation}\label{eqFirstMoment1}
	\pr\brk{\sigma\mbox{ is a $k$-coloring of }G(n,m)}=\bink{N-\Forb(\sigma)}{m}/\bink Nm\leq O((1-1/k)^m).
	\end{equation}
As there are $k^n$ possible maps $\sigma$ in total, the linearity of expectation and~(\ref{eqFirstMoment1}) imply
	$$\Erw[\Zkc(\gnm)]=O(k^n(1-1/k)^{m}).$$

To bound $\Erw[\Zkc(\gnm)]$ from below,  call $\sigma:\brk n\ra\brk k$ {\em balanced} if $|\sigma^{-1}(i)-\frac nk|\leq\sqrt n$ for all $i\in\brk k$.
Let $\Bal=\Bal_{n,k}$ be the set of all balanced $\sigma:\brk n\ra\brk k$.
For $\sigma\in\Bal$ we verify easily that $\cF(\sigma)=(1-1/k)N+O(n)$.
Thus, (\ref{eqFirstMoment1}) and Stirling's formula yield
	\begin{equation}\label{eqFirstMoment3}
	\pr\brk{\sigma\mbox{ is a $k$-coloring of }G(n,m)}=\Omega((1-1/k)^m)\qquad\mbox{ for any }\sigma\in\Bal.
	\end{equation}
As $\abs\Bal=\Omega(k^n)$ by Stirling, the linearity of expectation and~(\ref{eqFirstMoment3}) imply
	$\Erw[\Zkc(\gnm)]=\Omega(k^n(1-1/k)^{m})$, whence \Lem~\ref{Lemma_firstMoment} follows.

Letting $\Zkb$ denote the number of balanced $k$-colorings, we obtain from the above argument

\begin{corollary}\label{Cor_balanced}
For any $d\geq0$ we have $\Erw[\Zkb(G(n,m))]=\Theta(k^n(1-1/k)^{m})$.
\end{corollary}

\noindent
As a further consequence of \Lem~\ref{Lemma_firstMoment}, we obtain

\begin{corollary}\label{Cor_firstMoment}
For any $c>0$ we have $\limsup_{n\ra\infty}\Erw[Z_k(G(n,c/n))^{1/n}]\leq k(1-1/k)^{c/2}$.
\end{corollary}
\begin{proof}
\Lem~\ref{Lemma_firstMoment} and Jensen's inequality yield
	\begin{equation}\label{eqFirstMoment4}
	\Erw[Z_k(G(n,m))^{1/n}]\leq \Erw[Z_k(G(n,m))]^{1/n}\leq k(1-1/k)^{d/2}+o(1).
	\end{equation}
Now, let $c>0$ and set $d=c-\eps$ for some $\eps>0$.
The number of edges in $G(n,c/n)$ is binomially distributed with mean $(1+o(1))cn/2=m+\Omega(n)$.
Hence, by the Chernoff bound the probability of the event $\cA$ that $G(n,c/n)$ has at least $m$ edges tends to $1$ as $n\ra\infty$.
Because adding further edges can only decrease the number of $k$-colorings and since the number of $k$-colorings is trivially bounded by $k^n$, we obtain
from~(\ref{eqFirstMoment4}) that
	\begin{eqnarray*}
	\Erw[Z_k(G(n,c/n))^{1/n}]&\leq&
		\Erw[Z_k(G(n,c/n))^{1/n}\cdot\vecone_\cA]+\pr\brk{\cA\mbox{ does not occur}}\cdot k\\
		&\leq&\Erw[Z_k(G(n,m))^{1/n}]+o(1)\leq k(1-1/k)^{d/2}+o(1).
	\end{eqnarray*}
Consequently, 
$\limsup\Erw[Z_k(G(n,c/n))^{1/n}]\leq k(1-1/k)^{d/2}$.
This holds for any $d>c$.
Hence, letting $\eps=d-c\ra0$, we see that
$\limsup\Erw[Z_k(G(n,c/n))^{1/n}]\leq k(1-1/k)^{c/2}$, as desired.
\end{proof}

\subsection{The second moment lower bound.}
The main technical step in the article~\cite{Danny} that yields the lower bound~(\ref{eqThrLoc})
on $\dk$ is a second moment argument for a random variable $\Zkg$ related to the number of $k$-colorings.
We are going employ this second moment estimate to bound $\Zkc(G(n,d/n))$ from below.

The random variable $\Zkg$ counts $k$-colorings with some additional properties.
Suppose that $\sigma$ is a balanced $k$-coloring of a graph $G$ on $V=\brk n$.
We call $\sigma$ {\em separable} if for any balanced $\tau\in\cC(G,\sigma)$ and any $i\in\brk k$ we have
	$$\rho_{ii}(\sigma,\tau)\geq(1-\kappa)/k,\mbox{ where }\kappa=\ln^{20}k/k.$$
Thus, if $\sigma$ is a balanced, separable $k$-coloring, then for any color $i$ and for any other balanced $k$-coloring $\tau$ in the cluster
of $\sigma$, a $1-\kappa+o(1)$-fraction of the vertices colored $i$ under $\sigma$ are colored $i$ under $\tau$ as well.
In particular, the clusters of any two such colorings 
are either disjoint or identical.

\begin{definition}
Let $G$ be a graph with $n$ vertices and $m$ edges.
A $k$-coloring $\sigma$ of $G$ is {\em tame} if 
\begin{description}
\item[T1] $\sigma$ is balanced,
\item[T2] $\sigma$ is separable, and
\item[T3] $|\cC(G,\sigma)\cap\Bal|\leq k^n(1-1/k)^m$.
\end{description}
\end{definition}

\noindent
Let $\Zkg(G)$ denote the number of tame $k$-colorings of $G$.

\begin{lemma}[\cite{Danny}]\label{Lemma_smm}
Assume that $d>0$ is such that 
	\begin{equation}\label{eqDanny}
	\liminf_{n\ra\infty}\frac{\Erw[\Zkg(G(n,m))]}{k^n(1-1/k)^m}>0.
	\end{equation}
Then $$\liminf_{n\ra\infty}\frac{\Erw[\Zkg(G(n,m))]^2}{\Erw[\Zkg(G(n,m))^2]}>0.$$
Furthermore, there exists $\eps_k=o_k(1)$ such that (\ref{eqDanny}) is satisfied if $d\leq (2k-1)\ln k-2\ln 2-\eps_k$.
\end{lemma}

As fleshed out in~\cite{Danny}, together with the sharp threshold result from~\cite{AchFried}, \Lem~\ref{Lemma_smm} implies 
that $G(n,d/n)$ is $k$-colorable \whp\ if $d\leq (2k-1)\ln k-2\ln 2-\eps_k$.
Here we are going to combine \Lem~\ref{Lemma_smm} with the following variant of that sharp threshold result to obtain a lower bound
on the {\em number} of $k$-colorings.

\begin{lemma}[\cite{Barriers}]\label{Lemma_Xfriedgut}
For any $k\geq3$ and for any real $\xi>0$ there is a sequence $d_{k,\xi}(n)$ such that for any $\eps>0$ the following holds.
\begin{enumerate}
\item If $p(n)<(1-\eps)d_{k,\xi}(n)/n$, then $\Zkc(G(n,p(n)))\geq\xi^n$ \whp
\item If $p(n)>(1+\eps)d_{k,\xi}(n)/n$, then $\Zkc(G(n,p(n)))<\xi^n$ \whp 
\end{enumerate}
\end{lemma}

\noindent
\Lem s~\ref{Lemma_smm} and~\ref{Lemma_Xfriedgut} entail the following lower bound on $\dcrit$.

\begin{lemma}\label{Lemma_Xsmm}
Assume that $d^*>0$ and $\eps>0$ are such that~(\ref{eqDanny}) holds for any $d\in(d^*-\eps,d^*)$. 
Then $\dcrit\geq d^*$.
\end{lemma}
\begin{proof}
Assume for contradiction that $d^*$ is such that~(\ref{eqDanny}) holds for all $d\in(d^*-\eps,d^*)$ but $\dcrit<d^*$.
Pick and fix a number $\max\{d^*-\eps,\dcrit\}<d_*<d^*$.
\Cor~\ref{Cor_firstMoment} implies that $\limsup\Erw[\Zkc(G(n,d_*/n))^{1/n}]\leq k(1-1/k)^{d_*/2}$.
Therefore, since $d_*>\dcrit$, there exists $\eps_*>0$ such that
	\begin{equation}\label{eqXsmm_1}
	\liminf_{n\ra\infty}\Erw[\Zkc(G(n,d_*/n))^{1/n}]<k(1-1/k)^{d_*/2}-\eps_*.
	\end{equation}
Further, pick and fix $d_*<\hat d<d^*$ such that
	$k(1-1/k)^{\hat d/2}>k(1-1/k)^{d_*/2}-\eps_*$
and $\xi$ such that 
		\begin{equation}\label{eqXsmm_2}
		k(1-1/k)^{d_*/2}-\eps_*<\xi<k(1-1/k)^{\hat d/2}.
		\end{equation}

We are going to use \Lem s~\ref{Lemma_smm} and~\ref{Lemma_Xfriedgut} to establish a lower bound 
on $\Zkc(G(n,d_*/n))$ that contradicts~(\ref{eqXsmm_1}).
By the Paley-Zygmund inequality and because~(\ref{eqDanny}) holds for any $d^*-\eps<d<d^*$,
	\begin{equation}\label{eqCor_Xfriedgut0}
	\pr\brk{\Zkg(\gnm)\geq\frac12\Erw[\Zkg(\gnm)]}\geq\frac{\Erw[\Zkg(\gnm)]^2}{4\cdot\Erw[\Zkg(\gnm)^2]}
		\quad\mbox{ for any $d^*-\eps<d<d^*$}.
	\end{equation}
Moreover, \Lem~\ref{Lemma_smm} and~(\ref{eqCor_Xfriedgut0}) imply
	\begin{equation}\label{eqCor_Xfriedgut1}
	\liminf_{n\ra\infty}\pr\brk{\Zkg(\gnm)\geq\frac12\Erw[\Zkg(\gnm)}>0\qquad\mbox{for any }d^*-\eps<d<d^*.
	\end{equation}
Further, because (\ref{eqDanny}) is true for any $d^*-\eps<d<d^*$ and $\xi<k(1-1/k)^{d/2}$ for any $d<\hat d<d^*$, we see that
	$$\frac12\Erw[\Zkg(\gnm)]=\Omega(k^n(1-1/k)^m)>\xi^n\qquad\mbox{for any }d<\hat d.$$
Hence, (\ref{eqCor_Xfriedgut1}) implies
	\begin{equation}\label{eqCor_Xfriedgut2}
	\liminf_{n\ra\infty}\pr\brk{\Zkg(\gnm)\geq\xi^n}>0\qquad\mbox{for any }d<\hat d.
	\end{equation}
Since the number of edges in $G(n,d/n)$ has a binomial distribution with mean $m$,
with probability at least $1/3$ the number of edges in $G(n,d/n)$ does not exceed $m$.
Therefore, (\ref{eqCor_Xfriedgut2}) implies that
	\begin{equation}\label{eqCor_Xfriedgut3}
	\liminf_{n\ra\infty}\pr\brk{\Zkc(G(n,d/n))\geq\xi^n}
		\geq\frac13\liminf_{n\ra\infty}\pr\brk{\Zkg(\gnm)\geq\xi^n}>0\quad\mbox{ for any }d<\hat d.
	\end{equation}
Moreover, (\ref{eqCor_Xfriedgut3}) entails that the sequence $d_{k,\xi}(n)$ from \Lem~\ref{Lemma_Xfriedgut} satisfies $\liminf d_{k,\xi}(n)\geq \hat d$.
Therefore, 
	\begin{equation}\label{eqCor_Xfriedgut4}
	\lim_{n\ra\infty}\pr\brk{\Zkc(G(n,d/n))\geq\xi^n}=1\quad\mbox{ for any }d<\hat d. 
	\end{equation}
Since $d_*<\hat d$, (\ref{eqCor_Xfriedgut4}) entails that
	\begin{equation}\label{eqCor_Xfriedgut5}
	\liminf_{n\ra\infty}\Erw\brk{\Zkg(G(n,d_*/n))^{1/n}}\geq \xi.
	\end{equation}
Combining~(\ref{eqXsmm_1}), (\ref{eqXsmm_2}) and~(\ref{eqCor_Xfriedgut5}) yields a contradiction, which
refutes our assumption that $\dcrit<d^*$.
\end{proof}

\subsection{Proof of \Prop~\ref{Lemma_dcrit}}
We start with the following observation.

\begin{lemma}\label{Lemma_observation}
Let 
	\begin{eqnarray*}
	D_*&=&\cbc{d>0:\liminf\Erw[\Zkc(G(n,d/n))^{1/n}]<k(1-1/k)^{d/2}},\\
	D^*&=&\cbc{d>0:\limsup\Erw[\Zkc(G(n,d/n))^{1/n}]<k(1-1/k)^{d/2}}.
	\end{eqnarray*}
If $d_1\in D_*$ and $d_2>d_1$, then $d_2\in D_*$.
Similarly, if $d_1\in D^*$ and $d_2>d_1$, then $d_2\in D^*$.
\end{lemma}
\begin{proof}
Let $0<d_1<d_2$ and let $q\sim (d_2-d_1)/n$ be such that $d_1/n+(1-d_1/n)q=d_2/n$.
Let us denote the random graph $G(n,d_1/n)$ by $G_1$.
Furthermore, let $G_2$ be a random graph obtained from $G_1$ by joining any two vertices
that are not already adjacent in $G_1$ with probability $q$ independently.
Then $G_2$ is identical to $G(n,d_2/n)$, because in $G_2$ any two vertices are adjacent with probability
$d_1/n+(1-d_1/n)q=d_2/n$ independently.
Set $N=\bink n2$.

Let $e(G_i)$ signify the number of edges in $G_i$ for $i=1,2$.
Because $e(G_i)$ is a binomial random variable with mean $\mu_i=\frac{d_i}n\cdot N=nd_i/2+O(1)$,
the Chernoff bound implies that
	\begin{equation}\label{eqLemma_observation0}
	\pr\brk{|e(G_1)-\mu_1|>n^{2/3}}=o(1),\qquad\pr\brk{|e(G_2)-e(G_1)-(\mu_2-\mu_1)|>n^{2/3}}=o(1).
	\end{equation}
Further, since $\Zkc^{1/n}\leq k$ with certainty,
(\ref{eqLemma_observation0}) implies that
	\begin{eqnarray}\nonumber
	\Erw[\Zkc(G_2)^{1/n}\,|\,\Zkc(G_1)]&\leq&\Erw[\Zkc(G_2)^{1/n}\,|\,\Zkc(G_1),\,|e(G_2)-e(G_1)-(\mu_2-\mu_1)|\leq n^{2/3}]\\
		&&\qquad+k\cdot
			\pr\brk{|e(G_2)-e(G_1)-(\mu_2-\mu_1)|\leq n^{2/3}}\nonumber\\
		&\leq&\Erw[\Zkc(G_2)^{1/n}\cdot\vecone_{|e(G_2)-e(G_1)-(\mu_2-\mu_1)|\leq n^{2/3}}\,|\,\Zkc(G_1)]+o(1).
			\label{eqLemma_observation101}
	\end{eqnarray}

Suppose that we condition on $e(G_1),e(G_2)$ and $|e(G_1)-\mu_1|\leq n^{2/3}$, $|e(G_2)-e(G_1)-(\mu_2-\mu_1)|\leq n^{2/3}$.
Assume that $\sigma$ is a $k$-coloring of $G_1$.
What is the probability that $\sigma$ remains a $k$-coloring of $G_2$?
For this to happen, none of the $e(G_2)-e(G_1)$ additional edges must be among the $\Forb(\sigma)$
pairs of vertices with the same color under $\sigma$.
Using Stirling's formula, we see that the probability of $\sigma$ remaining a $k$-coloring in $G_2$ is bounded by
	\begin{equation}\label{eqLemma_observation100}
	\gamma=\bink{N-\Forb(\sigma)-e(G_1)}{e(G_2)-e(G_1)}/\bink{N-e(G_1)}{e(G_2)-e(G_1)}\leq(1-1/k)^{(d_2-d_1+o(1)) n/2}.
	\end{equation}	
Hence, by~(\ref{eqLemma_observation101}), Jensen's inequality and~(\ref{eqLemma_observation100})
	\begin{eqnarray}
	\Erw[\Zkc(G_2)^{1/n}\,|\,\Zkc(G_1)]
		&\leq&\Erw\brk{\Zkc(G_2)\cdot\vecone_{|e(G_2)-e(G_1)-(\mu_2-\mu_1)|\leq n^{2/3}}\,\big|\,\Zkc(G_1)}^{1/n}+o(1)\nonumber\\
		&\leq& \gamma^{1/n}\Zkc(G_1)^{1/n}+o(1)
			\leq(1-1/k)^{(d_2-d_1)/2}\Zkc(G_1)^{1/n}+o(1).
				\label{eqLemma_observation111}
	\end{eqnarray}
Averaging~(\ref{eqLemma_observation111}) over $G_1$, we obtain
	\begin{eqnarray*}
	\Erw[\Zkc(G(n,d_2/n)^{1/n}]&=&\Erw[\Zkc(G_2)^{1/n}]\\
		&\hspace{-4cm}\leq&\hspace{-2cm}(1-1/k)^{(d_2-d_1)/2}\Erw[\Zkc(G_1)^{1/n}\cdot\vecone_{|e(G_1)-\mu_1|\leq n^{2/3}}]+
			k\cdot\pr\brk{\vecone_{|e(G_1)-\mu_1|> n^{2/3}}}+o(1)\\
		&\hspace{-4cm}\leq&\hspace{-2cm}(1-1/k)^{(d_2-d_1)/2}\Erw[\Zkc(G(n,d_1/n))^{1/n}]+o(1)\qquad\qquad\qquad[\mbox{due to~(\ref{eqLemma_observation0})}].
	\end{eqnarray*}
Thus, if $\Erw[\Zkc(G(n,d_1/n))^{1/n}]<k(1-1/k)^{d_1/2}-\delta+o(1)$, then $\Erw[\Zkc(G(n,d_2/n)^{1/n}]\leq k(1-1/k)^{d_2/2}-\eps+o(1)$ for some $\eps=\eps(\delta,k,d_1,d_2)>0$.
Taking $n\ra\infty$ yields the assertion.
\end{proof}

\begin{proof}[Proof of \Prop~\ref{Lemma_dcrit}]
\Cor~\ref{Cor_firstMoment} implies that
	$$\dcrit=\sup\cbc{d\geq0:\liminf_{n\ra\infty}\Erw[Z_k(G(n,d/n))^{1/n}]\geq k(1-1/k)^{d/2}}.$$
Hence, the first and the third assertion are immediate from 
\Lem~\ref{Lemma_observation}.

Further, (\ref{eqTreeHug2}) implies that $\dcrit>0$.
Assume for contradiction that $\dcrit$ is smooth.
Then there is $\eps>0$ such that the limit $\Phi_k(d)$ exists for all $d\in(\dcrit-\eps,\dcrit+\eps)$ and such that the function $d\mapsto\Phi_k(d)$ is given by an
absolutely convergent power series on this interval.
Moreover, the first assertion implies that $\Phi_k(d)=k(1-1/k)^{d/2}$ for all $d\in(\dcrit-\eps,\dcrit)$.
Consequently, the uniqueness of analytic continuations implies that $\Phi_k(d)=k(1-1/k)^{d/2}$ for all $d\in(\dcrit-\eps,\dcrit+\eps)$, in contradiction to
the definition of $\dcrit$.
Thus, $\dcrit$ is a phase transition.
\end{proof}

\section{The planted model}\label{Sec_plantedCluster}

\subsection{Overview.}
The aim in this section is to prove \Prop~\ref{Lemma_plantedCluster}.
The proof of the first part is fairly straightforward.
More precisely, in \Sec~\ref{Sec_plantedCluster_easy} we are going to establish

\begin{lemma}\label{Lemma_plantedCluster_easy}
Assume that $(2k-1)\ln k-2\leq d\leq (2k-1)\ln k$ is such that~(\ref{eqLemma_plantedCluster2}) holds.
Then $\dcrit\geq d$.
\end{lemma}

The more challenging claim is that $d\geq\dcrit$ if 
typically the  cluster in the planted model is ``too big''.
To prove this, we consider a variant of the planted model in which the number of edges is fixed.
More precisely, for a map $\sigma:\brk n\ra\brk k$ we let $G(n,m,\sigma)$ denote a graph on the vertex set $V=\brk n$
with precisely $m$ edges that do not join vertices $v,w$ with $\sigma(v)=\sigma(w)$ chosen uniformly at random.
In other words, $G(n,m,\sigma)$ is just the random graph $G(n,m)$  conditioned on the event that $\sigma$ is a $k$-coloring.
The following lemma, which is a variant of the ``planting trick'' from~\cite{Barriers}, establishes a general relationship between $G(n,m)$ and $G(n,m,\sigma)$.

\begin{lemma}\label{Lemma_antiPlanting}
Let $d>0$.
Assume that there exists a sequence $(\cE_n)_{n\geq1}$ of events such that
	\begin{equation}\label{eqantiPlanting}
	\lim_{n\ra\infty}\pr\brk{G(n,m)\in\cE_n}=1\quad\mbox{while}\quad
		\limsup_{n\ra\infty}\pr\brk{G(n,m,\vec\sigma)\in\cE_n}^{1/n}<1.
	\end{equation}
Then for any $c>d$ we have $\limsup\Erw[\Zkc(G(n,c/n))^{1/n}]<k(1-1/k)^{c/2}$. In particular, $\dcrit\leq d$.
\end{lemma}

We prove \Lem~\ref{Lemma_antiPlanting} in \Sec~\ref{Sec_antiPlanting}.
Hence, assuming that 
the typical cluster size in the planted model is ``too big'' \whp,
we need to exhibit events $\cE_n$ such that~(\ref{eqantiPlanting}) holds.
An obvious choice seems to be
	$$\cE_n(\eps)=\cbc{\Zkc^{1/n}\leq k(1-1/k)^{d/2}+\eps}$$
But (\ref{eqantiPlanting}) requires that the probability that $\cE_n$ occurs in $G(n,m,\sigma)$ is {\em exponentially} small, and
neither the cluster size nor $\Zkc$ are known to be sufficiently concentrated to obtain such an exponentially small probability.

Therefore, we define the events $\cE_n$ by means of another random variable.
For a graph $G=(V,E)$ and a map $\sigma:V\ra\brk k$ let $\cH_G(\sigma)$ be the number of edges $\cbc{v,w}$ of $G$ such that $\sigma(v)=\sigma(w)$.
In words, $\cH_G(\sigma)$ is the number of edges of $G$ that are monochromatic under $\sigma$.
Furthermore, given $\beta>0$ let
	$$Z_{\beta,k}(G)=\sum_{\sigma:V\ra\brk k}\exp(-\beta\cdot\cH_G(\sigma)),$$
a quantity known as the {\em partition function of the $k$-spin Potts antiferromagnet on $G$ at inverse temperature $\beta$}.

For large $\beta$ there is a stiff ``penalty factor'' of $\exp(-\beta)$ for any monochromatic edge.
Thus, we expect that $Z_{\beta,k}$ becomes a good proxy for $\Zkc$  as $\beta\ra\infty$.
At the same time, $\ln Z_{\beta,k}$ enjoys a Lipschitz property.
Namely, suppose that we obtain a graph $G'$ from $G$ by either adding or removing a single edge.
Then
	\begin{equation}\label{eqLip}
	|\ln(Z_{\beta,k}(G))-\ln(Z_{\beta,k}(G'))|\leq\beta.
	\end{equation}
Due to this Lipschitz property, one can easily show that $\ln Z_{\beta,k}$ is tightly concentrated.
More precisely, we have

\begin{lemma}\label{Cor_ZAzuma}
For any fixed $d>0$, $\eps>0$ there is 
$\alpha>0$ such that the following is true.
Suppose that 
$(\sigma_n)_{n\geq1}$ is a sequence of maps $\brk n\ra\brk k$. 
Then for all large enough $n$,
	$$\pr\brk{|\ln(Z_{\beta,k}(G(n,p',\sigma_n)))-\Erw[\ln Z_{\beta,k}(G(n,p',\sigma_n))]|>\eps n}\leq \exp(-\alpha n).$$
\end{lemma}
\begin{proof}
This is immediate from the Lipschitz property~(\ref{eqLip}) and McDiarmid's inequality~\cite[\Thm~3.8]{McDiarmid}.
\end{proof}

Furthermore, in \Sec~\ref{Sec_partition} we show that \Lem~\ref{Cor_ZAzuma} implies

\begin{lemma}\label{Lemma_partition}
Assume that $d$ is such that~(\ref{eqLemma_plantedCluster3}) holds.
Then there exist $z,\beta>0$ such that
	$$\lim_{n\ra\infty}\pr\brk{\frac1n\ln Z_{\beta,k}(\gnm)\leq z}=1\quad\mbox{while}\quad
		\limsup_{n\ra\infty}\pr\brk{\frac1n\ln Z_{\beta,k}(G(n,m,\vec\sigma))\leq z}^{1/n}<1.$$
\end{lemma}

\noindent
Finally, \Prop~\ref{Lemma_plantedCluster} is immediate from \Lem s~\ref{Lemma_plantedCluster_easy}, \ref{Lemma_antiPlanting} and~\ref{Lemma_partition}.

\subsection{Proof of \Lem~\ref{Lemma_plantedCluster_easy}}\label{Sec_plantedCluster_easy}
We use the following observation from~\cite{Danny}.

\begin{lemma}[\cite{Danny}]\label{Lemma_separable}
Suppose that $(2k-1)\ln k-2\leq d\leq (2k-1)\ln k$.
Let $p'$ be as in~(\ref{eqLemma_plantedCluster1}).
Then the planted coloring $\vec\sigma$ is separable in $G(n,p',\vec\sigma)$ \whp
\end{lemma}

If~(\ref{eqLemma_plantedCluster2}) holds, then there exists $\eps>0$ such that with $p'$ from~(\ref{eqLemma_plantedCluster1}) we have
	\begin{equation}\label{eqLemma_plantedCluster_easy1}
	\lim_{n\ra\infty}\pr\brk{|\cC(G(n,p',\vec\sigma),\vec\sigma)|\leq k^n(1-1/k)^m\exp(-\eps n)}=1.
	\end{equation}
Pick a number $d^*>d$ such that with $m^*=\lceil d^*n/2\rceil$ we have
	$$k^n(1-1/k)^{m^*}\geq k^n(1-1/k)^m\exp(-\eps n/2).$$

We claim that if we choose $\vec\sigma:\brk n\ra\brk k$ uniformly at random and independently a random graph $G(n,m^*)$, then
	\begin{eqnarray}\label{eqLemma_plantedCluster_easy6}
	\liminf_{n\ra\infty}\pr\brk{\vec\sigma\mbox{ is tame}|\vec\sigma\mbox{ is a $k$-coloring of }G(n,m^*)}>0.
	\end{eqnarray}
To see this, let $\cE$ be the event that the random graph $G(n,p',\vec\sigma)$ has no more than $m^*$ edges.
Because the number of edges in $G(n,p',\vec\sigma)$ is binomially distributed with mean $m<m^*-\Omega(n)$,
the Chernoff bound implies that $\pr\brk\cE=1-o(1)$.
Therefore, (\ref{eqLemma_plantedCluster_easy1}) implies
	\begin{equation}\label{eqLemma_plantedCluster_easy2}
	\lim_{n\ra\infty}\pr\brk{|\cC(G(n,p',\vec\sigma),\vec\sigma)|\leq k^n(1-1/k)^m\exp(-\eps n)\,|\,\cE}=1.
	\end{equation}
Further, set $d''=kd^*/(k-1)$ and let $p''=d''/n>p'$.
Then we can think of $G(n,p'',\vec\sigma)$ as being obtained from $G(n,p',\vec\sigma)$ by adding further random edges.
More precisely, let $\cA$ be the event that $G(n,p'',\vec\sigma)$ contains precisely $m^*$ edges and set
	$$p_n'=\pr\brk{|\cC_{\vec\sigma}(G(n,p'',\vec\sigma))|\leq k^n(1-1/k)^{m^*}\,|\,\cA}.$$
Since adding edges can only decrease the cluster size, (\ref{eqLemma_plantedCluster_easy2}) entails
	\begin{equation}\label{eqLemma_plantedCluster_easy3}
	\lim_{n\ra\infty}p_n'\geq
		\lim_{n\ra\infty}\pr\brk{|\cC_{\vec\sigma}(G(n,p',\vec\sigma))|\leq k^n(1-1/k)^m\exp(-\eps n)\,|\,\cE}=1.
	\end{equation}
Similarly, let
	$p_n''=\pr\brk{\vec\sigma\mbox{ is separable in }G(n,p'',\vec\sigma)\,|\,\cA}.$
Then \Lem~\ref{Lemma_separable} implies
	\begin{eqnarray}\label{eqLemma_plantedCluster_easy4}
	\lim_{n\ra\infty}p_n''&\geq&\lim_{n\ra\infty}\pr\brk{\vec\sigma\mbox{ is separable in }G(n,p',\vec\sigma)\,|\,\cE}=1.
	\end{eqnarray}
Further, consider $p_n'''=\pr\brk{\vec\sigma\mbox{ is balanced}}$.
Then by Stirling's formula,
	\begin{eqnarray}\label{eqLemma_plantedCluster_easy5}
	\liminf_{n\ra\infty}p_n'''&>&0.
	\end{eqnarray}
Finally, let $p_n=\pr\brk{\vec\sigma\mbox{ is a tame $k$-coloring of }G(n,p'',\vec\sigma)|\cA}$.
Given the event $\cA$, $G(n,p'',\vec\sigma)$ is just a uniformly random graph with $m^*$ edges in which $\vec\sigma$ is a $k$-coloring.
Hence,
	$$p_n=\pr\brk{\vec\sigma\mbox{ is tame}|\vec\sigma\mbox{ is a $k$-coloring of }G(n,m^*)}.$$
As~(\ref{eqLemma_plantedCluster_easy3})--(\ref{eqLemma_plantedCluster_easy5}) yield $\liminf_{n\ra\infty}p_n>0$,
we obtain~(\ref{eqLemma_plantedCluster_easy6}) 

The estimate~(\ref{eqLemma_plantedCluster_easy6}) enables us to bound $\Erw[\Zkg(G(n,m^*)]$ from below.
Indeed, by the linearity of expectation
	\begin{eqnarray*}
	\Erw[\Zkg(G(n,m^*)]&=&\sum_{\sigma:\brk n\ra\brk k}\pr\brk{\sigma\mbox{ is a tame $k$-coloring of $G(n,m^*)$}}\\
		&=&k^n\cdot\pr\brk{\vec\sigma\mbox{ is a tame $k$-coloring of $G(n,m^*)$}}\\
		&=&k^n\pr\brk{\vec\sigma\mbox{ is a $k$-coloring of $G(n,m^*)$}}\cdot\pr\brk{\vec\sigma\mbox{ is tame}|\,\vec\sigma\mbox{ is a $k$-coloring of $G(n,m^*)$}}\\
		&=&k^n\pr\brk{\vec\sigma\mbox{ is a $k$-coloring of $G(n,m^*)$}}\cdot p_n=\Erw[\Zkc]\cdot p_n.
	\end{eqnarray*}
Thus, \Lem~\ref{Lemma_firstMoment} and (\ref{eqLemma_plantedCluster_easy6}) yield
	$$\liminf_{n\ra\infty}\frac{\Erw[\Zkg(G(n,m^*)]}{k^n(1-1/k)^m}>0.$$
As this holds for all $d^*$ in an interval $(d+\eta,d+2\eta)$ with $\eta>0$, \Lem~\ref{Lemma_Xsmm} implies that $\dcrit\geq d$.
\qed

\subsection{Proof of \Lem~\ref{Lemma_antiPlanting}}\label{Sec_antiPlanting}

\begin{lemma}\label{Cor_observation}
Assume that $d>0$ is such that
$\limsup\Erw[\Zkc(\gnm)^{1/n}]<k(1-1/k)^{d/2}$.
Then for any $c>d$ we have
	$\limsup\Erw[\Zkc(G(n,c/n))^{1/n}]<k(1-1/k)^{c/2}.$
\end{lemma}
\begin{proof}
Assume that $d,\delta>0$ are such that $\limsup\Erw[\Zkc(\gnm)^{1/n}]<k(1-1/k)^{d/2}-\delta$.
We claim that
	\begin{equation}\label{eqLemma_observation1}
	d^*\in D^*=\cbc{c>0:\limsup\Erw[\Zkc(G(n,c/n))^{1/n}]<k(1-1/k)^{c/2}}
	\qquad\mbox{ for any $d^*>d$.}
	\end{equation}
Indeed, the number $e(G(n,d^*/n))$ of edges of $G(n,d^*/n)$ is binomially distributed with mean $(1+o(1))d^*n/2$.
Since $d,d^*$ are independent of $n$ and $d^*>d$, the Chernoff bound implies that
	\begin{equation}\label{eqLemma_observation2}
	\pr\brk{e(G(n,d^*/n))\leq m}\leq\exp(-\Omega(n)).
	\end{equation}
Further, if we condition on the event that $m^*=e(G(n,d^*/n))>m$, then we can think of $G(n,d^*/n)$ as follows:
	first, create a random graph $\gnm$; then, add another $m^*-m$ random edges.
Since the addition of further random edges cannot increase the number of $k$-colorings, (\ref{eqLemma_observation2}) 
	\begin{eqnarray*}
	\Erw[\Zkc(G(n,d^*/n))^{1/n}]&\leq&\Erw[\Zkc(G(n,d^*/n))^{1/n}|m^*>m]+k\cdot\pr\brk{e(G(n,d^*/n))\leq m}\\
		&\leq&\Erw[\Zkc(G(n,m))^{1/n}]+o(1).
	\end{eqnarray*}
Taking $n\ra\infty$, and assuming that $d^*>d$ is sufficiently close to $d$, we conclude that
	$$\limsup_{n\ra\infty}\Erw[\Zkc(G(n,d^*/n))^{1/n}]\leq k(1-1/k)^{d/2}-\delta< k(1-1/k)^{d_*/2}.$$
Hence, for any $\eps>0$ there is $d^*\in(d,d+\eps)$ such that $d^*\in D^*$.
Thus, (\ref{eqLemma_observation1}) follows from \Lem~\ref{Lemma_observation}.
\end{proof}

\begin{proof}[Proof of \Lem~\ref{Lemma_antiPlanting}]
Assuming the existence of $d$ and $(\cE_n)_n$ as in \Lem~\ref{Lemma_antiPlanting}, we are going to argue that
	\begin{equation}\label{eqLemma_antiPlanting1}
	\limsup\Erw[\Zkc(\gnm)^{1/n}]<k(1-1/k)^{d/2}.
	\end{equation}
Then the assertion follows from \Lem~\ref{Cor_observation}.

Since $\Zkc^{1/n}\leq k$ with certainty and $\pr[\gnm\in\cE_n]=1-o(1)$, Jensen's inequality yields
	\begin{equation}\label{eqLemma_antiPlanting2}
	\Erw[\Zkc(\gnm)^{1/n}]=\Erw[\Zkc(\gnm)^{1/n}\cdot\vecone_{\cE_n}]+o(1)
		\leq\Erw[\Zkc(\gnm)\cdot\vecone_{\cE_n}]^{1/n}+o(1).
	\end{equation}
Furthermore, by the linearity of expectation,
	\begin{eqnarray}
	\Erw[\Zkc(\gnm)\cdot\vecone_{\cE_n}]&=&\sum_{\sigma:\brk n\ra\brk k}\pr\brk{\cE_n\mbox{ occurs and $\sigma$ is a $k$-coloring of }\gnm}\nonumber\\
		&=&\sum_{\sigma:\brk n\ra\brk k}\pr\brk{\cE_n|\mbox{$\sigma$ is a $k$-coloring of }\gnm}\cdot\pr\brk{\mbox{$\sigma$ is a $k$-coloring of }\gnm}\nonumber\\
		&=&\sum_{\sigma:\brk n\ra\brk k}\pr\brk{G(n,m,\sigma)\in\cE_n}\cdot\pr\brk{\mbox{$\sigma$ is a $k$-coloring of }\gnm}.
			\label{eqLemma_antiPlanting3}
	\end{eqnarray}
To estimate the last factor, 
we use (\ref{eqForb2}) and Stirling's formula, which yield
	$$\pr\brk{\mbox{$\sigma$ is a $k$-coloring of }\gnm}\leq\bink{\bink n2-\Forb(\sigma)}{m}/\bink{\bink n2}{m}\leq O((1-1/k)^m).$$
Plugging this estimate into~(\ref{eqLemma_antiPlanting3}) and recalling that $\vec\sigma$ is a random map $\brk n\ra\brk k$, we obtain
	\begin{eqnarray}
	\Erw[\Zkc(\gnm)\cdot\vecone_{\cE_n}]&\leq&O((1-1/k)^m)
			\sum_{\sigma:\brk n\ra\brk k}\pr\brk{G(n,m,\sigma)\in\cE_n}\nonumber\\
			&\hspace{-2cm}=&\hspace{-1cm}
				O((1-1/k)^m)\cdot k^n\pr\brk{G(n,m,\vec\sigma)\in\cE_n}
			\leq O(\Erw[\Zkc])\cdot\pr\brk{G(n,m,\vec\sigma)\in\cE_n}.\qquad
				\label{eqLemma_antiPlanting4}
	\end{eqnarray}
Finally, using our assumption that $\limsup\pr\brk{G(n,m,\vec\sigma)\in\cE_n}^{1/n}<1$ and combining~(\ref{eqLemma_antiPlanting3}) and~(\ref{eqLemma_antiPlanting4}),
we see that
	$$\limsup \Erw[\Zkc(\gnm)^{1/n}]\leq k(1-1/k)^{d/2}\cdot\limsup\pr\brk{G(n,m,\vec\sigma)\in\cE_n}^{1/n}<k(1-1/k)^{d/2},$$
thereby completing the proof of~(\ref{eqLemma_antiPlanting1}).
\end{proof}

\subsection{Proof of \Lem~\ref{Lemma_partition}}\label{Sec_partition}

\begin{lemma}\label{Lemma_upperBoundAnnealed}
Let $d>0$.
For any $\eps>0$ there exists $\beta>0$ such that
	$$\frac1n\ln\Erw[Z_{\beta,k}(\gnm)]\leq \ln k+\frac d2\ln(1-1/k)+\eps.$$
\end{lemma}
\begin{proof}
For any fixed number $\gamma>0$ we can choose $\beta(\gamma)>0$ so large that $\ln k-\beta\gamma<0$.
Now, let $\cM(\gnm)$ be the set of all $\sigma:\brk n\ra\brk k$ such that at least $\gamma n$ edges are monochromatic under $\sigma$,
and let $\overline\cM(\gnm)$ contain all $\sigma\not\in\cM(\gnm)$.
Then
	\begin{eqnarray}\nonumber
	Z_{\beta,k}(\gnm)&\leq&|\cM(\gnm)|\cdot \exp(-\beta\gamma n)+|\overline\cM(\gnm)|\\
		&\leq&k^n\cdot \exp(-\beta\gamma n)+|\overline\cM(\gnm)|\leq1+|\overline\cM(\gnm)|.
			\label{eqLemma_upperBoundAnnealed1}
	\end{eqnarray}
Further, if $\sigma\in\overline\cM(\gnm)$, then $\sigma$ is a $k$-coloring of a subgraph of $\gnm$ containing $m-\gamma n$ edges.
Hence, 
we obtain from Stirling's formula that
for $\gamma=\gamma(\eps)>0$ small enough,
	\begin{eqnarray*}
	\pr\brk{\sigma\in\overline\cM(\gnm)}&\leq&\bink{\bink n2}{\gamma n}\cdot\bink{\bink n2-\Forb(\sigma)}{m-\gamma n}/{\bink{\bink n2}{m}}
		\leq(1-1/k)^m\cdot\exp(\eps n/2).
	\end{eqnarray*}
Hence,
	\begin{equation}\label{eqLemma_upperBoundAnnealed2}
	\Erw[\overline\cM(\gnm)]\leq k^n(1-1/k)^m\cdot\exp(\eps n/2).
	\end{equation}
Combining~(\ref{eqLemma_upperBoundAnnealed1}) and~(\ref{eqLemma_upperBoundAnnealed2}), we obtain
	$$\Erw[Z_{\beta,k}(\gnm)]\leq 1+k^n(1-1/k)^m\cdot\exp(\eps n/2)<k^n(1-1/k)^m\cdot\exp(\eps n).$$
Taking logarithms completes the proof.
\end{proof}

\begin{lemma}\label{Lemma_pickAndChoose}
Assume that~(\ref{eqLemma_plantedCluster3}) is true.
Then there exist a fixed number $\eps>0$, a sequence $\sigma_n$ of balanced maps $\brk n\ra\brk k$ and a sequence $\mu_n$ of numbers satisfying
	$|\mu_n-dn/2|\leq\sqrt n$ such that
	$$\lim_{n\ra\infty}\pr\brk{|\cC(G(n,\sigma_n,\mu_n),\sigma_n)|^{1/n}>k(1-1/k)^{d/2}+\eps}=1.$$
\end{lemma}
\begin{proof}
Let $\cA$ be the event that the number of edges in the random graph $G(n,p',\vec\sigma)$ differs from $dn/2$ by at most $\sqrt n$.
Let $N=\bink n2$.
For any balanced $\sigma:\brk n\ra\brk k$ the expected number of edges in $G(n,p',\vec\sigma)$ is
	\begin{equation}\label{eqLemma_pickAndChoose1}
	(N-\Forb(\sigma))p'=(1-1/k)Np'+O(1)=dn/2+O(1).
	\end{equation}
Since the number of edges in $G(n,p',\vec\sigma)$ is a binomial random variable, (\ref{eqLemma_pickAndChoose1}) shows together with
the central limit theorem that there exists a fixed $\gamma>0$ such that for sufficiently large $n$
	\begin{equation}\label{eqLemma_pickAndChoose2}
	\pr\brk{G(n,p',\sigma)\in\cA}\geq\gamma\qquad\mbox{for all balanced $\sigma$}.
	\end{equation}
Furthermore, by Stirling's formula there is an $n$-independent number $\delta>0$ such that for sufficiently large $n$ we have
	\begin{equation}\label{eqLemma_pickAndChoose3}
	\pr\brk{\vec\sigma\in\Bal}\geq\delta.
	\end{equation}
Combining~(\ref{eqLemma_pickAndChoose2}) and~(\ref{eqLemma_pickAndChoose3}), we see that
	\begin{equation}\label{eqLemma_pickAndChoose4}
	\pr\brk{\vec\sigma\in\Bal,\,G(n,p',\vec\sigma)\in\cA}=
		\pr\brk{\vec\sigma\in\Bal}\cdot\pr\brk{G(n,p',\vec\sigma)\in\cA|\vec\sigma\in\Bal}\geq\gamma\delta>0.
	\end{equation}
Thus, pick $\sigma_n\in\Bal$ and $\mu_n\in[dn/2-\sqrt n,dn/2+\sqrt n]$ that maximize
	$$p(\sigma_n,\mu_n)=\pr\brk{|\cC(G(n,\sigma_n,\mu_n),\sigma_n)|^{1/n}>k(1-1/k)^{d/2}+\eps}.$$
Then~(\ref{eqLemma_plantedCluster3}) and~(\ref{eqLemma_pickAndChoose4}) imply that $\lim_{n\ra\infty}p(\sigma_n,\mu_n)=1$.
\end{proof}

\begin{lemma}\label{Lemma_nearlyBalanced}
For any $\eta>0$ there is $\delta>0$ such that
	$\lim_{n\ra\infty}\frac1n\ln\pr\brk{\sum_{i=1}^k|\vec\sigma^{-1}(i)-n/k|>\eta n}\leq-\delta.$
\end{lemma}
\begin{proof}
For each $i\in\brk k$ the number $|\vec\sigma^{-1}(i)|$ is a binomially distributed random variable with mean $n/k$.
Moreover, if $\sum_{i=1}^k|\vec\sigma^{-1}(i)-n/k|>\eta n$, then there is some $i\in\brk k$ such that $|\vec\sigma^{-1}(i)-n/k|>\eta n/k$.
Thus, the assertion is immediate from the Chernoff bound.
\end{proof}

Let $\Vol_G(S)$ be the sum of the degrees of the vertices in $S$ in the graph $G$.

\begin{lemma}\label{Lemma_Vol}
For any $\gamma>0$ there is $\alpha>0$ such that for any set $S\subset\brk n$
of size $|S|\leq\alpha n$ and any $\sigma:\brk n\ra\brk k$ we have
	$\limsup\frac1n\ln\pr\brk{\Vol_{G(n,p',\sigma)}(S)>\gamma n}\leq -\alpha.$
\end{lemma}
\begin{proof}
Let $(X_v)_{v\in\brk n}$ be a family of independent random variables with distribution $\Bin(n,p')$.
Then for any set $S$ the volume $\Vol(S)$ in $G(n,p',\sigma)$ is stochastically dominated by $X_S=2\sum_{v\in S}X_s$.
Indeed, for each vertex $v\in S$ the degree is a binomial random variable with mean at most $np'$, and
the only correlation amongst the degrees of the vertices in $S$ is that each edge joining two vertices in $S$ contributes two to $\Vol(S)$.
Furthermore, $\Erw[X_S]\leq 2d'|S|$.
Thus, for any $\gamma>0$ we can choose an $n$-independent $\alpha>0$ such that for any $S\subset\brk n$ of size $|S|\leq\alpha n$ we have $\Erw[X_S]\leq\gamma n/2$.
In fact, the Chernoff bound shows that by picking $\alpha>0$ sufficiently small, we can ensure that
	$$\pr\brk{\Vol(S)\geq\gamma n}\leq\pr\brk{X_S\geq\gamma n}\leq\exp(-\alpha n),$$
as desired.
\end{proof}

\begin{lemma}\label{Lemma_sigmaRandom}
Assume that there exist numbers $z>0$, $\eps>0$ and a sequence $(\sigma_n)_{n\geq1}$ of balanced maps $\brk n\ra\brk k$ 
such that
	$$\lim_{n\ra\infty}\frac1n\Erw\brk{\ln Z_{\beta,k}(G(n,p',\sigma_n))}>z+\eps.$$
Then
	$\limsup_{n\ra\infty}\pr\brk{\ln Z_{\beta,k}(G(n,p',\vec\sigma))\leq nz}^{1/n}<1.$
\end{lemma}
\begin{proof}
Let $Y=\frac1n\ln Z_{\beta,k}$ for the sake of brevity.
Suppose that $n$ is large enough so that $\Erw\brk{Y(G(n,p',\sigma_n))}>z+\eps/2$.
Set $n_i=|\sigma_n^{-1}(i)|$ and let $T$ be the set of all $\tau:\brk n\ra\brk k$ such that $|\tau^{-1}(i)|=n_i$ for $i=1,\ldots,k$.
As $Z_{\beta,k}$ is invariant under permutations of the vertices, we have
	\begin{equation}\label{eqLemma_sigmaRandom1}
	\Erw\brk{Y(G(n,p',\tau))}=\Erw\brk{Y(G(n,p',\sigma_n))}>z+\eps/2\quad\mbox{for any }\tau\in T.
	\end{equation}

Let $\gamma=\eps/(4\beta)>0$.
By \Lem~\ref{Lemma_Vol} there exists $\alpha>0$ such that for large enough $n$ for any set $S\subset V$ of size $|S|\leq\alpha n$ and
any $\sigma:\brk n\ra\brk k$ we have
	\begin{equation}\label{eqLemma_sigmaRandom4}
	\pr\brk{\Vol_{G(n,p',\sigma)}(S)>\gamma n}\geq1-\exp(-\alpha n).
	\end{equation}
Pick and fix a small $0<\eta<\alpha/3$ and let $\cA$ be the event that $\sum_{i=1}^k|\vec\sigma^{-1}(i)-n/k|\leq\eta n$.
Then by \Lem~\ref{Lemma_nearlyBalanced} there exist an ($n$-independent) number $\delta=\delta(\beta,\eps,\eta)>0$ such that
for $n$ large enough
	\begin{equation}\label{eqLemma_sigmaRandom3}
	\pr\brk{\cA}\geq1-\exp(-\delta n).
	\end{equation}

Because $\sigma_n$ is balanced, we have $|n_i-n/k|\leq\sqrt n$ for all $i\in\brk k$.
Therefore, if $\cA$ occurs, then it is possible to obtain from $\vec\sigma$ a map $\tau_{\vec\sigma}\in T$ by changing the colors of at most $2\eta n$ vertices.
If $\cA$ occurs, we let $\vec G_1=G(n,p',\tau_{\vec\sigma})$.
Further, let $\vec G_2$ be the random graph obtained by removing from $\vec G_1$ all edges
that are monochromatic under $\vec\sigma$.
Finally, let $\vec G_3$ be the random graph obtained from $\vec G_2$ by inserting an edge between any two
vertices $v,w$ with $\vec\sigma(v)\neq\vec\sigma(w)$ but $\tau_{\vec\sigma}(v)=\tau_{\vec\sigma}(w)$ with probability $p'$ independently.
Thus, the bottom line is that in $\vec G_3$, we connect any two vertices that are colored differently under $\vec\sigma$ with probability $p'$ independently.
That is, $\vec G_3=G(n,p',\vec\sigma)$.

Let $S_{\vec\sigma}$ be the set of vertices $v$ with $\vec\sigma(v)\neq\tau_{\vec\sigma}(v)$ and
let $\Delta$ be the number of edges we removed to obtain $\vec G_2$ from $\vec G_1$.
Then $\Delta$ is bounded by the volume of $S_{\vec\sigma}$ in  $\vec G_1=G(n,p',\tau_{\vec\sigma})$.
Hence, (\ref{eqLemma_sigmaRandom4}) implies that 
	\begin{equation}\label{eqLemma_sigmaRandom4a}
	\pr\brk{\Delta\leq\gamma n|\cA}\geq1-\exp(-\alpha n).
	\end{equation}
Since removing a single edge can reduce $Y$ by at most $\beta/n$, we obtain
	\begin{align*}
	\pr[Y(G(n,p',\vec\sigma))\leq z]&=\pr[Y(\vec G_3)\leq z]\leq \exp(-\delta n)+\pr[Y(\vec G_3)\leq z|\cA]&[\mbox{by~(\ref{eqLemma_sigmaRandom3})}]\\
		&\leq\exp(-\delta n)+\exp(-\alpha n)+\pr[Y(\vec G_3)\leq z|\cA,\Delta\leq\gamma n]&[\mbox{by~(\ref{eqLemma_sigmaRandom4a})}]\\
		&\leq\exp(-\delta n)+\exp(-\alpha n)+\pr[Y(\vec G_1)-\gamma\beta\leq z|\cA,\Delta\leq\gamma n]\\
		&\leq\exp(-\delta n)+\exp(-\alpha n)+2\pr[Y(\vec G_1)\leq z+\eps/4|\cA]&[\mbox{by~the choice of $\gamma$}]\\
		&\leq\exp(-\delta n)+\exp(-\alpha n)+3\pr[Y(G(n,p',\sigma_n))\leq z+\eps/4]\\
		&\leq\exp(-\delta n)+\exp(-\alpha n)\\
		&\qquad+3\pr[Y(G(n,p',\sigma_n))\leq \Erw[Y(G(n,p',\sigma_n))]-\eps/4].
			&[\mbox{by~(\ref{eqLemma_sigmaRandom1})}]
	\end{align*}
Finally, the assertion follows from 
	\Cor~\ref{Cor_ZAzuma}.
\end{proof}

\begin{proof}[Proof of \Lem~\ref{Lemma_partition}]
\Lem~\ref{Lemma_pickAndChoose} shows that there exist $\eps>0$, balanced maps $\sigma_n:\brk n\ra\brk k$ and a sequence
$\mu_n$ satisfying $|\mu_n-dn/2|\leq\sqrt n$ such that
	\begin{equation}\label{eqLemma_partition_1}
	\lim_{n\ra\infty}\pr\brk{\frac1n\ln|\cC(G(n,\mu_n,\sigma_n,\sigma_n))|\geq \ln k+\frac d2\ln(1-1/k)+\eps}=1.
	\end{equation}
By the definition of $Z_{\beta,k}$, (\ref{eqLemma_partition_1}) implies that
	\begin{equation}\label{eqLemma_partition_2}
	\lim_{n\ra\infty}\pr\brk{\frac1n\ln Z_{\beta,k}(G(n,\mu_n,\sigma_n))\geq \ln k+\frac d2\ln(1-1/k)+\eps}=1\qquad\mbox{for all }\beta>0.
	\end{equation}
By comparison, \Lem~\ref{Lemma_upperBoundAnnealed} yields $\beta>0$ such that with $z=\ln k+\frac d2\ln(1-1/k)+\eps/8$ we have
	\begin{equation}\label{eqLemma_partition_3}
	\lim_{n\ra\infty}\pr\brk{\frac1n\ln Z_{\beta,k}(G(n,m))\leq z}=1.
	\end{equation}
Thus, we aim to prove that there is $\alpha>0$ such that for sufficiently large $n$
	\begin{equation}\label{eqLemma_partition_4}
	\pr\brk{\frac1n\ln Z_{\beta,k}(G(n,m,\vec\sigma))\leq z+\eps/8}\leq\exp(-\alpha n).
	\end{equation}
Indeed, since $\ln Z_{\beta,k}(G(n,\mu_n,\sigma_n))\leq \beta\mu_n=O(n)$, (\ref{eqLemma_partition_2}) implies that for large enough $n$
	\begin{equation}\label{eqLemma_partition_5}
	\frac1n\Erw[\ln Z_{\beta,k}(G(n,\mu_n,\sigma_n))]\geq \ln k+\frac d2\ln(1-1/k)+\eps-o(1)\geq\ln k+\frac d2\ln(1-1/k)+\eps/2.
	\end{equation}
Thus, (\ref{eqLemma_partition_4})  follows from \Lem~\ref{Lemma_sigmaRandom}.
\end{proof}

\section{The fixed point problem}\label{Sec_fix}

\subsection{The branching process}
Throughout this section we assume that $(2k-1)\ln k-3\leq d\leq(2k-1)\ln k$.
Moreover, we recall that $d'=kd/(k-1)$.

\begin{lemma}\label{Lemma_GW}
Suppose that $d\geq(2k-1)\ln k-2$.
\begin{enumerate}
\item The function
	\begin{eqnarray}\label{eqThm_condF}
	F_{d,k}:[0,1]^k\ra[0,1]^k,&&(q_1,\ldots,q_k)\mapsto 
		\bigg(\frac1k{\prod_{j\in\brk k\setminus\cbc i}1-\exp\bc{-d' q_j}}\bigg)_{i\in\brk k}
	\end{eqnarray}
	has a unique fixed point $\vec q^*=(q_1^*,\ldots,q_k^*)$ 
		such that $\sum_{j\in\brk k}q_j^*\geq2/3$.
	This fixed point has the property that $q_1^*=\cdots =q_k^*$.
	Moreover, $q^*=kq_1^*$ is the unique fixed point of the function~(\ref{eqSimpleFix})
		in the interval $[2/3,1]$, and
		$q^*=1-O_k(1/k)$.
\item The branching process $\GW(d,k,\vec q^*)$ is sub-critical.
\item 	Furthermore,
	$\frac{\partial}{\partial d}\Erw\brk{\frac{\ln\cZ(\T_{d,k}(\vec q^*))}{|\T_{d,k}(\vec q^*)|}}=\tilde O_k(k^{-2})$.
\end{enumerate}
\end{lemma}

The proof of \Lem~\ref{Lemma_GW} requires several steps.
We begin by studying the fixed points of $F_{d,k}$.

\begin{lemma}\label{Prop_Fixpunkt}
The function $F_{d,k}$ maps the compact set $[\frac 2{3k},\frac 1k]^k$ into itself and has a unique fixed point $\vec q^*$ in this set.
Moreover, the function from~(\ref{eqSimpleFix}) has a unique fixed point $q^*$ in the set $[2/3,1]$ and $\vec q^*=(q^*/k,\ldots,q^*/k)$.
Furthermore, 
	\begin{equation}\label{eqq*apx}
	q^*=1-1/k+o_k(1/k).
	\end{equation}	
In addition, if $\vec q\in[0,1]^k$ is a fixed point of $F_{d,k}$, then
	\begin{equation}\label{eqFdkFixedPointSym}
	q_1=\cdots=q_k.
	\end{equation}	
\end{lemma}
\begin{proof}
Let $I=[\frac 2{3k},\frac1k]^k$.
As a first step, we show that $F_{d,k}(I) \subset I$.
Indeed, let $\vec q\in I$.
Then for any $i\in[k]$
	\begin{eqnarray*}
	\br{F_{d,k}(\vec q)}_i&=&\frac 1k \prod_{j\neq i}1-\exp(-d'q_j)\leq \frac 1k.
	\end{eqnarray*}
On the other hand, as $d\geq(2k-1)\ln k$ we see that $d'\geq1.99k\ln k$.
Hence,
	\begin{eqnarray*}
	\br{F_{d,k}(\vec q)}_i&=&\frac 1k \prod_{j\neq i}1-\exp(-d'q_j)
		\geq\frac 1k \bc{1-\exp\bc{-\frac{2d'}{3k}}}^{k-1}\geq\frac 1k(1-k^{-1.1})^k=\frac{1-o_k(1)}k.
	\end{eqnarray*}
Thus, $F_{d,k}(I)\subset I$.

In addition, we claim that $F_{d,k}$ is contracting on $I$.
In fact, for any $i,j\in\brk k$
	\begin{eqnarray*}
	\frac\partial{\partial q_j}\br{F_{d,k}(\vec q)}_i&=&\frac{\vecone_{i\neq j}}k\frac\partial{\partial q_j}\prod_{l\neq i}1-\exp(-d'q_l)
		=\frac{\vecone_{i\neq j}d'}{k\exp(d'q_j)}\cdot\prod_{l\neq i,j}1-\exp(-d'q_l)\\
		&=&(1+o_k(1))\frac{\vecone_{i\neq j}d'}{k\exp(d'q_j)}\qquad\mbox{[as $d'\geq1.99k\ln k$ and $q_l\geq2/3$ for all $l$]}\\
		&\leq&k^{-1.3}\qquad\qquad\qquad\qquad\qquad\qquad\mbox{[for the same reason]}.
	\end{eqnarray*}
Therefore, for $\vec q\in I$ the Jacobi matrix $D F_{d,k}(\vec q)$ satisfies
	\begin{eqnarray*}
	\norm{D F_{d,k}(\vec q)}^2&\leq&\sum_{i,j\in\brk k}\bc{\frac\partial{\partial q_j}\br{F_{d,k}(\vec q)}_i}^2
		\leq k^2\cdot k^{-2.6}<1.
	\end{eqnarray*}
Thus, $F_{d,k}$ is a contraction on the compact set $I$.
Consequently, Banach's fixed point theorem implies that there is a unique fixed point $\vec q_*\in I$.

To establish~(\ref{eqFdkFixedPointSym}), assume without loss that $\vec q=(q_1,\ldots,q_k)\in[0,1]^k$ is a fixed point such that $q_1\leq\cdots\leq q_k$.
Then $q_1>0$ because $F_{d,k}$ maps $[0,1]^k$ into $(0,1]^k$.
Moreover, because $\vec q$ is a fixed point, we find
	$$\frac{q_k}{q_1}=\frac{(F_{d,k}(\vec q))_k}{(F_{d,k}(\vec q))_1}
		=\frac{1-\exp(-d'q_1)}{1-\exp(-d'q_k)}\leq1\qquad\qquad\qquad\mbox{[as $q_1\leq q_k$]},$$
whence~(\ref{eqFdkFixedPointSym}) follows.

Further, we claim that the function
	$f_{d,k}:[0,1]\ra[0,1]$, $q\mapsto(1-\exp(-dq/(k-1)))^{k-1}$
maps the interval $[2/3,1]$ into itself.
This is because for $q\in[2/3,1]$ we have
	$0\leq\exp(-dq/(k-1))\leq k^{-1.3}$ due to our assumption on $d$.
Moreover, the derivative of $f$ works out to be
	$f'_{d,k}(q)=d\exp(-dq/(k-1))(1-\exp(-dq/(k-1)))^{k-2}$.
Thus, for $q\in[2/3,1]$ we find $0\leq f'_{d,k}(q)<1/2$.
Hence, $f_{d,k}$ has a unique fixed point $q_*\in[2/3,1]$.
Comparing the expressions $f_{d,k}(q)$ and $F_{d,k}(\vec q)$, we see
that $(q_*/k,\ldots,q_*/k)$ is a fixed point of $F_{d,k}$.
Consequently, $\vec q_*=(q_*/k,\ldots,q_*/k)$.

Finally, since $f'_{d,k}(q)>0$ for all $q$, the function $f_{d,k}$ is strictly increasing.
Therefore, as $d=(2-o_k(1))k\ln k$,
	\begin{equation}\label{eqq*1}
	q_*=f_{d,k}(q_*)\leq f_{d,k}(1)=(1-\exp(-d/(k-1)))^{k-1}=1-1/k+o_k(1/k).
	\end{equation}
Similarly, $q_*\geq f_{d,k}(2/3)\geq 1-k^{-0.3}$.
Hence, because $d\geq(2k-1)\ln k-3$, we obtain
	\begin{eqnarray}\nonumber
	q_*&=&f_{d,k}(q_*)\geq f_{d,k}(1-k^{-0.3})=\bc{1-\exp\brk{-\frac{d(1-k^{-0.3})}{k-1}}}^{k-1}\\
		&=&\bc{1-k^{-2}+O_k(k^{-2.1})}^{k-1}=1-1/k+o_k(1/k).
			\label{eqq*2}
	\end{eqnarray}
Combining~(\ref{eqq*1}) and~(\ref{eqq*2}), we conclude that $q_*=1-1/k+o_k(1/k)$,
as claimed.
\end{proof}

\begin{remark}
The proof of several statements in this section
(\Lem s~\ref{Prop_Fixpunkt}, \ref{Lemma_hardFields}, \ref{Lemma_separateColors}, \ref{Cor_balanced} and \Cor~\ref{Claim_hardFields1})
directly incorporate parts of the calculations outlined in the physics work~\cite{LenkaFlorent} that predicted the existence and location of $\dc$.
We redo these calculations here in detail to be self-contained and because not all steps are carried out in full detail in~\cite{LenkaFlorent}.
\end{remark}

From here on out, we let $\vec q^*$ denote the fixed point of $F_{d,k}$ in $[2/(3k),1]^k$ and we
denote the fixed point of the function~(\ref{eqSimpleFix}) in the interval $[2/3,1]$ by $q^*$.
Hence, $\vec q^*=(q^*/k,\ldots,q^*/k)$.
If we keep $k$ fixed, how does $q^*$ vary with $d$?

\begin{corollary}\label{Cor_q*diff}
We have $\frac{\dd q^*}{\dd d}={\Theta}_k \br{ k^{-2}}$.
\end{corollary}

\begin{proof}
The map $d\mapsto q^*$ is differentiable by the implicit function theorem.
Moreover, differentiating~(\ref{eqSimpleFix}) while keeping in mind that $q^*=q^*(d)$ is a fixed point, we find
	\begin{eqnarray*}
	\frac{\dd q^*}{\dd d}&=&\frac{\dd}{\dd d}\,(1-\exp(-d q^*/(k-1)))^{k-1}\\
		&=&\frac{(k-1)\br{1-\exp\br{-dq^*/(k-1)}}^{k-2}}
			{\exp\br{dq^*/(k-1)}}\cdot \br{\frac{q^*}{k-1}+\frac d{k-1}\frac{\dd q^* }{\dd d}}\\
	\end{eqnarray*}
Rearranging the above using $d=2k\ln k+O_k(\ln k)$ and~(\ref{eqq*apx}) yields the assertion.
\end{proof}

\begin{corollary}\label{lem_q_il}
We have $q_{i,\ell}^*=\tilde{\Theta}_k\left(k^{-(2|\ell|-1) }\right)$  for all $(i,\ell)\in \cT$.
Moreover, 
	$\frac{\dd q_{i,\ell}^*}{\dd d} =\tilde{O}_k\left(|\ell|k^{-2|\ell|}\right).$
\end{corollary}\label{lem_bound_par_q}
\begin{proof}
\Lem~\ref{Prop_Fixpunkt} shows that $q_j^*=q_*/k$ for all $j\in\brk k$.
Hence, due to~(\ref{eqq*apx}) and because $d'=2k\ln k+O_k(\ln k)$ we obtain
	\begin{eqnarray*}
	q_{i,\ell}^*&=&\frac 1k\prod_{j\in\brk k\setminus \ell}1-\exp\br{-d'q_j^*}
		\prod_{j\in \ell \setminus \{i\}}\exp\br{-d'q_j^*}=\tilde\Theta_k(k^{-(2|\ell|-1)}).
	\end{eqnarray*}
Furthermore, applying \Cor~\ref{Cor_q*diff}, we get
	\begin{eqnarray*}
	\frac{\dd q_{i,\ell}^*}{\dd d}&=& \frac 1k \frac{\dd }{\dd d}
			\brk{\prod_{j\in\brk k\setminus \ell}1-\exp\br{-d'q_j^*}\prod_{j\in \ell \setminus \{i\}}\exp\br{-d'q_j^*}} \\
			&=&\frac 1k \frac{\dd }{\dd d}
			\brk{(1-\exp\br{-d'q_*/k})^{k-|\ell|}\exp\br{-d'q_*/k}^{|\ell|-1}
				} \\
			&=& \frac 1k\br{\frac{q_*}{k-1}+\frac{d'}k\frac{\dd q_*}{\dd d}} 
					\bigg[\frac{k-|\ell|}{ \exp(d'q_*/k)} 
							 (1-\exp(-d'q_*/k))^{k-|\ell|-1}\\
		&&\qquad - (|\ell|-1)(1-\exp\br{-d'q_*/k})^{k-|\ell|}\bigg]\exp\br{-d'(|\ell|-1)q_*/k} \\
		&=&|\ell|O_k(k^{-2})\exp\br{-d'(|\ell|-1)q_*/k}=\tilde O_k(|\ell|k^{-2|\ell|}),
	\end{eqnarray*}
provided that $|\ell|\leq\ln k$.
\end{proof}

\begin{lemma}\label{Lemma_subcrit}
The branching process $\GW(d,k,\vec q_*)$ is sub-critical.
\end{lemma}
\begin{proof}
We introduce another branching process $\GW'(d,k,\vec q^*)$ with only three types $1,2,3$.
The idea is that type 1 of the new process represents all types $(h,\cbc{h})\in\cT$ with $h\in\brk k$,
	that $2$ represents all types $(h,\cbc{j,h})\in T$ with $h,j\in\brk k$, $j\neq h$,
	and that $3$ lumps together all of the remaining types.
More specifically, in $\GW'(d,k,\vec q^*)$ an individual of type $i$ spawns a Poisson number $\Po(M_{ij})$ of offspring of type $j$ ($i,j\in\cbc{1,2,3}$), where
	$M=(M_{ij})$ is the following matrix.
If either $i=1$ or $j=1$, then $M_{ij}=0$.
Moreover,
	\begin{align*}
	M_{22}&=\sum_{(i,\ell)\in\cT_{(1,\cbc{1,2})}:|\ell|=2}q_{i,\ell}^*d'
		&
		M_{23}&=\sum_{(i,\ell)\in\cT_{(1,\cbc{1,2})}:|\ell|>2}q_{i,\ell}^*d',\\
	M_{32}&=\sum_{(i,\ell)\in\cT_{(1,\brk k)}:|\ell|=2}q_{i,\ell}^*d',&
	M_{33}&=\sum_{(i,\ell)\in\cT_{(1,\brk k)}:|\ell|>2}q_{i,\ell}^*d'.
	\end{align*}
Due to the symmetry of the fixed point $\vec q^*$ (i.e., $\vec q^*=(q^*/k,\ldots,q^*/k)$),
$M_{22}$ is precisely the expected number of offspring of type $(i,\ell)$ with $|\ell|=2$ that an individual of type $(i_0,\ell_0)\in\cT$ with $|\ell_0|=2$ spawns
in the branching process $\GW(d,k,\vec q^*)$.
Similarly, $M_{23}$ is just the expected offspring of type $(i,\ell)$ with $|\ell|>2$ of an individual with $|\ell_0|=2$.
Furthermore, $M_{32}$ is an upper bound on the expected offspring of type $(i',\ell')$ with $|\ell'|=2$ of
an individual of type $(i_0,\ell_0)$ with $|\ell_0|>2$.
Indeed, $M_{32}$ is the the expected offspring in the case that $\ell_0=\brk k$, which is the case that yields the largest possible expectation.
Similarly, $M_{33}$ is an upper bound on the expected offspring of type $(i',\ell')$ with $|\ell'|>2$ in the case $|\ell_0|>2$.
Therefore, if $\GW'(d,k,\vec q^*)$ is sub-critical, then so is $\GW(d,k,\vec q^*)$.

To show that this is the case, we need to estimate the entries $M_{ij}$.
Estimating the $q_{i,\ell}^*$ via \Cor~\ref{lem_q_il}, we obtain
	\begin{align*}
	M_{22}&\leq 2kq_{1,\cbc{1,2}}^*d'\leq\tilde O_k(k^{-1}),&
	M_{23}&\leq 2\sum_{l\geq3}l\bink{k}{l-1}q_{1,\brk l}^*d'\leq\tilde O_k(k^{-2}),\\
	M_{32}&\leq k(k-1)q_{1,\cbc{1,2}}^*d'\leq\tilde O_k(1),&
	M_{33}&\leq k\sum_{l\geq3}l\bink{k}{l-1}q_{1,\brk l}^*d'\leq\tilde O_k(k^{-1}).
	\end{align*}
The branching process $\GW'(d,k,\vec q^*)$ is sub-critical iff all eigenvalues of $M$ are less than $1$ in absolute value.
Because the first row and column of $M$ are $0$, this is the case iff the eigenvalues of the $2\times 2$ matrix
$M_*=(M_{ij})_{2\leq i,j\leq3}$ are less than $1$ in absolute value.
Indeed, since the above estimates show that $M_*$ has trace $\tilde O_k(k^{-1})$ and determinant
$\tilde O_k(k^{-2})$, both eigenvalues of $M_*$ are $\tilde O_k(k^{-1})$.
\end{proof}

\begin{lemma}\label{Lemma_diffTree}
We have $\frac{\dd}{\dd d}\Erw[|\T_{d,k,\vec q^*}|^{-1}\ln\cZ(\T_{d,k,\vec q^*})]\leq\tilde O_k(k^{-2})$.
\end{lemma}
\begin{proof}
Fix a number $d\in[(2k-1)\ln k-2,(2k-1)\ln k]$ and a small number $\eps>0$ and let $\hat d=d+\eps$.
Let $\vec q^*$ be the unique fixed point of $F_{d,k}$ in $[2/3,1]^k$ and let
	$\vec{\hat q}^*$ be the unique fixed point of $F_{\hat d,k}$ in $[2/3,1]^k$.
Set $d'=dk/(k-1)$ and $\hat d'=\hat dk/(k-1)$.
Moreover, let us introduce the shorthands $\T=\T_{d,k,\vec q^*}$ and $\hat \T=\T_{d,k,\vec{\hat q}^*}$.
We aim to bound
	\begin{align*}
	\Delta=\abs{\Erw\brk{\frac{\ln\cZ(\T)}{|\T|}}-
		\Erw\brk{\frac{\ln\cZ(\hat\T)}{|\hat\T|}}}
	\end{align*}

To this end, we couple $\T$ and $\hat \T$ as follows.
\begin{itemize}
\item In $\T,\hat\T$ the type $(i_0,\ell_0)$ resp.\ $(\hat i_0,\hat\ell_0)$ of the root $v_0$ is chosen from the distribution
		$$Q=(q_{i,\ell})_{(i,\ell)\in\cT}\quad\mbox{resp.}\quad\hat Q=(\hat q_{i,\ell})_{(i,\ell)\in\cT}.$$
	We couple $(i_0,\ell_0)$, $(\hat i_0,\hat\ell_0)$ optimally.
\item If $(i_0,\ell_0)\neq(\hat i_0,\hat\ell_0)$, then we generate $\T$, $\hat\T$ independently
		from the corresponding conditional distributions given the type of the root.
\item If $(i_0,\ell_0)=(\hat i_0,\hat\ell_0)$, we generate a random tree $\widetilde\T$ by means of the following branching process.
		\begin{itemize}
		\item Initially, there is one individual. Its type is $(i_0,\ell_0)$.
		\item Each individual of type $(i,\ell)$ spawns a $\Po(\Lambda_{i',\ell'})$ number of offspring of each type $(i',\ell')\in\cT_{i,\ell}$,
			where $$\Lambda_{i',\ell'}=\max\cbc{q_{i',\ell'}^*d',\hat q_{i',\ell'}^*\hat d'}.$$
		\item Given that the total progeny is finite, we obtain $\widetilde\T$ by linking each individual to its offspring.
		\end{itemize}
\item For each type $(i,\ell)$ let
		$$\lambda_{i,\ell}=1-\min\cbc{d'q_{i,\ell}^*,\hat d'\hat q_{i,\ell}^*}/\Lambda_{i,\ell}.$$
	For every vertex $v$ of $\widetilde\T$ let $s_v$ be a random variable with distribution $\Be(\lambda_{i_v,\ell_v})$,
		where $(i_v,\ell_v)$ is the type of~$v$.
	The random variables $(s_v)_v$ are mutually independent.
\item Obtain $\T$ from $\widetilde\T$ by deleting all vertices $v$ such that 
		$d'q^*_{i_v,\ell_v}<d'\hat q_{i_v,\ell_v}^*$ and 	$s_v=1$, along with the pending sub-tree.
\item	Similarly, obtain $\hat\T$ from $\widetilde\T$ by deleting all $v$ and their sub-trees such that 
		$d'q^*_{i_v,\ell_v}>d'\hat q_{i_v,\ell_v}^*$ and $s_v=1$.
\end{itemize}

Let $\cA$ be the event that the type of the root satisfies $\ell_0=\cbc{i_0}$ and let $\hat\cA$ be the event $\hat\ell_0=\{\hat i_0\}$.
If $\cA\cap\hat\cA$ occurs, then both $\T$, $\hat\T$ consist of a single vertex and have precisely one legal coloring.
Thus, $|\T|^{-1}\ln\cZ(\T)=|\hat\T|^{-1}\ln\cZ(\hat\T)=0$.
Consequently,
	\begin{align*}
	\Delta
		&\leq\Erw\brk{\abs{\frac{\ln\cZ(\T)}{|\T|}-\frac{\ln\cZ(\hat\T)}{|\hat\T|}}\,\bigg|
			\,\neg\cA\vee\neg\hat\cA}\cdot\pr\brk{\neg\cA\vee\neg\hat\cA}.
	\end{align*}
Further, since $|\T|^{-1}\ln\cZ(\T),|\hat\T|^{-1}\ln\cZ(\hat\T)\leq\ln k$ with certainty, we obtain
	\begin{align*}
	\Delta&\leq\bc{\pr\brk{\neg\cA\wedge\hat\cA}+\pr\brk{\cA\wedge\neg\hat\cA}}\ln k+
		\Erw\brk{\abs{\frac{\ln\cZ(\T)}{|\T|}-\frac{\ln\cZ(\hat\T)}{|\hat\T|}}\,\bigg|
			\,\neg\cA\wedge\neg\hat\cA}\cdot\pr\brk{\neg\cA\wedge\neg\hat\cA}.
	\end{align*}
Because $(i_0,\ell_0)$ and $(\hat i_0,\hat\ell_0)$ are coupled optimally and
	$\pr[\cA]=kq_1^*$, $\pr[\hat \cA]=k\hat q_1^*$, \Cor~\ref{Cor_q*diff} implies that
$\pr[\neg\cA\wedge\hat\cA],\pr[\cA\wedge\neg\hat\cA]\leq\eps\tilde O_k(k^{-2})$.
Hence,
	\begin{align}\label{eqCoupling1}
	\Delta&\leq\eps\tilde O_k(k^{-2})+
		\Erw\brk{\abs{\frac{\ln\cZ(\T)}{|\T|}-\frac{\ln\cZ(\hat\T)}{|\hat\T|}}\,\bigg|
			\,\neg\cA\wedge\neg\hat\cA}\cdot\pr\brk{\neg\cA\wedge\neg\hat\cA}.
	\end{align}
Now, let $\cE$ be the event that $\ell_0\neq\{i_0\}$, $\hat\ell_0\neq\{\hat i_0\}$ and $(i_0,\ell_0)=(\hat i_0,\hat\ell_0)$.
Due to \Cor~\ref{lem_q_il} and because $(i_0,\ell_0)$, $(\hat i_0,\hat\ell_0)$ are coupled optimally, we see that
	\begin{equation}\label{eqCoupling2}
	\pr\brk{\neg\cA\wedge\neg\hat\cA\wedge\neg\cE}\leq \eps\tilde O_k(k^{-2}).
	\end{equation}
Combining~(\ref{eqCoupling1}) and~(\ref{eqCoupling2}), we conclude that
	\begin{align}\label{eqCoupling3}
	\Delta&\leq\eps\tilde O_k(k^{-2})+\Erw\brk{\abs{\frac{\ln\cZ(\T)}{|\T|}-\frac{\ln\cZ(\hat\T)}{|\hat\T|}}\,\bigg|\cE}
			\cdot\pr\brk{\neg\cA\wedge\neg\hat\cA}
	\end{align}
Further, since $\pr\brk{\neg\cA\wedge\neg\hat\cA}\leq\pr\brk{\neg\cA}\leq1-kq_1^*\leq O_k(1/k)$ by \Lem~\ref{Lemma_GW}, (\ref{eqCoupling3}) yields
	\begin{align}
	\Delta&\leq\eps\tilde O_k(k^{-2})+O_k(1/k)\cdot\Erw\brk{\abs{\frac{\ln\cZ(\T)}{|\T|}-\frac{\ln\cZ(\hat\T)}{|\hat\T|}}\,\bigg|\cE}
		\leq\eps\tilde O_k(k^{-2})+O_k(\ln k/k)\cdot\pr\brk{\T\neq\hat\T|\cE}.
			\label{eqCoupling4}
	\end{align}

Thus, we are left to estimate the probability that $\T\neq\hat\T$, given that both trees have a root of the same type $(i_0,\ell_0)$ with $|\ell_0|>1$.
Our coupling ensures that this event occurs iff $s_v=1$ for some vertex $v$ of $\widetilde\T$.
To estimate the probability of this event, we observe that by \Cor~\ref{lem_q_il}
	\beq\label{eqCoupling5}
	\lambda_{i,\ell}\leq\begin{cases}
		\eps\tilde O_k(1/k)&\mbox{ if }|\ell|=2,\\
		\eps\tilde O_k(1)&\mbox{ if }|\ell|>2.
		\end{cases}\eeq
Now, let $\cN_1$ be the number of vertices $v\neq v_0$ of $\widetilde\T$ such that $|\ell_v|=2$,
and let $\cN_2$ be the number of $v\neq v_0$ such that $|\ell_v|>2$.
Then~(\ref{eqCoupling4}), (\ref{eqCoupling5}) and the construction of the coupling yield
	\beq\label{eqCoupling6}
	\Delta/\eps\leq\tilde O_k(k^{-2})+\tilde O_k(k^{-1})\bc{k^{-1}\Erw[\cN_1|\cE]+\cdot\Erw[\cN_2|\cE]}.
	\eeq

To complete the proof, we claim that
	\begin{equation}\label{eqCoupling6a}
	\Erw[\cN_1|\cE]\leq\tilde O_k(k^{-1}),\qquad\Erw[\cN_2|\cE]\leq\tilde O_k(k^{-2}).
	\end{equation}
Indeed, consider the matrix $\tilde M=(\tilde M_{ij})_{i,j=1,2}$ with entries
	\begin{align*}
	\tilde M_{11}&=\sum_{(i,\ell)\in\cT_{1,\cbc{1,2}}:|\ell|=2}\Lambda_{i,\ell},&
		\tilde M_{12}&=\sum_{(i,\ell)\in\cT_{1,\cbc{1,2}}:|\ell|>2}\Lambda_{i,\ell},\\
	\tilde M_{21}&=\sum_{(i,\ell)\in\cT_{1,\brk k}:|\ell|=2}\Lambda_{i,\ell},&
		\tilde M_{22}&=\sum_{(i,\ell)\in\cT_{1,\brk k}:|\ell|>2}\Lambda_{i,\ell}.
	\end{align*}
Then \Cor~\ref{lem_q_il} entails that
	\begin{align}\label{eqCoupling7}
	\tilde M_{11}&=\tilde O_k(k^{-1}),&
		\tilde M_{12}&=\tilde O_k(k^{-2}),&
		\tilde M_{21}&=\tilde O_k(1),&\tilde M_{22}&=\tilde O_k(k^{-1}).
	\end{align}
In addition, let $\xi=\bink{x_1}{x_2}$, where $\xi_1=1-\xi_2=\pr\brk{|\ell_0|=2|\cE}$.
Then \Cor~\ref{lem_q_il} shows that $\xi_2=\tilde O_k(k^{-2})$.
Furthermore, by the construction of the branching process and~(\ref{lem_q_il})
	\begin{align*}
	\bink{\Erw\brk{\cN_1|\cE}}{\Erw\brk{\cN_2|\cE}}&\leq\sum_{t=1}^\infty\tilde M^t\xi
		=\bink{\tilde O_k(k^{-1})}{\tilde O_k(k^{-2})},
	\end{align*}
which implies~(\ref{eqCoupling6a}).

Finally, (\ref{eqCoupling6}) and~(\ref{eqCoupling6a}) imply that $\Delta\leq\eps\tilde O_k(k^{-2})$.
Taking $\eps\ra0$ completes the proof.
\end{proof}

\begin{proof}[Proof of \Lem~\ref{Lemma_GW}.]
The first assertion is immediate from \Lem~\ref{Prop_Fixpunkt}.
The second claim follows from \Lem~\ref{Lemma_subcrit}, and the third one from \Lem~\ref{Lemma_diffTree}.
\end{proof}

\subsection{
	The ``hard fields''} \label{sec_proof_hard_fields}

In this section we make the first step towards proving that $\pi_{d,k,\vec q^*}$ is the unique frozen fixed point of $\cF_{d,k}$.
More specifically, identifying the set $\Omega$ with the $k$-simplex,
we show that every face of $\Omega$ carries the same probability mass under any frozen fixed point of $\cF_{d,k}$ as under the measure $\pi_{d,k,\vec q^*}$.
Formally, let us denote the extremal points of $\Omega$ by $\atom_h=(\vecone_{i=h})_{i\in\brk k}$, i.e.,
	$\atom_h$ is the probability measure on $\brk k$ that puts mass $1$ on the single point $h\in\brk k$.
In addition, let $\Omega_\ell$ be the set of all $\mu\in\Omega$ with support $\ell$
	(i.e., $\mu(i)>0$ for all $i\in\ell$ and $\mu(i)=0$ for all $i\in\brk k\setminus\ell$).
Further,  for a probability measure $\pi\in\cP$ we let $\rho_h(\pi)=\pi(\cbc{\atom_h})$ denote the probability mass of $\atom_h$ under $\pi$.
In physics jargon, the numbers $\rho_h(\pi)$ are called the ``hard fields'' of $\pi$.
In addition, recalling that $\dd\pi_i(\mu)=k\mu(i)\dd\pi(\mu)$, we set
	$\rho_{i,\ell}(\pi)=\pi_i(\Omega_\ell)$ for any $(i,\ell)\in\cT$.
The main result of this section is

\begin{lemma}\label{Lemma_hardFields}
Suppose that $d\geq(2k-1)\ln k-2$.
Let $q^*\in[2/3,1]$ be the fixed point of~(\ref{eqSimpleFix}).
If $\pi\in\cP$ is a frozen fixed point of $\cF_{d,k}$, then 
	$\rho_i(\pi)=q^*/k$ and $\rho_{i,\ell}(\pi)=kq_{i,\ell}^*$ for all $(i,\ell)\in\cT$.
\end{lemma}

To avoid many case distinctions, we introduce the following convention when working with product measures.
Let us agree that $\Omega^0=\cbc\emptyset$.
Hence, if $B:\Omega^0\ra \Omega$ is a map, then $B(\emptyset)\in\Omega$.
Furthermore, there is a precisely one probability measure $\pi_0$ on $\Omega^0$, namely the measure that puts
mass one on the point $\emptyset\in \Omega^0$.
Thus, the integral $\int_{\Omega^\emptyset}B(\mu)\dd\pi_0(\mu)$ is simply equal to $B(\emptyset)$.
If $\pi_1,\pi_2,\ldots$ are probability mesures on $\Omega$, what we mean by
the  empty product measure $\bigotimes_{\gamma=1}^0\pi_\gamma$ is just the measure $\pi_0$ on $\Omega^0$.

Further, for a real $\lambda\geq0$ and an integer $y\geq1$ we let
	$$p_{\lambda}(y)=\lambda^y\exp(-\lambda)/y!.$$
Moreover, for $i\in\brk k$ we let $\Gamma_i$ be the set of all non-negative integer vectors
	$\vec\gamma=(\gamma_j)_{j\in\brk k\setminus\cbc i}$ and
for $\vec \gamma \in \Gamma_i$ we set
	\begin{align*}
	p_{i}(\vec \gamma)&= \prod_{h \in [k] \setminus \{i\}} p_{\frac{d}{k-1}}(\gamma_{h}).
	\end{align*}
We also let $\Omega^{\vec\gamma}=\prod_{h\in [k] \setminus \{i\}}\prod_{j\in [\gamma_{h}]}\Omega$  for $\vec\gamma\in\Gamma_i$.
The elements of $\Omega^{\vec\gamma}$ are denoted by
	$\mu_{\vec \gamma}= (\mu_{h,j})_{h\in [k] \setminus \{i\}, j \in [\gamma_{h}]}.$
Moreover, let
	\begin{align*}
	\pi_{i, \vec \gamma}&= \bigotimes_{h\in [k] \setminus \{i\}} \bigotimes_{j \in [\gamma_{h}]} \pi_{h}.
	\end{align*}
Thus, with the convention from the previous paragraph, in the case $\vec \gamma=0$ the set $\Omega^{\vec \gamma}=\cbc\emptyset$
contains only one element, namely $\mu_0=\emptyset$.
Moreover, $\pi_{i, \vec \gamma}$ is the probability measure on $\Omega^0$ that gives mass one to the point $\emptyset$.
We recall the map $\cB:\bigcup_{\gamma\geq1}\Omega^\gamma\ra\Omega$ from~(\ref{eqBPOperator}) and extend
this map to $\Omega^0$ by letting $\cB(\emptyset)=\frac1k\vecone$ be the uniform distribution on~$\Omega$.
We start the proof of \Lem~\ref{Lemma_hardFields} by establishing the following identity.

\begin{lemma}\label{Lemma_separateColors}
If $\pi$ is fixed point of $\cF_{d,k}$, then for any $i\in\brk k$ we have
	$$\pi_i= \sum_{\vec\gamma\in\Gamma_i}
				\int_{\Omega^{\vec\gamma}}\atom_{\cB[\mu_{\vec \gamma}]}
			\measurei .$$
\end{lemma}

To establish \Lem~\ref{Lemma_separateColors} we need to calculate the normalising quantities $Z_{\gamma}(\pi)$.

\begin{lemma}\label{Cor_balanced_fix}
If $\pi$ is fixed point of $\cF_{d,k}$, then $Z_\gamma(\pi)=(k-1)^\gamma/k^{\gamma-1}$.
\end{lemma}
\begin{proof}
Assume that $\pi$ is fixed point of $\cF_{d,k}$.
We claim that
	\begin{equation}\label{eqLemma_balanced}
	\int_\Omega \mu(h) \dd \pi(\mu)=1/k
	\qquad\mbox{for all $h\in\brk k$.}
	\end{equation}
Indeed, set
	$\nu(h)= \int_\Omega \mu(h)  \dd \pi(\mu).$
Then $\nu$ is a probability distribution on $\brk k$.
Since $\pi$ is a fixed point of $\cF_{d,k}$, we find
	\begin{eqnarray}
	\nu(h)&=&\int_\Omega\mu(h)\dd\cF_{d,k}[\pi](\mu)
		=\sum_{\gamma=0}^\infty\frac{p_d(\gamma)}{Z_\gamma(\pi)}\int_{\Omega^\gamma}
			\brk{\sum_{h=1}^k \prod_{j=1}^{\gamma} 1 - \mu_j(h) }
				\cB[\mu_1, \dots, \mu_\gamma](h)\bigotimes_{j=1}^\gamma  \dd \pi(\mu_j)\nonumber\\
		&=&\sum_{\gamma=0}^\infty\frac{p_d(\gamma)}{Z_\gamma(\pi)}\int_{\Omega^\gamma}
			\prod_{j=1}^\gamma1-\mu_j(h)\bigotimes_{j=1}^\gamma  \dd \pi(\mu_j)
				\qquad\mbox{[plugging in~(\ref{eqBPOperator})]}\nonumber\\
		&=&\sum_{\gamma=0}^\infty\frac{p_d(\gamma)}{Z_\gamma(\pi)}\brk{\int_\Omega1-\mu(h)\dd\pi(\mu)}^\gamma
		= \sum_{\gamma \geq 0} \frac{\left(1-\nu(h)\right)^\gamma p_{d}(\gamma) }{\sum_{h' \in [k]} \left( 1- \nu(h')\right)^\gamma}
			\quad[\mbox{due to~(\ref{eqZgamma})}].
		\label{eq_rs_averaged}	
	\end{eqnarray}
Now, assume that $h_1,h_2\in\brk k$ are such that $\nu(h_1)\leq\nu(h_2)$.
Then~(\ref{eq_rs_averaged}) yields
	\begin{eqnarray*}
	\nu(h_2)&=&
		\sum_{\gamma \geq 0} \frac{\left(1-\nu(h_1)\right)^\gamma p_{d}(\gamma)}
				{\sum_{h' \in [k]} \left( 1- \nu(h')\right)^\gamma}
			\leq\sum_{\gamma \geq 0} \frac{\left(1-\nu(h_2)\right)^\gamma p_{d}(\gamma)}
				{\sum_{h' \in [k]} \left( 1- \nu(h')\right)^\gamma}=\nu(h_1).
	\end{eqnarray*}
Hence, $\nu(h_1)=\nu(h_2)$ for all $h_1,h_2\in\brk k$, which implies~(\ref{eqLemma_balanced}).
Finally, the assertion follows from~(\ref{eqLemma_balanced}) and  the definition~(\ref{eqZgamma}) of $Z_\gamma(\pi)$.
\end{proof}

\begin{proof}[Proof of \Lem~\ref{Lemma_separateColors}]
If $\pi$ is a fixed point of $\cF_{d,k}$, then by \Lem~\ref{Cor_balanced_fix} and the definition~(\ref{eqBPOperator}) of the map $\cB$ we have
	\begin{align*}
	\pi_i&=\int_\Omega k\mu(i)\atom_\mu\dd\pi(\mu)=\int_\Omega k\mu(i)\atom_\mu\dd\cF_{d,k}[\pi](\mu)\\
		&=\sum_{\gamma=0}^\infty\frac{p_d(\gamma)}{Z_\gamma(\pi)}\int_{\Omega^\gamma}
			\brk{\sum_{h=1}^k \prod_{j=1}^{\gamma} 1 - \mu_j(h) }
				k\cB[\mu_1, \dots, \mu_\gamma](i)\atom_{\cB[\mu_1, \dots, \mu_\gamma]}
			\bigotimes_{j=1}^\gamma  \dd \pi(\mu_j)\\
		&=\sum_{\gamma=0}^\infty\frac{k^\gamma p_d(\gamma)}{(k-1)^\gamma}
				\int_{\Omega^\gamma}\brk{\prod_{j=1}^\gamma1-\mu_j(i)}
					\cdot\atom_{\cB[\mu_1, \dots, \mu_\gamma]}
						\bigotimes_{j=1}^\gamma  \dd \pi(\mu_j).
	\end{align*}
Further, for any $\mu\in\Omega$ we have $1-\mu(i)=\sum_{i'\neq i}\mu(i')$.
Hence,
	\begin{align}
	\pi_i&=\sum_{\gamma=0}^\infty\frac{k^\gamma p_d(\gamma)}{(k-1)^\gamma}
				\sum_{i_1,\ldots,i_\gamma\in\brk k\setminus\cbc i}\int_{\Omega^\gamma}
				\brk{\prod_{j=1}^\gamma\mu_j(i_j)}
					\cdot\atom_{\cB[\mu_1, \dots, \mu_\gamma]}
						\bigotimes_{j=1}^\gamma  \dd \pi(\mu_j)\nonumber\\
		&=
			\sum_{\gamma=0}^\infty\frac{p_d(\gamma)}{(k-1)^\gamma}
				\sum_{i_1,\ldots,i_\gamma\in\brk k\setminus\cbc i}
				\int_{\Omega^\gamma}\atom_{\cB[\mu_1, \dots, \mu_\gamma]}
						\bigotimes_{j=1}^\gamma  \dd \pi_{i_j}(\mu_j).
							\label{eqLastExpression}
	\end{align}
In the last expression, we can think of generating the sequence $i_1,\ldots,i_\gamma$ as follows:
	first, choose $\gamma$ from the Poisson distribution $\Po(d)$.
Then, choose the sequence $i_1,\ldots,i_\gamma$ by independently choosing
$i_j$ from the set $\brk k\setminus\cbc i$ uniformly at random.
Thus, in the overall experiment the number
of times that each color $h$ occurs has distribution $\Po(d/(k-1))$, independently for all $h\in\brk k\setminus\cbc i$,
whence (\ref{eqLastExpression}) implies the assertion.
\end{proof}

\begin{corollary}\label{Claim_hardFields1}
If $\pi$ is fixed point of $\cF_d$, then
$(\rho_i(\pi))_{i\in\brk k}$ is a fixed point of the function $F_{d,k}$ from \Lem~\ref{Lemma_GW}.
\end{corollary}
\begin{proof}
Invoking \Lem~\ref{Lemma_separateColors}, we obtain for any $i\in\brk k$
	\begin{align}
	\rho_i(\pi)&=\pi(\cbc{\atom_i})=\frac{\pi_i(\cbc{\atom_i})}k
		=\frac1k\sum_{\vec\gamma\in\Gamma_i}\int_{\Omega^{\vec \gamma}}
				\vecone_{\atom_i=\cB[\mu_{\vec \gamma}]}
				\measurei .
								\label{eqClaim_hardFields1_1}
	\end{align}
A glimpse at the definition~(\ref{eqBPOperator}) of $\cB$ reveals that $\atom_i=\cB[\mu_{\vec \gamma}]$
iff for each $h\in\brk k\setminus\cbc i$ there is $j\in\brk{\gamma_h}$ such that $\mu_{h,j}=\atom_h$.
Further, in~(\ref{eqClaim_hardFields1_1}) the $\mu_{h,j}$ are chosen independently from the distribution $\pi_h$,
and $\pi_h(\cbc{\atom_h})=k\rho_h(\pi)$.
In effect, the r.h.s.\ of~(\ref{eqClaim_hardFields1_1}) is simply the probability that
if we choose numbers $\gamma_h$ independently from the Poisson distribution with mean $d/(k-1)$ for $h\neq i$
and then perform $\gamma_h$ independent Bernoulli experiments with success probability $k\rho_h(\pi)$,
then there occurs at least one success for each $h\neq i$.
Of course, this is nothing but the probability that $k-1$ independent Poisson variables $(\Po(\rho_h(\pi)dk/(k-1)))_{h\neq i}$
are all strictly positive.
Hence,
	$$\rho_i(\pi)=\frac1k\prod_{h\in\brk k\setminus\cbc i}\pr[\Po(\rho_h(\pi)dk/(k-1))>0]
		=\frac1k\prod_{h\in\brk k\setminus\cbc i}1-\exp(-\rho_h(\pi)d')\quad\mbox{for any }i \in\brk k.$$
Consequently, $(\rho_i(\pi))_{i\in\brk k}=F_{d,k}((\rho_i(\pi))_{i\in\brk k})$.
\end{proof}

\begin{proof}[Proof of \Lem~\ref{Lemma_hardFields}]
Assume that $\pi\in\cP$ is a frozen fixed point of $\cF_{d,k}$.
Then $\rho_i(\pi)\geq\frac 2{3k}$ for all $i\in\brk k$.
Hence, \Cor~\ref{Claim_hardFields1} shows that $(\rho_1(\pi),\ldots,\rho_k(\pi))\in[\frac 2{3k},1]$ is a fixed point of $F_{d,k}$.
Therefore, \Lem~\ref{Lemma_GW} implies that $\rho_i(\pi)=q^*/k$ for all $i\in\brk k$.

To prove the second assertion, let $(i,\ell)\in\cT$.
Then \Lem~\ref{Lemma_separateColors} yields
	\begin{align}\label{eqLemma_hardFields1}
	\rho_{i,\ell}(\pi)&=
	\sum_{\vec \gamma\in\Gamma_i}\int_{\Omega^{\vec \gamma}}\vecone_{\cB[\mu_{\vec \gamma}]\in\Omega_\ell}
	\measurei .
	\end{align}
Now, the definition~(\ref{eqBPOperator}) is such that $\cB[\mu_{\vec \gamma}]\in\Omega_\ell$ iff
	\begin{enumerate}
	\item for each $h\in\brk k\setminus\ell$ there is
			$j\in\brk{\gamma_h}$ such that $\mu_{h,j}=\atom_h$, and
	\item for each $h\in\ell\setminus\cbc i$ and any $j\in\brk{\gamma_h}$ we have $\mu_{h,j}\neq\atom_h$.

	\end{enumerate}
Given $\vec\gamma$, the distributions $\mu_{h,j}$ are chosen independently from $\pi_h$ for all $h\neq i$, $j\in\brk{\gamma_h}$.
Hence, for a given $\vec\gamma$ the probability that (1) and (2) occur is precisely
	\begin{align}\nonumber
	\eta(\vec\gamma)&=\prod_{h\in\ell\setminus\cbc i}(1-\pi_h(\cbc{\atom_h}))^{\gamma_h}
			\cdot\prod_{h\in\brk k\setminus\ell}1-(1-\pi_h(\cbc{\atom_h}))^{\gamma_h}\\
			&=\prod_{h\in\ell\setminus\cbc i}(1-k\rho_h(\pi))^{\gamma_h}
				\cdot\prod_{h\in\brk k\setminus\ell}1-(1-k\rho_h(\pi))^{\gamma_h}.
			\label{eqLemma_hardFields2}
	\end{align}
Thus, combining~(\ref{eqLemma_hardFields1}) and~(\ref{eqLemma_hardFields2}), we see that
	\begin{align}
	\rho_{i,\ell}(\pi)&=
	\sum_{\vec\gamma\in\Gamma_i}\eta(\vec\gamma) p_i(\vec \gamma) 
		\nonumber\\
		&=\prod_{h\in\ell\setminus\cbc i}\brk{\sum_{\gamma_h\geq0}(1-k\rho_h(\pi))^{\gamma_h}p_{\frac d{k-1}}(\gamma_h)}
			\cdot\prod_{h\in\brk k\setminus\ell}
				\brk{\sum_{\gamma_h\geq0}(1-(1-k\rho_h(\pi))^{\gamma_h})p_{\frac d{k-1}}(\gamma_h)}\nonumber\\
		&=\prod_{h\in \ell \setminus \{i\}}\pr\brk{\Po(dk\rho_h(\pi)/(k-1)=0)}\prod_{h\in [k] \setminus \ell}\pr\brk{\Po(dk\rho_h(\pi)/(k-1)>0)}
			\nonumber\\
		&=\prod_{h\in\ell\setminus\cbc i}\exp(-d'\rho_h(\pi))
			\prod_{h\in [k] \setminus \ell}1-\exp(-d'\rho_h(\pi)).
							\label{eqLemma_hardFields3}
	\end{align}
Finally, as we already know from the first paragraph that $\rho_h(\pi)=q^*/k$, (\ref{eqLemma_hardFields3}) implies that $\rho_{i,\ell}(\pi)=kq_{i,\ell}^*$.
\end{proof}

\subsection{The fixed point}\label{Sec_softFields}
The objective in this section is to establish

\begin{lemma}\label{Lemma_softFields}
Suppose that $d\geq(2k-1)\ln k-2$.
Then $\pi_{d,k,\vec q^*}$ is the unique frozen fixed point of $\cF_{d,k}$.
\end{lemma}

To prove \Lem~\ref{Lemma_softFields},
let $\cP_\ell$ be the set of all probability measures $\pi\in\cP$ whose support is contained in $\Omega_\ell$
	(i.e., $\pi(\Omega_\ell)=1$).
For each $\pi\in\cP$ and any $(i,\ell)\in\cT$ we define a measure $\pi_{i,\ell}$ by letting
	$$\dd\pi_{i,\ell}(\mu)=\frac{\vecone_{\mu\in\Omega_\ell}}{kq_{i,\ell}^*}\dd\pi_i(\mu)
		=\frac{\mu(i)}{q_{i,\ell}^*}\vecone_{\mu\in\Omega_\ell}\dd\pi(\mu).$$
In addition, let
	$\widetilde\cP=\prod_{(i,\ell)\in\cT}\cP_\ell$
be the set of all families $(\pi_{i,\ell})_{i,\ell\in\cT}$ such that $\pi_{i,\ell}\in\cP_\ell$ for all $(i,\ell)$.

\begin{lemma}\label{Lemma_decomp}
If $\pi$ if a frozen fixed point of $\cF_{d,k}$, then $\widetilde\pi=(\pi_{i,\ell})_{(i,\ell)\in\cT}\in\widetilde\cP$.
\end{lemma}
\begin{proof}
Let $(i,\ell)\in\cT$.
By construction, the support of $\pi_{i,\ell}$ is contained in $\Omega_\ell$.
Furthermore, \Lem~\ref{Lemma_hardFields} implies that
	$$\pi_{i,\ell}(\Omega_\ell)=\frac1{kq_{i,\ell}^*}\int_\Omega \vecone_{\mu\in\Omega_\ell}\dd\pi_i(\mu)
		=\frac{\pi_i(\Omega_\ell)}{kq_{i,\ell}^*}=\frac{\rho_{i,\ell}(\pi)}{kq_{i,\ell}^*}=1.$$
Thus, $\pi_{i,\ell}$ is a probability measure.
\end{proof}		

Let $\Gamma_{i,\ell}$ be the set of all non-negative integer vectors $\vec\hgamma=(\hgamma_{i',\ell'})_{(i',\ell')\in\cT_{i,\ell}}$.
For $\vec \hgamma \in \Gamma_{i, \ell}$, we let
\begin{align*}
p_{i,\ell}(\vec \hgamma)&= \prod_{(i',\ell')\in \cT_{i,\ell}} p_{d'q^*_{i', \ell'}}(\hgamma_{i', \ell'})
\end{align*}
Moreover, we let $\Omega^{\vec \hgamma}=\prod_{(i',\ell')\in \cT_{i,\ell}}\prod_{j \in [\hgamma_{i',\ell'}]}\Omega$ and denote its points by
		$\mu_{\vec \hgamma} = (\mu_{i',\ell',j})_{(i',\ell')\in \cT_{i,\ell}, j \in [\hgamma_{i',\ell'}]}$.
In addition, if $\pi$ is a probability measure on $\Omega$ and $\vec\hgamma\in\Gamma_{i,\ell}$, we set
	$$\pi_{i,\ell,\vec\hgamma}=\bigotimes_{(i',\ell')\in\cT_{i,\ell}}\bigotimes_{j=1}^{\hgamma_{i',\ell'}}\pi_{i',\ell'}.$$

Further, we define for any non-empty set $\ell\subset\brk k$ a map
	\begin{eqnarray}\label{eqBell}
	\cB_{\ell}&:&\bigcup_{\gamma=1}^\infty\Omega^\gamma\ra\Omega,\quad
		(\mu_1,\ldots,\mu_\gamma)\mapsto\cB_{\ell}[\mu_1,\ldots,\mu_\gamma],\quad\mbox{where}\\
	\cB_{\ell}[\mu_1,\ldots,\mu_\gamma](h)&=&
		\begin{cases}
		\frac{\vecone_{h\in\ell}}{|\ell|}&\mbox{ if }\sum_{h'\in\ell}  \prod_{j=1}^\gamma 1 - \mu_j(h')=0,\\
		\frac{\vecone_{h\in\ell}\cdot \prod_{j=1}^\gamma 1 - \mu_j(h)}
			{\sum_{h'\in\ell}  \prod_{j=1}^\gamma 1 - \mu_j(h')}&\mbox{ if }\sum_{h'\in\ell}  \prod_{j=1}^\gamma 1 - \mu_j(h')>0.
		\end{cases}
			\nonumber
	\end{eqnarray}
Additionally, to cover the case $\gamma=0$ we define
	$\cB_\ell[\emptyset](h)=\frac{\vecone_{h\in\ell}}{|\ell|}$.
Thus, $\cB_\ell[\emptyset]$ is the uniform distribution on $\ell$.

\begin{lemma}\label{Lemma_fixDecomp}
Let $\cX$ be the set of all frozen fixed points of $\cF_{d,k}$.
Moreover, let $\widetilde\cX$ be the set of all fixed points of
	\begin{align*}
	\widetilde\cF_{d,k}&:\widetilde\cP\ra\widetilde\cP,&
	(\pi_{i,\ell})_{(i,\ell) \in \cT} &\mapsto
			\left( \sum_{\vec \hgamma \in \Gamma_{i, \ell}}
				\int_{\Omega^{\vec \hgamma}} \atom_{\cB_\ell[\mu_{\vec \hgamma}]}
				\measureii \right)_{(i,\ell)\in \cT}.
	\end{align*}
Then the map $\pi\in\cX\mapsto\widetilde\pi=(\pi_{i,\ell})_{(i,\ell)\in\cT}$ induces a bijection between $\cX$ and $\widetilde\cX$.
\end{lemma}

\begin{proof}
Suppose that $\pi\in\cX$.
Let $(i,\ell)\in\cT$. 
Then \Lem~\ref{Lemma_separateColors} yields
	\beq \begin{split} \label{eq_aux_pi_i_ell_1}
	\pi_{i,\ell}=\int_{\Omega_\ell}\frac{\atom_{\mu}}{kq_{i,\ell}^*}\dd\pi_i(\mu)
	&=\sum_{\vec\gamma\in\Gamma_{i}}\int_{\Omega^{\vec\gamma}}
				\frac{\vecone_{\cB[\mu_{\vec \gamma}]\in\Omega_\ell}\atom_{\cB[\mu_{\vec \gamma}]}}{kq_{i,\ell}^*}
				\measurei.
	\end{split} \eeq
Now let us fix a pair $(i, \ell) \in \cT$ and $(\vec \gamma, \mu_{\vec \gamma})$.
We denote, for $h\neq i$, by $\hgamma_{h}=\hgamma_{h}(\mu_{\vec\gamma})$
the number of occurence of $\delta_h$ in the tuple $\mu_{\vec \gamma}$.
The event $\cB[\mu_{\vec \gamma}]\in\Omega_\ell$ occurs iff
	\begin{enumerate}
	\item for each $h\in\brk k\setminus\ell$ there is $j\in\brk{\gamma_h}$ such that $\mu_{h,j}=\atom_h$, i.e. $\hgamma_{h} >0$,
	\item  for each  $h\in\ell\setminus\cbc i$ and all $j\in\brk{\gamma_h}$ we have $\mu_{h,j}=\atom_h$, i.e. $\hgamma_{h} =0$,
	\end{enumerate}
Thus, \Lem~\ref{Lemma_hardFields} implies that
	\begin{align}\label{eq_aux_pi_i_ell_1a}
	\sum_{\vec\gamma\in\Gamma_i}\int_{\Omega^{\vec\gamma}}\frac{\vecone_{\cB[\mu_{\vec \gamma}]\in\Omega_\ell}}{kq_{i,\ell}^*}
				\measurei&=
		 \frac{1}{kq_{i,\ell}^*}
		 \prod_{h\in\brk k\setminus\ell}\pr\brk{\Po(q_h^*d')>0} \prod_{h\in\ell\setminus\cbc i}\pr\brk{\Po(q_h^*d')=0}=1.
	\end{align}
Furthermore, given that the event $\cB[\mu_{\vec \gamma}]\in\Omega_\ell$ occurs,
the measure $\cB[\mu_{\vec \gamma}]$ is determined by those components $\mu_{i',\ell',j}$
with $(i',\ell')\in\cT_{i,\ell}$ only.
Thus, with $\vec \hgamma = (\hgamma_{i', \ell'})_{(i',\ell') \in \cT_{i, \ell}}$ and
$\mu_{\vec \hgamma}=(\mu_{i',\ell',j})_{(i',\ell')\in\cT_{i,\ell},j\in[\hgamma_{i',\ell'}]}$
we obtain from~(\ref{eq_aux_pi_i_ell_1}) and~(\ref{eq_aux_pi_i_ell_1a})
	\begin{align*}
	\pi_{i,\ell}&=
		\sum_{\vec\hgamma\in\Gamma_{i,\ell}}\int_{\Omega^{\vec\hgamma}}\atom_{\cB_\ell[\mu_{\vec \hgamma}]} \measureii .
	\end{align*}
Thus, if $\pi$ is a frozen fixed point of $\cF_{d,k}$, then $\widetilde\pi$ is a fixed point of $\widetilde\cF_{d,k}$.

Conversely, if $\widetilde\pi=(\pi_{i,\ell})$ is a fixed point of $\widetilde\cF_{d,k}$, then the measure $\pi$ defined by
	$$\dd \pi(\mu)= \sum_{\ell \subset [k]}  \frac{1}{|\ell|} \sum_{i \in \ell} \frac{q_{i,\ell}^*}{\mu(i)} \dd\pi_{i,\ell}(\mu)$$
is easily verified to be a fixed point of $\cF_{d,k}$.
Moreover, for $i \in [k]$, $\rho_i(\pi) = q_{i, \{i\}}^* = q^*/k \geq 2/(3k)$ and $\pi$ is thus a frozen fixed point of $\cF_{d,k}$.
 \end{proof}

\begin{corollary}\label{Cor_tildeFixedPoint}
The distribution $\pi_{d,k,\vec q^*}$ is a fixed point of $\cF_{d,k}$.
\end{corollary}
\begin{proof}
To unclutter the notation we write $\pi=\pi_{d,k,q^*}$.
Moreover, we let $\T=\T_{d,k,\vec q^*}$;
	by \Lem~\ref{Lemma_GW} we may always assume that $\T$ is a finite tree.
Recall that $\pi$ is the distribution of $\mu_{\T}$, which is the distribution of the color of the root
under a random legal coloring of $\T$.
In light of \Lem~\ref{Lemma_fixDecomp} it suffices to show that $\widetilde\pi=(\pi_{i,\ell})$ is a fixed point of $\widetilde\cF_{d,k}$.
Thus, we need to show that for all $(i,\ell)\in\cT$,
	\begin{align}\label{eqCor_tildeFixedPoint1}
	\pi_{i,\ell}=	\sum_{\vec\gamma\in\Gamma_{i,\ell}}
				\int_{\Omega^{\vec \gamma}} \atom_{\cB_\ell[(\mu_{i',\ell'}^{(j)})]}
					\prod_{(i',\ell')\in\cT_{i,\ell}}p_{d'q_{i',\ell'}^*}(\gamma_{i',\ell'})
					\bigotimes_{j=1}^{\gamma_{i', \ell'}} \dd \pi_{i', \ell'}(\mu_{i', \ell'}^{(j)}). 
	\end{align}

Let us denote by $\T_{i,\ell}$ the random tree $\T$ given that the root has type $(i,\ell)$.
We claim that $\pi_{i,\ell}$ is the distribution of $\mu_{\T_{i,\ell}}$.
Indeed, let $\ell\subset\brk k$.
If the root $v_0$ of $\T$ has type $(i,\ell)$ for some $i\in\ell$, then 
	the support of the measure $\mu_{\T}$ is contained in $\ell$ (because under any legal coloring, $v_0$ receives a color from $\ell$).
Moreover, all children of $v_0$ have types in $\cT_{i,\ell}$, and if $(i',\ell')\in\cT_{i,\ell}$, then $|\ell'|\geq2$.
Hence, inductively we see that if $v_0$ has type $(i,\ell)$, then for any color $h\in\ell$ there is a legal coloring under which $v_0$ receives color $h$.
	Consequently, the support of $\mu_{\T}$ is precisely $\ell$.
Furthermore, the distribution $\mu_{\T}$ is invariant under the following operation:
	obtain a random tree $\T'$ by choosing a legal color $\vec\tau$ of $\T$ randomly
	and then changing the types $\thet(v)=(i_v,\ell_v)$ of  the vertices to $\thet'(v)=(\vec\tau(i_v),\ell_v)$;
		this is because the trees $\T$ and $\T'$ have the same set of legal colorings. 
These observation imply that for any measurable set $A$ we have
	\begin{align*}
	\pr\brk{\mu_{\T}\in A|\thet(v_0)=(i,\ell)}&=
		\frac{\pr\brk{\mu_{\T}\in A,\thet(v_0)=(i,\ell)}}{\pr\brk{\thet(v_0)=(i,\ell)}}=
		\frac{\pr\brk{\mu_{\T}\in A\cap\Omega_\ell,\thet(v_0)=(i,\ell)}}{q_{i,\ell}^*}\\
		&=\frac1{q_{i,\ell}^*}\int_A\mu(i)\vecone_{\mu\in\Omega_\ell}\dd\pi(\mu)=\pi_{i,\ell}(A).
	\end{align*}

To prove that $\widetilde\pi$ is a fixed point of $\widetilde\cF_{d,k}$, we observe that the random tree $\T_{i,\ell}$ can be described by
	the following recurrence.
	There is a root of $v_0$ of type $(i,\ell)$.
	For each $(i',\ell')$, $v_0$ has a random number $\gamma_{i',\ell'}=\Po(d'q_{i,\ell}^*)$ of children 
		$(v_{i',\ell',j})_{j=1,\ldots,\gamma_{i',\ell'}}$ of type $(i',\ell')$.
	Moreover, each $v_{i',\ell',j}$ is the root of a random tree $\T_{i',\ell',j}$.
	Of course, the random variables $(\gamma_{i',\ell'})_{(i',\ell')\in\cT_{i,\ell}}$ and
	the random trees $\T_{i',\ell',j}$ are chosen independently.

This recursive description of the random tree $\T_{i,\ell}$ leads to a recurrence for the distribution $\pi_{i,\ell}$.
Indeed, given the numbers $(\gamma_{i',\ell'})_{i',\ell'}$,
	the distribution $\mu_{\T_{i',\ell',j}}$ of the color of the root of the random tree $\T_{i',\ell',j}$
	is an $\Omega_{\ell'}$-valued random variable with distribution $\pi_{i',\ell'}$ for each $j=1,\ldots,\gamma_{i',\ell'}$.
Moreover, the random variables $(\mu_{\T_{i',\ell',j}})_{i',\ell',j}$ are mutually independent.
In addition, we claim that given the distributions $(\mu_{\T_{i',\ell',j}})_{i',\ell',j}$, the color of the root $v_0$ of the entire tree $\T_{i,\ell}$
has distribution
	\begin{equation}\label{eqBPvindication}
	\mu_{\T_{i,\ell}}=\cB_\ell[(\mu_{\T_{i',\ell',j}})_{i',\ell',j}].
	\end{equation}
Indeed, given that $v_0$ has type $(i,\ell)$, $v_0$ receives a color from $\ell$ under any legal coloring.
Further, for any $h\in\ell$ the probability that $v_0$ takes color $h$ under a random coloring of $\T_{i,\ell}$
is proportional to the probability that none of its children $v_{i',\ell',j}$ takes color $h$ in a random coloring of the tree
$\T_{i',\ell',j}$ whose root $v_{i',\ell',j}$ is.

Finally, we recall that $\pi_{i,\ell}$ is the distribution of $\mu_{\T_{i,\ell}}$.
Hence, (\ref{eqBPvindication}) implies together with the fact that the $\gamma_{i',\ell',j}$ are independent Poisson variables that
	$\pi_{i,\ell}$ satisfies~(\ref{eqCor_tildeFixedPoint1}).
\end{proof}

\begin{lemma} \label{lemma_wtcF_unicity}
The map $\widetilde\cF_{d,k}$ has at most one fixed point.
\end{lemma}
\begin{proof}
As before, we let $\T$ denote the random tree $\T_{d,k,\vec q^*}$.
Moreover, $\T_{i,\ell}$ is the random tree $\T$ given that the root has type $(i,\ell)$.

Let $t\geq0$ be an integer and let $\widetilde\pi=(\pi_{i,\ell})\in\widetilde\cP$.
We define a distribution $\widetilde\pi_t=(\pi_{i,\ell,t})\in\widetilde\cP$ by means of the following experiment.
Let $(i,\ell)\in\cT$.
Let $v_0$ denote the root of $\T_{i,\ell}$ and let $\thet(v)$ signifiy the type of each vertex $v$.
	\begin{description}
	\item[TR1] Let $\T_{i,\ell,t}$ be the tree obtained from $\T_{i,\ell}$ by deleting all vertices
			at distance greater than $t$ from $v_0$.
	\item[TR2] Let $V_t$ be the set of all vertices at distance exactly $t$ from $v_0$.
			For each $v\in V_t$ independently, choose $\mu_v\in\Omega$
			from the distribution $\pi_{\thet(v)}$.
	\item[TR3] Let $\mu_{i,\ell,t}$ be the
			distribution of the color of $v_0$ under a random coloring $\vec\tau$ chosen as follows.
		\begin{itemize}
		\item Independently for each vertex $v\in V_t$ choose a color $\vec\tau_t(v)$ from the distribution $\mu_v$.
		\item Let $\vec\tau$ be a uniformly random legal coloring of $\T_{i,\ell,t}$
				such that $\vec\tau(v)=\vec\tau_t(v)$ for all $v\in V_t$;
					if there is no such coloring, discard the experiment. % !!!
		\end{itemize}
	\end{description}
Step {\bf TR3} of the above experiment yields a {\em distribution} $\mu_{i,\ell,t}\in\Omega$.
Clearly $\mu_{i,\ell,t}$ is determined by the random choices in steps {\bf TR1}--{\bf TR2}.
Thus, let we let $\pi_{i,\ell,t}$ be the distribution of $\mu_{i,\ell,t}$ with respect to {\bf TR1}--{\bf TR2}.

We now claim that for any integer $t\geq0$ the following is true.
	\begin{equation}\label{equniqueFixedPoint1}
	\mbox{If $\widetilde\pi$ is a fixed point of $\widetilde\cF_{d,k}$, then $\widetilde\pi=\widetilde\pi_t$.}
	\end{equation}
The proof of~(\ref{equniqueFixedPoint1}) is by induction on $t$.
It is immediate from the construction that $\pi_{i,\ell,0}=\pi_{i,\ell}$ for all $(i,\ell)\in\cT$.
Thus, assume that $t\geq1$.
By induction, it suffices to show that $\widetilde\pi_t=\widetilde\pi_{t-1}$.
To this end, let us condition on the random tree $\T_{i,\ell,t-1}$.
Consider a vertex $v\in V_{t-1}$ of type $\thet(v)=(i_v,\ell_v)$.
We obtain the random tree $\T_{i,\ell,t}$ from $\T_{i,\ell,t-1}$ by attaching to each such $v\in V_{t-1}$
a random number $\gamma_{i',\ell',v}=\Po(d'q_{i',\ell'}^*)$ of children of each type $(i',\ell')\in\cT_{i_v,\ell_v}$ where,
	of course, the random variables $\gamma_{i',\ell',v}$ are mutually independent.
Further, in step {\bf TR2} of the above experiment we choose $\mu_{i',\ell',v,j}\in\Omega_{\ell'}$ independently from $\pi_{i',\ell'}$
for each $v\in V_{t-1}$, $(i',\ell')\in\cT_{i,\ell}$ and $j=1,\ldots,\gamma_{i',\ell',v}$.

Given the distributions $\mu_{i',\ell',v,j}$, suppose that we choose a legal coloring $\vec\tau_v$ of the sub-tree consisting of $v\in V_{t-1}$
and its children only from the following distribution.
	\begin{itemize}
	\item  Independently choose the colors $\vec\tau_v(u_{i',\ell',j})$ of the children $u_{i',\ell',j}$ of $v$ of type $(i',\ell')$ from
		 $\mu_{i',\ell',j}$.
	\item Choose a color $\vec\tau_v(v)$ for $v$ uniformly from the set of all
			colors $h\in\ell$ that are not already assigned to a child of $v$ if possible. % !!!
	\end{itemize}
Let $\mu_v$ denote the distribution of the color $\vec\tau_v(v)$.
Then by construction,
	$$\mu_v=\cB_\ell[(\mu_{i',\ell',j})_{(i',\ell')\in\cT_{i,\ell},j\in[\gamma_{i',\ell',v}]}].$$
Hence, the distribution of $\mu_v$ with respect to the choice of the numbers $\gamma_{i',\ell',v}$ and the distributions $\mu_{i',\ell',j}$ is given by
	$$\sum_{\vec\gamma\in\Gamma_{i,\ell}}
				\int_{\Omega^{\vec \gamma}} \atom_{\cB_\ell[(\mu_{i',\ell'}^{(j)})]}
					\prod_{(i',\ell')\in\cT_{i,\ell}}p_{d'q_{i',\ell'}^*}(\gamma_{i',\ell'})
					\bigotimes_{j=1}^{\gamma_{i', \ell'}} \dd \pi_{i', \ell'}(\mu_{i', \ell'}^{(j)})=\pi_{i,\ell},$$
because $\widetilde\pi$ is a fixed point of $\widetilde\cF_{d,k}$.
Therefore, the experiment of first choosing $\T_{i,\ell,t}$, then choosing distributions $\mu_u$ independently from $\pi_{\thet(u)}$ for
the vertices at distance $t$, and then choosing a random legal coloring $\vec\tau$ as in {\bf TR3} is equivalent
to performing the same experiment with $t-1$ instead.
Hence, $\widetilde\pi_t=\widetilde\pi_{t-1}$.

To complete the proof, assume that $\widetilde\pi,\widetilde\pi'$ are fixed points of $\widetilde\cF_{d,k}$.
Then for any integer $t\geq0$ we have $\widetilde\pi=\widetilde\pi_t$, $\widetilde\pi'=\widetilde\pi'_t$.
Furthermore, as $\widetilde\pi_t$, $\widetilde\pi_t'$ result from the experiment {\bf TR1--TR3}, whose first step
	{\bf TR1} can be coupled, we see that for any $(i,\ell)\in\cT$,
	\begin{equation}\label{eqSubCriticalNorm}
	\norm{\pi_{i,\ell}-\pi_{i,\ell}'}_{\mathrm{TV}}=\norm{\pi_{i,\ell,t}-\pi_{i,\ell,t}'}_{\mathrm{TV}}\leq 2\pr\brk{|\T_{i,\ell}|\geq t}.
	\end{equation}
Because \Lem~\ref{Lemma_GW} shows that $\T$ results from a sub-critical branching process, we have $$\lim_{t\ra\infty}\pr\brk{|\T_{i,\ell}|\geq t}=0$$
for any $(i,\ell)\in\cT$.
Consequently, (\ref{eqSubCriticalNorm}) shows that $\widetilde\pi=\widetilde\pi'$.
\end{proof}

\noindent
Finally, \Lem~\ref{Lemma_softFields}
follows directly from \Lem~\ref{Lemma_fixDecomp}, \Cor~\ref{Cor_tildeFixedPoint} and \Lem~\ref{lemma_wtcF_unicity}.

\subsection{The number of legal colorings}

The final step of the proof of \Prop~\ref{Prop_fix} is to relate $\FF_{d,k}(\pi_{d,k,\vec q^*})$ to the number
of legal colorings of $\T_{d,k,\vec q^*}$.
The starting point for this is a formula for the (logarithm of the) number of legal colorings of a decorated tree $T,\thet$.
To write this formula down, we recall the map $\cB_\ell$ from~(\ref{eqBell}).
Moreover, suppose that $\ell \subset\brk k$ and $\mu_1,\ldots,\mu_\gamma\in\Omega$
are such that:
	\begin{equation}\label{eqFellCond}
	\exists h \in \ell\ \forall j \in [\gamma]:\mu_j(h)<1.
	 \end{equation}
Then we let 
	\begin{eqnarray*}
	\FF_{\ell}(\mu_1,\ldots,\mu_\gamma)&=&
		\FF^v_{\ell}(\mu_1,\ldots,\mu_\gamma)-\frac12\FF^e_{\ell}(\mu_1,\ldots,\mu_\gamma),
			\quad\mbox{where}\\
	\FF^v_{\ell}(\mu_1,\ldots,\mu_\gamma)&=&
		\ln \sum_{h \in \ell} \prod_{j=1}^{\gamma}1-\mu_{j}(h),\\
	\FF^e_{\ell}(\mu_1,\ldots,\mu_\gamma)&=&
		\sum_{j=1}^\gamma
			\ln \brk{1- \sum_{h \in \ell} \mu_{j}(h)\cB_\ell[\mu_1,\ldots,\mu_{j-1},\mu_{j+1},\ldots,\mu_\gamma](h)};
	\end{eqnarray*}
the condition~(\ref{eqFellCond}) ensures that these quantities are well-defined (i.e., the argument of the logarithm is positive in both instances).
Additionally, to cover the case $\gamma=0$ we set
	$\FF_{\ell}(\emptyset)=\ln|\ell|.$

Further, suppose that $T,\thet,v$ is a rooted decorated tree
	that has at least one legal coloring $\sigma$.
Let $v_1,\ldots,v_\gamma$ be the neighbors of the root vertex $v$ and suppose that $\thet(v)=(i,\ell)$ and
	$\thet(v_j)=(i_j,\ell_j)$ for $j=1,\ldots,\gamma$.
If we remove the root $v$ from $T$, then each of the vertices $v_1,\ldots,v_\gamma$ lies in a connected component $T_i$
of the resulting forest.
By considering the restrictions $\thet_i$ of $\thet$ to the vertex set of $T_i$, we obtain decorated trees $T_i,\thet_i$.
Recall that $\mu_{T_j,\thet_j,v_j}$ denotes the distribution of the color of the root in a random legal coloring of $T_j,\thet_j,v_j$.
Since $\sigma$ is a legal coloring, for $h=\sigma(v)$ for all $j\in[\gamma]$ we have $\mu_{T_j,\thet_j,v_j}<1$.
Thus, we can define
	$$\FF(T,\thet,v)=\FF_{\ell}(\mu_{T_1,\thet_1,v_1},\ldots,\mu_{T_\gamma,\thet_\gamma,v_\gamma}).$$

\begin{fact}\label{Fact_BetheFreeEnergy}
Let $T,\thet$ be a decorated tree 
such that $\cZ(T,\thet)\geq1$.
Then
	$\ln \cZ(T,\thet) = \sum_{v \in V(T)} \FF (T,\thet,v).$
\end{fact}
\begin{proof}
This follows from \cite[\Prop~3.7]{DeMo09}.
More specifically, let $(i_v,\ell_v)=\thet(v)$ be the type of vertex $v$.
In the terminology of~\cite{DeMo09} (and of the physicists ``cavity method''),
$\FF(T,\thet,v)$ is the {\em Bethe free entropy} of the Boltzmann distribution
	$$\nu:\brk k^{V(T)}\ra[0,1],\qquad
	\nu(\tau)=\frac1{\cZ(T,\thet)}\prod_{v\in V(T)}\vecone_{\tau(v)\in\ell_v}\cdot
		\prod_{e=\cbc{u,w}\in E(T)}\vecone_{\tau(u)\neq\tau(w)}.$$
Thus, $\nu$ is simply the uniform distribution over legal $k$-colorings of $T,\thet$, and $\cZ(T,\thet)$ is its partition function.
\end{proof}

Let $\T$ denote the random rooted decorated tree $\T_{d,k,\vec q^*}$.
Moreover, for $(i,\ell)\in\cT$ we let $\T_{i,\ell}$ denote the random tree $\T$ given that the root has type $(i,\ell)$.
The starting point of the proof is the following key observation.
Furthermore, if $(T,\thet,v)$ is a rooted decorated tree, then we let $(T,\thet,v)^{\star}$ signify
the isomorphism class of the random rooted decorated tree $(T,\thet,u)$ obtain from $(T,\thet,v)$ by choosing a vertex $u$ of $T$ uniformly
at random and rooting the tree at $u$.
In other words, $(T,\thet,v)^{\star}$ is obtained by re-rooting $(T,\thet,v)$ at random vertex.

\begin{lemma}\label{Lemma_unimodular}
Let $\T^\star$ be the random rooted decorated tree obtained by re-rooting $\T$ at a random vertex.
Then the distribution of $\T^\star$ coincides with the distribution of $\T$.
\end{lemma}
\begin{proof}
This follows from the general fact that Galton-Watson trees are unimodular in the sense of~\cite{BordenaveCaputo}.
\end{proof}

\begin{corollary}\label{Cor_unimodular}
We have $\Erw\brk{\frac{\ln\cZ(\T)}{|\T|}}=\Erw[\FF(\T)]$.
\end{corollary}
\begin{proof}
Letting $(T,\thet,v)$ range over rooted decorated trees, we find
	\begin{align*}
	\Erw\brk{\frac{\ln\cZ(\T)}{|\T|}}&=\sum_{(T,\thet,v)}\pr\brk{\T\ism(T,\thet,v)}\cdot\frac{\ln\cZ(T,\thet,v)}{|V(T)|}\\
		&=\sum_{(T,\thet,v)}\sum_{u\in V(T)}\frac{\pr\brk{\T\ism(T,\thet,v)}\FF(T,\thet,u)}{|V(T)|}&[\mbox{by Fact~\ref{Fact_BetheFreeEnergy}}]\\
		&=\sum_{(T,\thet,v)}\sum_{u\in V(T)}\frac{\pr\brk{\T\ism(T,\thet,u)}\FF(T,\thet,u)}{|V(T)|}&[\mbox{by \Lem~\ref{Lemma_unimodular}}]\\
		&=\sum_{(T,\thet,v)}\pr\brk{\T\ism(T,\thet,v)}\FF(T,\thet,v)=\Erw[\FF(\T)],
	\end{align*}
as claimed.
\end{proof}

\begin{lemma}\label{Lemma_FFqil}
We have
	$$\Erw[\FF(\T_{i,\ell})]=\sum_{\vec\gamma\in\Gamma_{i,\ell}}p_{i,\ell}(\vec\gamma)
			\int_{\Omega^{\vec\gamma}}\FF_\ell^v(\mu_{\vec\gamma})
				\dd\pi_{\vec\gamma}(\mu_{\vec\gamma})
					-\sum_{(\hat i,\hat\ell)\in\cT_{i,\ell}}\frac{q_{\hat i,\hat\ell}d'}2
						\int_{\Omega^2}\ln\brk{1-\sum_{h=1}^k\hat\mu(h)\mu(h)}
							\dd\pi_{i,\ell}(\mu)\tensor\pi_{\hat i,\hat \ell}(\hat\mu).$$
\end{lemma}
\begin{proof}
Writing $\pi=\pi_{d,k,\vec q^*}$ for the distribution of $\mu_{\T}$, we know from \Cor~\ref{Cor_tildeFixedPoint}
that $\pi_{i,\ell}$ is the distribution of $\mu_{\T_{i,\ell}}$ for any type $(i,\ell)$.
Furthermore, the distribution of $\T_{i,\ell}$ can be described by the following recurrence:
	there is a root $v_0$ of type $(i,\ell)$, to which we attach for each $(i',\ell')\in\cT_{i,\ell}$
		independently a number $\gamma_{i',\ell'}=\Po(d'q_{i',\ell'}^*)$
	of trees $(T_{i',\ell',j})_{j=1,\ldots,\gamma_{i',\ell'}}$ that are chosen independently from the distribution $\T_{i',\ell'}$.
By independence, the distribution of the color of the root of each $T_{i',\ell',j}$ is just an independent sample from the distribution $\pi_{i',\ell'}$.
Therefore, we obtain the expansion
	\begin{align*}
	\Erw[\FF(\T_{i,\ell})]&=\sum_{\vec\gamma\in\Gamma_{i,\ell}}\int_{\Omega^{\vec\gamma}}\FF_\ell(\mu_{\vec\gamma})
			p_{i,\ell}(\vec\gamma)\dd\pi_{\vec\gamma}(\mu_{\vec\gamma}).
	\end{align*}
Substituting in the
definition of $\FF_\ell$, we obtain
	\begin{align*}
	\Erw[\FF(\T_{i,\ell})]&=I_{i,\ell}-\frac12J_{i,\ell},&\mbox{where}\\
	I_{i,\ell}&=\sum_{\vec\gamma\in\Gamma_{i,\ell}}p_{i,\ell}(\vec\gamma)
			\int_{\Omega^{\vec\gamma}} \FF_\ell^v(\mu_{\vec\gamma})
				\dd\pi_{\vec\gamma}(\mu_{\vec\gamma}),\quad
	J_{i,\ell}=\sum_{\vec\gamma\in\Gamma_{i,\ell}}p_{i,\ell}(\vec\gamma)
				\int_{\Omega^{\vec\gamma}} \FF_\ell^e(\mu_{\vec\gamma})
					\dd\pi_{\vec\gamma}(\mu_{\vec\gamma}).
	\end{align*}
Further, by the definition of $\FF_\ell^e$ we have
	\begin{align*}
	J_{i,\ell}
		&=\sum_{\vec\gamma\in\Gamma_{i,\ell}}p_{i,\ell}(\vec\gamma)
			\sum_{(\hat i,\hat\ell)\in\cT_{i,\ell}}\sum_{\hat j=1}^{\gamma_{\hat i,\hat\ell}}
				\int_{\Omega^{\vec\gamma}} \ln\brk{1-\sum_{h\in\ell}\mu_{\hat i,\hat \ell,\hat j}(h)\cB[(\mu_{i',\ell',j})_{(i',\ell',j)\neq(\hat i,\hat\ell,\hat j)}](h)}
				\dd\pi_{\vec\gamma}(\mu_{\vec\gamma})
		\\&=\sum_{(\hat i,\hat\ell)\in\cT_{i,\ell}} \sum_{g \geq 1} p_{q_{\hat i, \hat \ell}^*}(g) \sum_{\hat j=1}^g \sum_{\vec\gamma\in\Gamma_{i,\ell}} p_{i,\ell}(\vec\gamma) \vecone_{\gamma_{\hat i, \hat \ell} = g} 
			\\& \quad\quad\quad\quad\quad\quad	\int_{\Omega \times \Omega^{\vec\gamma}} \ln\brk{1-\sum_{h\in\ell}\mu(h)\cB[(\mu_{i',\ell',j})_{(i',\ell',j)\neq(\hat i,\hat\ell,1)}](h)}
				\dd \pi_{\hat i,\hat \ell}(\mu) \tensor \dd\pi_{\vec\gamma}(\mu_{\vec\gamma})
		\\&=\sum_{(\hat i,\hat\ell)\in\cT_{i,\ell}} \sum_{g \geq 1} p_{q_{\hat i, \hat \ell}^* d'}(g) \sum_{\hat j=1}^g \sum_{\vec\gamma\in\Gamma_{i,\ell}} p_{i,\ell}(\vec\gamma) \vecone_{\gamma_{\hat i, \hat \ell} = g-1} 
				\\
				&\qquad\qquad\int_{\Omega \times \Omega^{\vec\gamma}} \ln\brk{1-\sum_{h\in\ell}\mu(h)\cB[\mu_{\vec \gamma}](h)}
				\dd \pi_{\hat i,\hat \ell}(\mu) \tensor \dd\pi_{\vec\gamma}(\mu_{\vec\gamma}).
	\end{align*}
To simplify this, we use the following elementary relation:
	if $X:\ZZ\ra\RRpos$ is a function and $g$ is a Poisson random variable,
	then $\Erw[\vecone_{g \geq 1}g X(g-1)]=\Erw[g] \Erw[X(g)]$.
Applying this observation to 
$$ X(g) =  \sum_{\vec \gamma \in \Gamma_{i,\ell}} p_{i,\ell}(\vec \gamma) \vecone_{\gamma_{i,\ell} = g-1} 
				\int_{\Omega \times \Omega^{\vec\gamma}}\ln\brk{1-\sum_{h\in\ell}\mu(h)\cB[\mu_{\vec \gamma}](h)}
				\dd \pi_{\hat i,\hat \ell}(\mu) \tensor \dd\pi_{\vec\gamma}(\mu_{\vec\gamma}), $$
we obtain
	\begin{align*}
	J_{i,\ell}	&=\sum_{(\hat i,\hat\ell)\in\cT_{i,\ell}}q_{\hat i,\hat\ell}d'
				\sum_{\vec\gamma\in\Gamma_{i,\ell}}p_{i,\ell}(\vec\gamma)
				\int_{\Omega \times \Omega^{\vec\gamma}} \ln\brk{1-\sum_{h\in\ell}\mu_{\hat i,\hat \ell}(h)\cB[\mu_{\vec\gamma}](h)}
					\dd\pi_{\hat i,\hat\ell}(\mu) \tensor \dd\pi_{\vec\gamma}(\mu_{\vec\gamma}).
	\end{align*}
Now, since $\pi$ is a fixed point of $\cF_{d,k}$, the distribution of the measure $\cB[\mu_{\vec \gamma}]$ is just $\pi_{i,\ell}$.
Hence,
	$$J_{i,\ell}=\sum_{(\hat i,\hat\ell)\in\cT_{i,\ell}}q_{\hat i,\hat\ell}d'\int_{\Omega^2} \ln\brk{1-\sum_{h\in\ell}\hat\mu(h)\mu(h)}
		\dd\pi_{i,\ell}(\mu)\tensor \dd \pi_{\hat i,\hat \ell}(\hat\mu).$$
Thus, we obtain the assertion.
\end{proof}

\begin{lemma}\label{Lemma_Bethe}
We have
	$\Erw\brk{\FF(\T_{d,k,\vec q^*})}=\FF_{d,k}(\pi_{d,k,q^*}).$
\end{lemma}
\begin{proof}
Summing over all $(i,\ell)\in\cT$, we obtain from \Lem~\ref{Lemma_FFqil} that
	\begin{align*}
	\Erw[\FF(\T)]&=I-\frac12J,&\mbox{where}\\
		I&=\sum_{(i,\ell)\in\cT}q_{i,\ell}^*\sum_{\vec\gamma\in\Gamma_{i,\ell}}p_{i,\ell}(\vec\gamma)
			\int_{\Omega^{\vec\gamma}} \FF_\ell^v(\mu_{\vec\gamma})
				\dd\pi_{\vec\gamma}(\mu_{\vec\gamma}),\\
		J&=d'\sum_{(i,\ell)\in\cT}\sum_{(\hat i,\hat\ell)\in\cT_{i,\ell}}
			q_{i,\ell}^*q_{\hat i,\hat\ell}^*\int_{\Omega^2} \ln\brk{1-\sum_{h=1}^k\hat\mu(h)\mu(h)}
						\dd\pi_{i,\ell}(\mu)\tensor \pi_{\hat i,\hat \ell}(\hat\mu).
	\end{align*}
Recalling that $\dd\pi_{i,\ell}(\mu)=\frac{\vecone_{\mu\in\Omega_\ell}}{kq_{i,\ell}^*}\dd\pi_{i}(\mu)$ and
	$\dd\pi_{\hat i,\hat \ell}(\hat \mu)=\frac{\vecone_{\hat\mu\in\Omega_{\hat\ell}}}{kq_{\hat i,\hat \ell}^*}\dd\pi_{\hat i}(\hat \mu)$, we get
	\begin{align*}
	J&=\frac{d'}{k^2}\sum_{(i,\ell)\in\cT}\sum_{(\hat i,\hat\ell)\in\cT_{i,\ell}}
			\int_{\Omega^2} \ln\brk{1-\sum_{h=1}^k\hat\mu(h)\mu(h)}\vecone_{\mu\in\Omega_\ell}\vecone_{\hat\mu\in\Omega_{\hat\ell}}
						\dd\pi_{i}(\mu)\tensor \pi_{\hat i}(\hat\mu)\\
	&=\frac{d}{k(k-1)}\sum_{i,\hat i\in\brk k:i\neq\hat i}\int_{\Omega^2} \sum_{\ell:(i,\ell)\in\cT}\sum_{\hat\ell:(\hat i,\hat \ell)\in\cT}
		\ln\brk{1-\sum_{h=1}^k\hat\mu(h)\mu(h)}\vecone_{\mu\in\Omega_\ell}\vecone_{\hat\mu\in\Omega_{\hat\ell}}
						\dd\pi_{i}(\mu)\tensor \pi_{\hat i}(\hat\mu)\\
	&=\frac{d}{k(k-1)}\sum_{i,\hat i\in\brk k:i\neq\hat i}
		\int_{\Omega^2} \ln\brk{1-\sum_{h=1}^k\hat\mu(h)\mu(h)}\dd\pi_{i}(\mu)\tensor\pi_{\hat i}(\hat\mu)
			=\FF^e_{d,k}(\pi).
	\end{align*}

It finally remains to simplify the expression for $I$. To do it, we introduce $\cT_i = \{(i',\ell') \in \cT, i' \neq i\}$. We let $\overline{\Gamma}_i$ be the set of non-negative vectors $\vec \ogamma = (\overline{\gamma}_{i',\ell'})_{(i',\ell')\in \cT_i}$. Moreover, we let $\Omega^{\vec \ogamma}=\prod_{(i',\ell')\in \cT_{i}}\prod_{j \in [\ogamma_{i',\ell'}]}\Omega$ and denote its points by
		$\mu_{\vec \ogamma} = (\mu_{i',\ell',j})_{(i',\ell')\in \cT_i, j \in [\ogamma_{i',\ell'}]}$. We note that if $\vec \gamma \in \Gamma_{i,\ell}$ and $\vec \ogamma \in \overline{\Gamma}_i$ are such that: \begin{itemize}
\item[(a)] $\forall i' \in \ell \setminus \{i\}, \ogamma_{i',\{i'\}} =0$,
\item[(b)] $\forall i' \in [k] \setminus \ell, \ogamma_{i',\{i'\}} >0$
\item[(c)] $\forall (i', \ell') \in \cT_{i, \ell}, \gamma_{i',\ell'} = \ogamma_{i',\ell'}$,
\end{itemize}
and that $\mu_{\vec \gamma}, \overline{\mu}_{\vec \ogamma}$ satisfy \begin{itemize}
\item[(d)] $\forall (i',\ell') \in \cT_{i, \ell}, \forall j \in [\gamma_{i', \ell'}], \mu_{i',\ell', j} = \overline{\mu}_{i',\ell',j}$,
\item[(e)] $\forall (i', \ell') \in \cT_i, \forall j \in [\ogamma_{i',\ell'}], \overline{\mu}_{i', \ell', j} \in \Omega_{\ell'}$,
\end{itemize} then 
$$  \prod_{(i', \ell') \in \cT_i} \prod_{j \in [\ogamma_{i', \ell'}]} 1- \overline{\mu}_{i', \ell',j}(h) = \begin{cases} 0 \textrm{ if $h \notin \ell$,} \\  \prod_{(i', \ell') \in \cT_{i,\ell}} \prod_{j \in [\ogamma_{i', \ell'}]} 1- {\mu}_{i', \ell',j}(h) \textrm{ if $h \in \ell$.} \end{cases}$$
Consequently
\beq \label{eq_phi_omu} \FF_\ell^v(\mu_{\vec \gamma} ) = \ln \left[ \sum_{h\in [k]} \prod_{(i', \ell') \in \cT_i} \prod_{j \in [\ogamma_{i', \ell'}]} 1- \overline{\mu}_{i', \ell',j}(h)\right] .\eeq Moreover, choosing the $\ogamma_{i',\ell'}$ from Poisson distributions of parameter $q_{i',\ell'}^* d'$, the event ``(a) and (b)'' happens with probability exactly $k q^*_{i,\ell}$.
This allows to write:
\begin{align*}
I &= \sum_{(i,\ell)\in\cT}q_{i,\ell}^*\sum_{\vec\gamma\in\Gamma_{i,\ell}}\prod_{(i',\ell')\in \cT_{i,\ell}} p_{q^*_{i', \ell'}d'}(\gamma_{i', \ell'})
			\int_{\Omega^{\vec \gamma}} \FF_\ell^v(\mu_{\vec\gamma})
				\bigotimes_{(i,'\ell')\in \cT_{i,\ell}} \bigotimes_{j \in [\hgamma_{i',\ell'}]} \dd \pi_{i',\ell'} (\mu_{i',\ell',j})\\
&= \frac{1}{k} \sum_{(i,\ell)\in\cT}   \sum_{\vec\ogamma\in \overline{\Gamma}_i}\prod_{(i',\ell')\in \cT_i} p_{q^*_{i', \ell'}d'}(\ogamma_{i', \ell'}) \prod_{i' \in \ell \setminus \{i\}} \vecone_{\ogamma_{i',\{i'\}} =0} \prod_{i' \in [k] \setminus \ell} \vecone_{\ogamma_{i',\{i'\}}>0}
		\\ & \hspace{0.6 cm}	\int_{\Omega^{\vec \ogamma}}   \ln \left[ \sum_{h\in [k]} \prod_{(i',\ell')\in \cT_i} \prod_{j \in [\ogamma_{i', \ell'}]} 1- \overline{\mu}_{i', \ell',j}(h)\right]  \prod_{(i',\ell') \in \cT_i} \prod_{j \in [\ogamma_{i',\ell'}]} \frac{ \vecone_{\overline{\mu}_{i',\ell',j} \in \Omega_{\ell'}}}{k q^*_{i',\ell'}}
				\bigotimes_{(i,'\ell')\in \cT_i} \bigotimes_{j \in [\ogamma_{i',\ell'}]} \dd \pi_{i'} (\overline{\mu}_{i',\ell',j})
\\ &=  \frac{1}{k} \sum_{i\in [k]} \sum_{\vec\ogamma\in\overline{\Gamma}_i}\prod_{(i',\ell')\in \cT_i} p_{q^*_{i', \ell'}d'}(\ogamma_{i', \ell'})
\\ & \hspace{0.6 cm}	\int_{\Omega^{\vec \ogamma}}  \ln \left[ \sum_{h\in [k]} \prod_{(i',\ell') \in \cT_i} \prod_{j \in [\ogamma_{i', \ell'}]} 1- \overline{\mu}_{i',\ell',j}(h)\right]  \prod_{(i',\ell') \in \cT_i} \prod_{j \in [\ogamma_{i',\ell'}]} \frac{ \vecone_{\overline{\mu}_{i',\ell',j'} \in \Omega_{\ell'}}}{k q^*_{i',\ell'}}
				\bigotimes_{(i,'\ell')\in \cT_i} \bigotimes_{j \in [\ogamma_{i',\ell'}]} \dd \pi_{i'} (\overline{\mu}_{i',\ell',j}).
\end{align*}
We used (\ref{eq_phi_omu}) to go from the first to the second line, and summed over $\ell \ni i$ to go from the second to the third. Re-indexing the vector $\overline{\mu}_{\vec \ogamma}$ in a vector $\mu_{\vec \gamma}$, $\vec \gamma \in \Gamma_i$ (with $\gamma_{i'} = \sum_{\ell':(i',\ell') \in \cT} \ogamma_{i',\ell'}$), we obtain with \Lem~\ref{Lemma_hardFields}:

\begin{align*}
I &=  \frac{1}{k} \sum_{i \in [k]} \sum_{\vec \gamma \in \Gamma_i} \prod_{i' \neq i} p_{\frac{d}{k-1}}(\gamma_{i'})  \int_{\Omega^{\vec \gamma}}   \ln \left[ \sum_{h\in [k]} \prod_{i' \neq i} \prod_{j \in [\gamma_{i'}]} 1- \mu_{i',j}(h)\right]  
				\bigotimes_{i' \neq i} \bigotimes_{j \in [\gamma_{i'}]} \dd \pi_{i'} (\mu_{i',j})
\\&=  \frac{1}{k} \sum_{i \in [k]} \sum_{\gamma_1, \dots, \gamma_h = 0}^\infty \prod_{i' \in [k]} p_{\frac{d}{k-1}}(\gamma_{i'})  \int_{\Omega^{\gamma_1+ \dots +\gamma_h}}   \ln \left[ \sum_{h\in [k]} \prod_{i' \neq i} \prod_{j \in [\gamma_{i'}]} 1- \mu_{i',j}(h)\right]  
				\bigotimes_{i' \in [k]} \bigotimes_{j \in [\gamma_{i'}]} \dd \pi_{i'} (\mu_{i',j}).
\end{align*}
\end{proof}

\begin{proof}{Proof of \Prop~\ref{Prop_fix}}
The first assertion is immediate from \Lem~\ref{Lemma_GW}, while the second assertion follows from \Lem~\ref{Lemma_softFields}.
The third claim follows by combining \Cor~\ref{Cor_unimodular} with \Lem~\ref{Lemma_Bethe}.
With respect to the last assertion, we observe that for $d=(2k-1)\ln k-2\ln2+o_k(1)$ we have
	$$\ln k+\frac d2\ln(1-1/k)=\frac{\ln 2+o_k(1)}k.$$
Moreover, as $q^*=1-1/k+o_k(1/k)$ by \Lem~\ref{Lemma_GW}, one checks easily that
	\begin{equation}\label{eqUniqueZero1}
	\Erw\brk{\frac{\ln\cZ(\T_{d,k,\vec q^*})}{|\T_{d,k,\vec q^*}|}}=\frac{\ln 2+o_k(1)}k.
	\end{equation}
Further, by \Lem~\ref{Lemma_GW}
	\begin{equation}\label{eqUniqueZero2}
	\frac{\partial}{\partial d}\Erw\brk{\frac{\ln\cZ(\T_{d,k}(\vec q^*))}{|\T_{d,k}(\vec q^*)|}}=\tilde O_k(k^{-2})
		\quad\mbox{while}\quad
		\frac{\partial}{\partial d}\ln k+\frac d2\ln(1-1/k)=\Omega_k(1/k).
		\end{equation}
Combining~(\ref{eqUniqueZero1}) and~(\ref{eqUniqueZero2}) and using the third part of \Prop~\ref{Prop_fix},
we conclude that $\Sigma_k$ has a unique zero $\dc$, as claimed.
\end{proof}

\section{The cluster size}\label{Sec_cluster}

\noindent
The objective in this section is to prove \Prop~\ref{Prop_cluster}.
For technical reasons, we consider a variant of the ``planted model'' $G(n,p',\vec\sigma)$
in which the number of vertices is not exactly $n$ but $n-o(n)$.
	This is necessary because we are going to perform inductive arguments in which small parts of
	the random graph get removed.
Thus, let $\eta=\eta(n)=o(n)$ be a non-negative integer sequence.
Throughout the section, we write $n'=n-\eta(n)$.
Moreover, we let $\vec G=G(n',p',\vec\sigma)$, where $p'=d'/n$ with $d'=kd/(k-1)$ as in~(\ref{eqLemma_plantedCluster1}).
Unless specified otherwise, all statements in this section are understood to hold for {\em any} sequence $\eta=o(n)$.

\subsection{Preliminaries}\label{Sec_cluster_pre}
Assume that $G=(V,E),\sigma$, 
	let $v\in V$ and let $\omega\geq1$ be an integer.
We write $\partial^\omega_{G}(v)$ for the subgraph of $G$ consisting of all vertices at distance at most $\omega$ from $v$.
Moreover, $|\partial^\omega_{G,\sigma}(v)|$ signifies the number of vertices of $\partial^\omega_{G}(v)$.
Where the reference to $G$ is clear from the context, we omit it.
We begin with the following standard fact about the random graph $\G$.

\begin{lemma}\label{Lemma_acyclic}
Let $\omega=10\lceil\ln\ln\ln n\rceil$.
\begin{enumerate}
\item With probability $1-\exp(-\Omega(\ln^2n))$ the random graph $\G$ is such that 
	 $|\partial_{\G}^{\omega}(v)|\leq n^{0.01}$ for all vertices $v$.
\item 	\Whp\ all but $o(n)$ vertices $v$ of $\G$ are such that 
	$\partial_{\G}^{\omega}(v)$ is acyclic.
\end{enumerate}
\end{lemma}

In addition, we need to know that the ``local structure'' of the random graph $\G$ endowed with the coloring $\SIGMA$
enjoys the following concentration property.

\begin{lemma}\label{Lemma_conc}
Let $\cS$ be a set of triples $(G_0,\sigma_0,v_0)$ such that $G_0$ is a graph, $\sigma_0$ is a $k$-coloring of $G_0$, and $v_0$ is a vertex of $G_0$.
Let $\omega=10\lceil\ln\ln\ln n\rceil$
and define a random variable $S_v=S_v(\G,\SIGMA)$ by letting 
	$$S_v=\vecone_{(\partial_{\G}^\omega(v),\SIGMA|_{\partial_{\G}^\omega(v)},v)\in\cS}.$$
Further, let $S=\sum_vS_v$.
Then $S=\Erw[S]+o(n)$ \whp
\end{lemma}

\noindent
The proof of \Lem~\ref{Lemma_conc} is based on standard arguments.
The full details can be found in \Sec~\ref{Sec_Lemma_conc}.

\subsection{Warning Propagation}\label{Sec_cluster_overview}
The goal in this section is to prove \Prop~\ref{Prop_cluster}, i.e., to determine the cluster size $|\cC(\G,\SIGMA)|$.
A key step in this endeavor will be to determine the sets
	$$\cL(v)=\cbc{\tau(v):\tau\in\cC(\G,\SIGMA)}$$
of colors that vertex $v$ may take under a $k$-coloring in $\cC(\SIGMA)$.
In particular, we called a vertex {\em frozen} in $\cC(\SIGMA)$ if $\cL(v)=\cbc{\SIGMA(v)}$.
To establish \Prop~\ref{Prop_cluster}, we will first
show that the sets $\cL(v)$ can be determined by means of a process called {\em Warning Propagation},
	which hails from the physics literature (see~\cite{MM} and the references therein).
More precisely, we will see that Warning Propagation yields color sets $L(v)$ such that $L(v)=\cL(v)$ for all but $o(n)$ vertices \whp\
Crucially, by tracing Warning Propagation we will be able to determine for any given type $(i,\ell)$ how many vertices of that type there are.
Moreover, we will show that the cluster $\cC(\SIGMA)$ essentially consists of all $k$-colorings $\tau$ of $\G$ such that $\tau(v)\in L(v)$ for all $v$.
In addition, the number of such colorings $\tau$ can be calculated by considering a certain reduced graph $\GGWP(\SIGMA)$.
This graphs turns out to be a forest (possibly after the removal of $o(n)$ vertices),
and the final step of the proof consists in arguing that, informally speaking, \whp\
the statistics of the trees in this forest are given by the distribution of the
multi-type branching process from \Sec~\ref{sec:outline}.

Let us begin by describing Warning Propagation on a general graph $G$ endowed with a $k$-coloring $\sigma$.
For each edge $e=\cbc{v,w}$ of $G$ and any color $i$ we define a sequence $(\mu_{v\ra w}(i,t|G,\sigma))_{t\geq1}$ such that
	$\mu_{v\ra w}(i,t|G,\sigma)\in\cbc{0,1}$ for all $i,v,w$.
The idea is that $\mu_{v\ra w}(i,t|G,\sigma)=1$ indicates that in the $t$th step of the process vertex $v$ ``warns'' vertex $w$
that the other neighbors $u\neq w$ of $v$ force $v$ to take color $i$.
We initialize this process by having each vertex $v$ emit a warning about its original $\sigma(v)$ at $t=0$, i.e.,
	\begin{equation}\label{eqWPini}
	\mu_{v\ra w}(i,0|G,\sigma)=\vecone_{i=\sigma(v)}
	\end{equation}
for all edges $\cbc{v,w}$ and all $i\in\brk k$.
Letting $\partial v=\partial_{G}(v)$ denote the neighborhood of $v$ in $G$, for $t\geq0$ we let
	\begin{equation}\label{eqWPit} 
	\mu_{v\ra w}(i,t+1|G,\sigma)=\prod_{j\in\brk k\setminus\cbc i}\max\cbc{\mu_{u\ra v}(j,t|G,\sigma):u\in \partial v\setminus\cbc w}.
	\end{equation}
That is, $v$ warns $w$ about color $i$ in step $t+1$ iff at step $t$ it received warnings from its other neighbors $u$ ({\em not} including $w$)
about all colors $j\neq i$.
Further, 
for a vertex $v$ and $t\geq0$ we let
	$$L(v,t|G,\sigma)=\cbc{j\in\brk k:\max_{u\in \partial v}\mu_{u\ra v}(j,t|G,\sigma)=0}
		\quad\mbox{and}\quad L(v|G,\sigma)=\bigcup_{t=0}^\infty L(v,t|G,\sigma).$$
Thus, $L(v,t|G,\sigma)$ is the set of colors that vertex $v$ receives no warnings about at step $t$.
To unclutter the notation, we omit the reference to $G,\sigma$ where it is apparent from the context.

To understand the semantics of this process, observe that by construction the list $L(v,t|G,\sigma)$
only depend on the vertices at distance at most $t+1$ from $v$.
Further, if we assume that  the $t$th neighborhood $\partial^{t}v$ in $G$ is a tree, then $L(v,t|G,\sigma)$ is precisely the set
of colors that $v$ may take in $k$-colorings $\tau$ of $G$ such that $\tau(w)=\sigma(w)$ for all vertices $w$ at distance greater than $t$ from $v$,
as can be verified by a straightforward induction on $t$.
As we will see, this observation together with the fact that the random graph $\G$ contains only few short cycles (cf.\ \Lem~\ref{Lemma_acyclic})
allows us to show that for most vertices $v$ we have $\cL(v)=L(v|\G,\SIGMA)$ \whp\
In effect,  the number of $k$-colorings $\tau$ of $\G$ with $\tau(v)\in L(v|\G,\SIGMA)$ for all $v$
will emerge to be a very good approximation to the cluster size $\cC(\G,\SIGMA)$.

Counting these $k$-colorings $\tau$ is greatly facilitated by the following observation.
For a graph $G$ together with a $k$-coloring $\sigma$, let us denote by
	$\GWP(t|\sigma)$ 
the graph obtained from $G$ by removing all edges
$\{v,w\}$ such that either $|L(v,t)|<2$, $|L(w,t)|<2$ or 
$L(v,t)\cap L(w,t)=\emptyset$.
Furthermore, obtain $\GWP(\sigma)$ from $G$ by
removing all edges $\{v,w\}$ such that $L(v,t)\cap L(w,t)=\emptyset$.
We view $\GWP(t|\sigma)$ and $\GWP(\sigma)$ as decorated graphs in which
each vertex $v$ is endowed with the color list $L(v,t)$ and $L(v)$ respectively.
As before, we let $\cZ$ denote the number of legal colorings of a decorated graph.
Thus, $\cZ(\GWP(\sigma))$ is the number of colorings $\tau$ of $\GWP(\sigma)$ such that $\tau(v)\in L(v|G,\sigma)$ for all~$v$.
The key statement in this section is

\begin{proposition}\label{Lemma_WPupper}
\Whp\ we have $\ln\cZ(\GGWP(\SIGMA))=\ln|\cC(\G,\SIGMA)|+o(n)$.
\end{proposition}

We begin by proving that $\cZ(\GGWP(\SIGMA))$ is a lower bound on the cluster size \whp\
To this end, let us highlight a few elementary facts.

\begin{fact}\label{Prop_WP}
The following statements hold for any $G,\sigma$.
\begin{enumerate}
\item 	For all $v,w,i$ and all $t\geq0$ we have $\mu_{v\ra w}(i,t+1)\leq \mu_{v\ra w}(i,t)$. 
\item We have $\sigma(v)\in L(v,t)$ for all $v,t$.
		Moreover, if $\mu_{v\ra w}(i,t)=1$, then $i=\sigma(v)$.
\item There is a number $t^*$ such that for any $t>t^*$ we have $\mu_{v\ra w}(i,t)=\mu_{v\ra w}(i,t^*)$ for all $v,w,i$.
\end{enumerate}
\end{fact}
\begin{proof}
We prove (1) and (2) by induction on $t$.
In the case $t=0$ both statements are immediate from~(\ref{eqWPini}).
Now, assume that $t\geq1$ and $\mu_{v\ra w}(i,t)=0$.
Then there is a color $j\neq i$ and a neighbor $u\neq w$ of $v$ such that $\mu_{u\ra v}(j,t-1)=0$.
By induction, we have $\mu_{u\ra v}(j,t)=0$.
Hence, (\ref{eqWPit}) implies that $\mu_{v\ra w}(i,t+1)=0$.
Furthermore, if $\mu_{v\ra w}(i,t+1)=1$ for some $i\neq\sigma(v)$, then $v$ has a neighbor
$u\neq w$ such that $\mu_{u\ra v}(\sigma(v),t)=1$.
But since $\sigma(u)\neq\sigma(v)$ because $\sigma$ is a $k$-coloring,
this contradicts the induction hypothesis.
Thus, we have established (1) and (2).
Finally, (3) is immediate from (1).
\end{proof}

\begin{fact}\label{Prop_WP_upper}
If for some $t\geq0$, $\tau$ is a coloring of  $\GWP(t|\sigma)$ such that $\tau(v)\in L(v,t)$ for all $v$,
then $\tau$ is a $k$-coloring of~$G$.
Moreover, if $\tau$ is a $k$-coloring of $\GWP(\sigma)$ such that $\tau(v)\in L(v)$ for all $v$,
then $\tau$ is a $k$-coloring of $G$.
\end{fact}
\begin{proof}
Let $\cbc{v,w}$ be an edge of $G$.
Clearly, if $L(v,t)\cap L(w,t)=\emptyset$, then $\tau(v)\neq\tau(w)$.
Thus, assume that $L(v,t)\cap L(w,t)\neq\emptyset$.
Then $|L(v,t)|>1$.
Indeed, if $|L(v,t)|=1$, then by Fact~\ref{Prop_WP} we have $L(v,t)=\cbc{\sigma(v)}$ and thus $\sigma(v)\not\in L(w,t)$ by~(\ref{eqWPit}).
Similarly, $|L(w,t)|>1$.
Hence, the edge $\cbc{v,w}$ is present in $\GWP(t|\sigma)$, and thus $\tau(v)\neq\tau(w)$.
This implies the first assertion.
The second assertion follows from the first assertion and Fact~\ref{Prop_WP}, which shows that there is a finite $t$ such that
$L(v,t)=L(v)$ for all $v$.
\end{proof}

To turn Fact~\ref{Prop_WP_upper} into a lower bound on the cluster size, we are going to argue that
\whp\ in $\G$ there are {\em a lot} of frozen vertices \whp\
In fact, \whp\ the number of such frozen vertices will turn out to be so large that {\em all} colorings $\tau$ as in Fact~\ref{Prop_WP_upper}
belong to the cluster $\cC(\G,\SIGMA)$ \whp\

To exhibit frozen vertices, we consider an appropriate notion of a ``core''.
More precisely, assume that $\sigma$ is a $k$-coloring of a graph $G$.
We denote by $\core(G,\sigma)$ the largest set $V'$ of vertices with the following property.
\begin{equation}\label{eqCoreDef}
\parbox{12cm}{If $v\in V'$ and $j\neq\sigma(v)$, then $|V'\cap \sigma^{-1}(j)\cap \partial v|\geq100$.}
\end{equation}
In words, any vertex in the core has at least $100$ neighbors of any color $j\neq\sigma(v)$ that also belong to the core.
The core is well-defined; for if $V',V''$ are two sets with this property, then so is $V'\cup V''$.
The following is immediate from the definition of the core.

\begin{fact}\label{Fact_WPcore}
Assume that $v\in\core(G,\sigma)$.
Then $L(v,t)=\cbc{\sigma(v)}$ for all $t$.
\end{fact}

The core has become a standard tool in the theory of random structures in general and in random graph coloring in particular.
Indeed, standard arguments show that $\G$ has a very large core \whp\
More precisely, we have

\begin{proposition}[\cite{Danny}]\label{Prop_core}
\Whp\  $\G,\SIGMA$ are such that the following two properties hold for all sets $S\subset\brk n$ of size $|S|\leq\sqrt n$.
\begin{enumerate}
\item Let $\G'$ be the subgraph obtained from $\G$ by removing the vertices in $S$. 
	Then 
		\begin{equation}\label{eqCoreSize}
		|\core(\G',\SIGMA)\cap\SIGMA^{-1}(i)|\geq\frac nk(1-k^{-2/3})\quad\mbox{ for all $i\in\brk k$.}
		\end{equation}
\item If $v\in\core(\G',\SIGMA')$, then $\SIGMA(v)=\tau(v)$ for all $\tau\in\cC(\G,\SIGMA)$.
\end{enumerate}
\end{proposition}

\begin{corollary}\label{Cor_core}
\Whp\ we have $|\cC(\G,\SIGMA)|\geq\cZ(\GGWP(\SIGMA))$.
\end{corollary}
\begin{proof}
By \Prop~\ref{Prop_core} we may assume that~(\ref{eqCoreSize}) is true for $S=\emptyset$.
Let $\tau$  be a $k$-coloring of $\GGWP(\SIGMA)$ such that $\tau(v)\in L(v)$ for all $v$.
Then Fact~\ref{Prop_WP_upper} implies that $\tau$ is a $k$-coloring of $\G$.
Furthermore, Fact~\ref{Fact_WPcore} implies that $\tau(v)=\sigma(v)$ for all $v\in\core(\G,\SIGMA)$.
Hence, (\ref{eqCoreSize}) entails that $\rho_{ii}(\sigma,\tau)\geq 1-k^{-2/3}>0.51$ for all $i\in\brk k$.
Thus, $\tau\in\cC(\G,\SIGMA)$.
\end{proof}

While $\cZ(\GGWP(\SIGMA))$ provides a lower bound on the cluster size, the two numbers do not generally coincide.
This is because for a few vertices $v$, the list $L(v)$ produced by Warning Propagation may be a proper subset of $\cL(v)$.
For instance, assume that the vertices $v_1,v_2,v_3,v_4$ induce a cycle of length four such that
$\sigma(v_1)=\sigma(v_3)=1$ and $\sigma(v_2)=\sigma(v_4)=2$,
	while $v_1,v_2,v_3,v_4$ are not adjacent to any further vertices of color $1$ or $2$.
Moreover, suppose that for each color $j\in\cbc{3,4,\ldots,k}$, each of $v_1,\ldots,v_4$ has at least one neighbor of color $j$ that belongs to the core.
Then Warning Propagation yields $L(v_1)=L(v_3)=\cbc 1$ and $L(v_2)=L(v_4)=\cbc 2$.
However, $v_1,v_2,v_3,v_4$ are actually unfrozen as we might as well give color $2$ to $v_1,v_3$ and color $1$ to $v_2,v_4$.
(A bipartite sub-structure of this kind is known as a ``Kempe chain'', cf.~\cite{Molloy}.)

The reason for this problem is, roughly speaking, that we launched Warning Propagation from the initialization~(\ref{eqWPini}),
which is the obvious choice but may be too restrictive.
Thus, to obtain an upper bound on the cluster size we will start Warning Propagation from a different initialization.
Ideally, this starting point should be such that only vertices that are frozen emit warnings.
By \Prop~\ref{Prop_core}, the vertices in the core meet this condition \whp\
Thus, we are going to compare the above installment of Warning Propagation with the result
of starting Warning Propagation from an initialization where only the vertices in the core send out warnings.

Thus, given a graph $G$ be a graph together with a $k$-coloring $\sigma$
we let
	\begin{eqnarray*}
	\mu'_{v\ra w}(i,0|G,\sigma)&=&\vecone_{i=\sigma(v)}\cdot\vecone_{v\in\core(G,\sigma)},\\
	\mu'_{v\ra w}(i,t+1|G,\sigma)&=&
		\prod_{j\in\brk k\setminus\cbc i}\max\cbc{\mu'_{u\ra v}(j,t|G,\sigma):u\in \partial v\setminus\cbc w}
	\end{eqnarray*}
for all edges $\cbc{v,w}$ of $G$, all $i\in\brk k$ and all $t\geq0$.
Furthermore, let
	$$L'(v,t|G,\sigma)=\cbc{j\in\brk k:\max_{u\in \partial v}\mu_{u\ra v}'(j,t)=0}\quad\mbox{and}\quad 
	L'(v|G,\sigma)=
			\bigcap_{t=0}^\infty L'(v,t|G,\sigma).$$
As before, we drop $G,\sigma$ from the notation where possible.

Similarly as before, we can use the lists $L'(v,t)$ to construct a decorated reduced graph.
Indeed,  let $\GWP'(t|\sigma)$ be the graph obtained from $G$ by removing all edges
$\{v,w\}$ such that $|L'(v,t)|<2$ or $|L'(w,t)|<2$ or $L'(v,t)\cap L'(w,t)=\emptyset$.
We decorate each vertex in this graph with the list $L'(v,t)$.
In addition, let $\GWP'(\sigma)$ be the graph obtain from $G$ by removing all edges
$\{v,w\}$ such that $L'(v)\cap L'(w)=\emptyset$ endowed with the lists $L(v)$.

\begin{fact}\label{Prop_WP'}
The following statements hold for all $G,\sigma$.
\begin{enumerate}
\item For all $v$ we have $\sigma(v)\in L'(v)$.
	Moreover, if there are $j,t,w$ such that $\mu_{v\ra w}'(j,t)=1$, then $j=\sigma(v)$.
\item If $v\in\core(G,\sigma)$, then $L'(v,t)=\cbc{\sigma(v)}$ for all $t$.
\item We have $\mu_{v\ra w}'(i,t+1)\geq \mu_{v\ra w}'(i,t)$. 
\item There is a number $t^*$ such that for any $t>t^*$ we have $\mu_{v\ra w}'(i,t)=\mu_{v\ra w}'(i,t^*)$ for all $v,w,i$.
\end{enumerate}
\end{fact}
\begin{proof}
This follows by induction on $t$ (cf.\ the proof of Fact~\ref{Prop_WP}).
\end{proof}

\begin{lemma}\label{Fact_WPupper}
\Whp\ for all vertices $v$ we have $\cL(v)=\cbc{\tau(v):\tau\in\cC(\G,\SIGMA)}\subset L'(v|\G,\SIGMA)$.
\end{lemma}
\begin{proof}
\Prop~\ref{Prop_core} shows that \whp
	\begin{equation}\label{eqWPUpper0}
	\tau(v)=\sigma(v)\quad\mbox{for all }v\in\core(\G,\SIGMA).
	\end{equation}
Assuming~(\ref{eqWPUpper0}), we are going to prove by induction on $t$ that
	\begin{equation}\label{eqWPUpper1}
	\cL(v)\subset L'(v,t)\qquad\mbox{for all }v\in\brk n,t\geq0.
	\end{equation}
By construction, for any vertex $v$ and any color $j$ we have $j\in L'(v,0)$, unless $v$ has a neighbor $w\in\core(\G,\SIGMA)$ such that $\sigma(w)=j$.
Moreover, if such a neighbor $w$ exists, (\ref{eqWPUpper0}) implies that \whp\ $\tau(w)=j$ and thus $\tau(v)\neq j$ for all $\tau\in\cC(\sigma)$.
Hence, (\ref{eqWPUpper1}) is true for $t=0$.

Now, assume that~(\ref{eqWPUpper1}) holds for $t$.
Suppose that $j\not\in L'(v,t+1)$.
Then $v$ has a neighbor $u$ such that $\mu_{u\ra v}'(j,t+1)=1$.
Therefore, for each $l\neq j$ there is $w_l\neq v$ such that $\mu_{w_l\ra u}'(l,t)=1$.
Consequently, $L'(u,t)=\cbc j$.
Hence, by induction we have $\tau(u)=j$ and thus $\tau(v)\neq j$ for all $\tau\in\cC(\G,\SIGMA)$.
\end{proof}

As an immediate consequence of \Lem~\ref{Fact_WPupper} we obtain

\begin{corollary}\label{Cor_WPupper}
\Whp\ we have
	$|\cC(\G,\SIGMA)|\leq\cZ(\GGWP'(\SIGMA))$.
\end{corollary}

Combining \Cor~\ref{Cor_core} and \Cor~\ref{Cor_WPupper}, we see that $\cZ(\GGWP(\SIGMA))\leq|\cC(\G,\SIGMA)|\leq\cZ(\GGWP'(\SIGMA))$ \whp\
To complete the proof of \Prop~\ref{Lemma_WPupper}, we are going to argue that $\ln\cZ(\GGWP'(\SIGMA))=\ln \cZ(\GGWP(\SIGMA))+o(n)$ \whp\

To this end, we need one more general construction.
Let $G$ be a graph and let $\sigma$ be a $k$-coloring of $G$.
Let $t\geq0$ be an integer.
For each vertex $v$ of $G$ we define a rooted, decorated graph $T(v,t|G,\sigma)$ as follows.
	\begin{itemize}
	\item The graph underlying $T(v,t|G,\sigma)$ is the connected component of $v$ in $\GWP(v,t|G,\sigma)$.
	\item The root of $T(v,t|G,\sigma)$ is $v$.
	\item The type of each vertex $w$ of $T(v,t|G,\sigma)$ is $(\sigma(w),L(w,t|G,\sigma))$.
	\end{itemize}
Analogously we obtain a rooted, decorated graph $T(v|G,\sigma)$ from $\GWP(\sigma)$,
$T'(v,t|G,\sigma)$  from $\GWP'(t|\sigma)$ and $T'(v|G,\sigma)$ from $\GWP'(\sigma)$.

Of course, the total number $\cZ(\GWP(\sigma))$ of legal colorings of $\GWP(\sigma)$ is just the product
of the number of legal colorings of all the connected components of $\GWP(\sigma)$.
The following lemma shows that \whp\ for all but $o(n)$ vertices the components in
$\GGWP(\SIGMA)$ and $\GGWP'(\SIGMA)$ coincide.

\begin{lemma}\label{Lemma_UpperMatchesLower}
\Whp\ $\G,\SIGMA$ is such that $T(v|\G,\SIGMA)=T'(v|\G,\SIGMA)$
for all but $o(n)$ vertices $v$. 
\end{lemma}

The main technical step towards the proof of \Lem~\ref{Lemma_UpperMatchesLower} is to show
that \whp\ most of the components  $T'(v|\G,\SIGMA)$ are ``small'' by comparison to $n$.
Technically, it is easier to establish this statement for $T'(v,0|\G,\SIGMA)$, which contains $T'(v|\G,\SIGMA)$
as a subgraph due to the monotonicity property Fact~\ref{Prop_WP'}, (3).

\begin{lemma}\label{Lemma_T''}
For any $\eps>0$ there is a number $\omega=\omega(\eps)>0$ such that \whp\
for at least $(1-\eps)n$ vertices $v$
the component $T'(v,0|\G,\SIGMA)$ contains no more than $\omega$ vertices.
\end{lemma}

The proof of \Lem~\ref{Lemma_T''}, which we defer to \Sec~\ref{Sec_T''}, is
a bit technical but based on known arguments.
\Lem~\ref{Lemma_acyclic} shows that \whp\ for most vertices $v$ such that
$T'(v,0|\G,\SIGMA)$ contains at most, say, $\omega=\lceil\ln\ln\ln n\rceil$ vertices,
$T'(v,0|\G,\SIGMA)$ is a tree.
In this case, the following observation applies.

\begin{lemma}\label{Lemma_ab}
Let $G$ be a graph and let $\sigma$ be a $k$-coloring of $G$.
Assume that $T'(v,0|G,\sigma)$ is a tree on $\omega$ vertices for some integer $\omega\geq1$.
Then for any vertex $y$ in $T'(v,0|G,\sigma)$ we have $L(y|G,\sigma)=L'(y|G,\sigma)$.
Moreover, if $T'(v,0|G,\sigma)$ has $\omega$ vertices, then
	$L(y|G,\sigma)=L(y,\omega+2|G,\sigma)$ and $L'(y|G,\sigma)=L'(y,\omega+2|G,\sigma)$.
\end{lemma}
\begin{proof}
We begin by establishing the following statement.
	\begin{equation}\label{eqLemma_ab_a}
	\parbox{14cm}{If $\cbc{x,z}$ is an edge of $G$ such that $x$ belongs to $T'(v,0)$ and $z$ does not belong to $T'(v,0)$
		then for any $t>0$ and any $j\in L'(x,0)$ we have $\mu_{z\ra x}(j,t)=\mu_{z\ra x}'(j,t)=\vecone_{L'(z,0)=\cbc j}$.}
	\end{equation}
To prove~(\ref{eqLemma_ab_a}), we consider two cases.
\begin{description}
\item[Case 1: $|L'(z,0)|>1$] we have $L'(x,0)\cap L'(z,0)=\emptyset$, because $T'(v,0)$ is a component of $\GWP'(0|\sigma)$.
	In particular, $\sigma(z)\not\in L'(x,0)$.
	As Facts~\ref{Prop_WP} and~\ref{Prop_WP'} show that $\mu_{z\ra x}(j,t)=1$ or $\mu_{z\ra x}'(j,t)=1$ only if $j=\sigma(z)$, we conclude that
	$\mu_{z\ra x}(j,t)=\mu_{z\ra x}'(j,t)=0$ for any $t>0$.
\item[Case 2: $|L'(z,0)|=1$] by Facts~\ref{Prop_WP} and~\ref{Prop_WP'} we have $L'(z,0)=\cbc{\sigma(z)}$.
	Hence, for any $j\neq\sigma(z)$ vertex $z$ has a neighbor $u_j$ in the core such that $\sigma(u_j)=j$.
	Since Fact~\ref{Fact_WPcore} and Fact~\ref{Prop_WP'} entail that $\mu_{u_j\ra z}(j,t-1)=\mu_{u_j\ra z}'(j,t-1)=1$ for all $t>0$,
		we see that $\mu_{z\ra x}(\sigma(z),t)=\mu_{z\ra x}'(\sigma(z),t)=1$ for all $t>0$.
	Moreover, once more by Facts~\ref{Prop_WP} and~\ref{Prop_WP'} we have $\mu_{z\ra x}(i,t)=\mu_{z\ra x}'(i,t)=0$ for all $i\neq\sigma(z)$.
\end{description}
Hence, in either case we obtain $\mu_{z\ra x}(j,t)=\mu_{z\ra x}'(j,t)=\vecone_{L'(z,0)=\cbc j}$, as claimed.

Now, pick and fix an arbitrary vertex $y$ in $T'(v,0)$.
We define the {\em $y$-height} $h_y(x)$ of a vertex $x\neq y$ in $T'(v,0)$ as follows.
Since  $T'(v,0)$ is a tree, there is a unique path from $x$ to $y$ in $T'(v,0)$.
Let $P_y(x)$ be the neighbor of $x$ on this path.
Then $h_y(x)$ is the maximum distance from $x$ to a leaf of $T'(v,0)$ that belongs to the component of $x$
in the subgraph of $T'(v,0)$ obtained by removing the edge $\cbc{x,P_y(x)}$.
We claim that for all $j\in\brk k$,
	\begin{equation}\label{eqLemma_UpperMatchesLower2}
	\mu_{x\ra P_y(x)}(j,t)=\mu_{x\ra P_y(x)}(j,h_y(x)+2)=\mu_{x\ra P_y(x)}'(j,h_y(x)+2)
		=\mu_{x\ra P_y(x)}'(j,t)\quad\mbox{if $t>h_y(x)+1$.}
	\end{equation}

The proof of (\ref{eqLemma_UpperMatchesLower2}) is by induction on $h_y(x)$.
To get started, suppose that $h_y(x)=0$.
Then $x$ is a leaf of $T'(v,0)$.
Let $U$ be the set of all neighbors $u\neq P_y(x)$ of $x$ in $G(n,p',\sigma)$.
Then~(\ref{eqLemma_ab_a}) shows that
	$$\mu_{u\ra x}(j,t)=\mu_{u\ra x}'(j,t)=\vecone_{L'(u,0)=\cbc j}\quad\mbox{ for all $u\in U$, $t>0$.}$$
Hence, for all $j\in\brk k$, $t>0$ we have
	$$\mu_{x\ra P_y(x)}(j,t+1)=\mu_{x\ra P_y(x)}'(j,t+1)=\prod_{j\neq\sigma(x)}\max\cbc{\vecone_{L'(u,0)=\cbc j}:u\in U}\quad\mbox{ if }h_y(x)=0.$$

Now, assume that $h_y(x)>0$.
Let $U$ be the set of all neighbors $u$ of $x$ that do not belong to $T'(v,0)$, and let $U'$ be the set of all neighbors $u'\neq P_y(x)$ of $x$ in $T'(v,0)$.
Then all $u'\in U'$ satisfy $h_y(u')<h_y(x)$.
Moreover, $P_y(u')=x$.
Therefore, by induction
	\begin{eqnarray}\label{eqLemma_UpperMatchesLower3}
	\mu_{u'\ra x}(j,t-1)&=&\mu_{u'\ra x}(j,h_y(x)+1)\\
		&=&\mu_{u'\ra x}'(j,h_y(x)+1)
		=\mu_{u'\ra x}'(j,t-1)\quad\mbox{for all $u'\in U'$, $j\in\brk k$, $t>h_y(x)+1$.}
		\nonumber
	\end{eqnarray}
Furthermore, (\ref{eqLemma_ab_a}) implies that for any $t>1$,
	\begin{equation}\label{eqLemma_UpperMatchesLower4}
	\mu_{u\ra x}(j,t-1)=\mu_{u\ra x}'(j,t-1)=\vecone_{L'(u,0)=\cbc j}\quad\mbox{for any }j\in\brk k.
	\end{equation}
Combining~(\ref{eqLemma_UpperMatchesLower3}) and~(\ref{eqLemma_UpperMatchesLower4}), we see that for any $t>h_y(x)+1$ and any $i\in\brk k$,
	\begin{eqnarray*}
	\mu_{x\ra P_y(x)}(i,t)&=&\prod_{j\neq i}\max\cbc{\mu_{u\ra x}(j,t-1),\mu_{u'\ra x}(j,t-1):u\in U,u'\in U'}\qquad\mbox{[by~(\ref{eqWPit})]}\\
		&=&\prod_{j\neq i}\max\cbc{\mu_{u\ra x}'(j,t-1),\mu_{u'\ra x}'(j,t-1):u\in U,u'\in U'}=\mu_{x\ra P_y(x)}'(i,t),\\
	\mu_{x\ra P_y(x)}(i,t)&=&\prod_{j\neq i}\max\cbc{\mu_{u\ra x}(j,t-1),\mu_{u'\ra x}(j,t-1):u\in U,u'\in U'}\\
		&=&\prod_{j\neq i}\max\cbc{\mu_{u\ra x}(j,h_y(x)+1),\mu_{u'\ra x}(j,h_y(x)+1):u\in U,u'\in U'}\\
		&=&\mu_{x\ra P_y(x)}(i,h_y(x)+2),\qquad\mbox{and analogously}\\
	\mu_{x\ra P_y(x)}'(i,t)
		&=&\mu_{x\ra P_y(x)}'(i,h_y(x)+2).
	\end{eqnarray*}
This completes the proof of~(\ref{eqLemma_UpperMatchesLower2}).

Finally, we observe that $h_y(x)\leq\omega=|T'(v,0)|$ for all $x$.
Hence, applying (\ref{eqLemma_UpperMatchesLower2}) to the neighbors $x$ of $y$ in $T'(v,0)$, we obtain
$\mu_{x\ra y}(j,t)=\mu_{x\ra y}(j,\omega+2)=\mu_{x\ra y}'(j,\omega+2)=\mu'_{x\ra y}(j,t)$ for all $j\in\brk k$ and all $t>\omega+1$.
Together with~(\ref{eqLemma_ab_a}), this show that for any $y\in T'(v,0)$ and any vertex $x$ that is adjacent to $y$ in $G$ we have
	\begin{equation}\label{eqLemma_UpperMatchesLower7}
	\mu_{x\ra y}(j,t)=\mu_{x\ra y}(j,\omega+2)=\mu_{x\ra y}'(j,\omega+2)=\mu'_{x\ra y}(j,t)\quad\mbox{for all $j\in\brk k$ and all $t>\omega+1$.}
	\end{equation}
Combining~(\ref{eqLemma_UpperMatchesLower7}) with the monotonicity properties from Facts~\ref{Prop_WP} and~\ref{Prop_WP'},
we see that $L(y)=L(y,\omega+2)=L'(y,\omega+2)=L'(y)$, as desired.
\end{proof}

\begin{proof}[Proof of \Lem~\ref{Lemma_UpperMatchesLower}]
\Lem~\ref{Lemma_T''} implies that all but $o(n)$ vertices $v$ we have $|T'(v,0)|\leq\ln\ln\ln n$ \whp\
Together with \Lem s~\ref{Lemma_acyclic}, this implies that \whp\ $T'(v,0)$ is a tree for all but $o(n)$ vertices $v$.
Thus, assume in the following that $v$ is such that $T'(v,0)$ is a tree.

It is immediate from Facts~\ref{Prop_WP}, \ref{Fact_WPcore} and~\ref{Prop_WP'} that
	$L(w)\subset L'(w)\subset L'(w,0)$ for all vertices $w$.
Therefore, $\GGWP(\SIGMA)\subset\GGWP'(\SIGMA)\subset\GWP'(0|\SIGMA)$ and thus
	\begin{equation}\label{eqtrivInclusion}
	T(v)\subset T'(v)\subset T'(v,0).
	\end{equation}
Conversely, \Lem~\ref{Lemma_ab} shows that $L(x)=L'(x)$ for all vertices $x$ in $T'(v,0)$.
Together with~(\ref{eqtrivInclusion}), this implies that $T(v)=T'(v)$.
\end{proof}

\begin{proof}[Proof of \Prop~\ref{Lemma_WPupper}]
By \Cor~\ref{Cor_core} and \Cor~\ref{Cor_WPupper} we have $\cZ(\GWP(\SIGMA))\leq|\cC(\G,\SIGMA)|\leq\cZ(\GWP'(\SIGMA))$ \whp\
Thus, it suffices to show that $\ln\cZ(\GWP(\SIGMA))=\ln\cZ(\GWP'(\SIGMA))+o(n)$ \whp\
Indeed, because the various connected components of $\GGWP(\SIGMA)$ can be colored independently, we find that
	\begin{eqnarray}\label{eqLemma_WPupper1}
	\ln\cZ(\GGWP(\SIGMA))=\sum_{v\in[n']}\frac{\ln\cZ(T(v|\G,\SIGMA))}{|T(v|\G,\SIGMA)|},&&\ln\cZ(\GWP(\SIGMA)')=\sum_{v\in [n']}\frac{\ln\cZ(T'(v|\G,\SIGMA))}{|T'(v|\G,\SIGMA)|}.
	\end{eqnarray}
Clearly, for any vertex $v$ we have $\frac{\ln\cZ(T(v|\G,\SIGMA))}{|T(v|\G,\SIGMA)|},\frac{\ln\cZ(T'(v|\G,\SIGMA))}{|T'(v|\G,\SIGMA)|}\leq\ln k$.
Hence, \Lem~\ref{Lemma_UpperMatchesLower} shows that \whp
	\begin{eqnarray}\label{eqLemma_WPupper2}
	\sum_{v\in [n']}\frac{\ln\cZ(T(v|\G,\SIGMA))}{|T(v|\G,\SIGMA)|}\sim\sum_{v\in [n']}\frac{\ln\cZ(T'(v|\G,\SIGMA))}{|T'(v|\G,\SIGMA)|}.
	\end{eqnarray}
Finally, the assertion follows from~(\ref{eqLemma_WPupper1}) and~(\ref{eqLemma_WPupper2}).
\end{proof}

\subsection{Counting legal colorings}
\Prop~\ref{Lemma_WPupper} reduces the proof of \Prop~\ref{Prop_cluster} to the problem of counting the legal colorings
of the reduced graph $\GGWP(\SIGMA)$.
\Lem~\ref{Lemma_T''} implies that \whp\ $\GGWP(\SIGMA)$ is a forest consisting mostly of trees of size, say at most $\ln\ln\ln n$.
In this section we are going to show that \whp\ the ``statistics'' of these trees follows the distribution of the random tree 
generated by the branching process from \Sec~\ref{sec:outline}.
To formalise this, let $\T=\T_{d,k,\vec q^*}$ with $\vec q^*$ from~(\ref{eqq*}) denote the random
isomorphism class of rooted, decorated trees produced by the process $\GW(d,k,\vec q^*)$.
Moreover, for a be a rooted, decorated tree $T$  
let $H_T$ be the number of vertices $v$ in $\GGWP(\SIGMA)$ such that $T(v|\G,\SIGMA)\ism T$.
In this section we prove

\begin{proposition}\label{Prop_empiricalTrees}
If $T$ is such that $\pr[T\in\T]>0$, then $(\frac1nH_T)_{n\geq1}$ converges to $\pr\brk{T\in\vec T}$ in probability.
\end{proposition}

We begin by showing that the fixed point problem $\vec q^*=F(\vec q^*)$ with $F$ from~(\ref{eqThm_condF}) provides
a good approximation to the number of vertices $v$ such that $L(v|\G,\SIGMA)=\cbc i$ for any $i$.
To this end, we let
	$$\vec q^0=(1/k,\ldots,1/k)\quad\mbox{and}\quad \vec q^t=F(\vec q^{t-1})\quad\mbox{for }t\geq1.$$
In addition, let $Q_{i}(t|\G,\SIGMA)$ be the set of vertices $v$ of $\G$ such that $L(v,t|\G,\SIGMA)=\cbc i$.

\begin{lemma}\label{Prop_empirical}
For any $i\in\brk k$ 
and any fixed $t>0$ we have
	$\frac 1n|Q_{i}(t|\G,\SIGMA)|=q_{i}^{t}+o(1)$ \whp
\end{lemma}
\begin{proof}
We proceed by induction on $t$.
To get started, we set
	$Q_{i}(-1|\G,\vec\sigma)=\SIGMA^{-1}(i)$ and
		$q_{i}^{-1}=1/k.$
Then \whp\ $\frac 1n|Q_{i}(-1|\G,\vec\sigma)|=q_i^{-1}+o(1)$.

Now, assuming that $t\geq0$ and that the assertion holds for $t-1$, we are going to argue that
	\begin{equation}\label{eqempiricalInduction1}
	\Erw[|Q_{i}(t|\G,\vec\sigma),\vec\sigma)|/n]=q_{i}^{t}+o(1).
	\end{equation}
Indeed, let $v=n'$ be the last vertex of the random graph,
and let us condition on the event that $\vec\sigma(v)=i$.
By symmetry and the linearity of expectation, 
it suffices to show that 
	\begin{equation}\label{eqempiricalInduction2}
	\pr[L(v,t|\G,\SIGMA)=\cbc i | \vec\sigma(v)=i]=kq_{i}^{t}+o(1).
	\end{equation}

To show~(\ref{eqempiricalInduction2}), 
let $\widetilde\G$ signify the subgraph obtained from $\G$ by removing $v$.
Moreover, let $\cQ^{t-1}(\eps)$ be the event that 
	\begin{equation}\label{eqempiricalInduction3a}
	|n^{-1}|Q_{j}(t-1|\widetilde\G,\vec\sigma)|-q_{j}^{t-1}|<\eps\quad\mbox{for all $j\in\brk k$}.
	\end{equation}
Since $\widetilde\G$ is nothing but a random graph $G(n'-1,p',\vec\sigma)$ with one less vertex
and as $n'-1=n-o(n)$, 
by induction we have
	\begin{equation}\label{eqempiricalInduction3}
	\pr[\cQ^{t-
1}(\eps)]=1-o(1)\qquad\mbox{for any }\eps>0.
	\end{equation}

Let $\cA(i)$ be the event that 
for each $j\in\brk k\setminus\cbc i$ there is $w\in \partial_{\G} v$ such that $L(w,t-1|\widetilde G,\SIGMA)=\cbc j$.
Given $\vec\sigma(v)=i$, we can obtain $\G$ from $\widetilde\G$
by connecting $v$ with each vertex $w\in[n'-1]$ such that $\vec\sigma(w)\neq i$ with probability $p'$ independently.
Therefore,
	\begin{eqnarray*}\nonumber
	\pr\brk{\cA(i)|\widetilde\G,\vec\sigma(v)=i}&=&
		\prod_{j\neq i}1-(1-p')^{n|Q_{j}(t-1|\widetilde\G,\vec\sigma)|}
		\sim\prod_{j\neq i}1-\exp(-np'|Q_{j}(t-1|\widetilde\G,\vec\sigma)|)\\
		&\sim&\prod_{j\neq i}1-\exp\brk{-\frac{kd}{k-1}\cdot |Q_{j}(t-1|\widetilde\G,\vec\sigma)|}.
			\label{eqempiricalInduction4}
	\end{eqnarray*}
Furthermore, for any fixed $\delta>0$ there is an ($n$-independent) $\eps>0$ such that given that $\cQ^{t-1}(\eps)$ occurs, we have
	\begin{equation}			\label{eqempiricalInduction4a}
	\abs{q^t_{i}-\prod_{j\neq i}1-\exp\bc{-\frac{kd}{k-1}\cdot |Q_{j}(t-1|\widetilde \G,\vec\sigma)|}}<\delta.
	\end{equation}
Combining~(\ref{eqempiricalInduction3}) and (\ref{eqempiricalInduction4a}), we see that for any fixed $\delta>0$ we have
	\begin{eqnarray}
	\abs{\pr\brk{\cA(i)|\vec\sigma(v)=i}-kq^t_{i}}<\delta+o(1).
			\label{eqempiricalInduction5}
	\end{eqnarray}
If $v$ is acyclic, $\vec\sigma(v)=i$ and $\cA(i)$ occurs, then $L(v,t|\G,\SIGMA)=\cbc i$.
Therefore, (\ref{eqempiricalInduction2}) follows from (\ref{eqempiricalInduction5}) and \Lem~\ref{Lemma_acyclic}.

Finally, the random variable $|Q^{t}_{i}(\G,\vec\sigma)|$ satisfies the assumptions of \Lem~\ref{Lemma_conc}.
Indeed, the event $v\in Q_{i}(t|\G,\vec\sigma)$ is determined solely by the sub-graph of $\G$ encompassing
those vertices at distance at most $t$ from $v$.
Thus, (\ref{eqempiricalInduction1}) and \Lem~\ref{Lemma_conc} imply that
	$\frac 1n|Q_{i}(t|\G,\vec\sigma)|=q_i^t+o(1)$ \whp, as desired.
\end{proof}

As a next step, we consider the statistics of the trees $T(v,\omega|\G,\SIGMA)$ with $\omega\geq0$ large but fixed as $n\ra\infty$.
Thus, for an isomorphism class $T$ of rooted, decorated graphs we let $H_{T,\omega}$ be the number of vertices
$v$ in $\GGWP(\omega|\SIGMA)$ such that $T(v,\omega|\G,\SIGMA)\in T$.

\begin{lemma}\label{Lemma_treeStat}
Assume that $T$ is an isomorphism class of rooted decorated trees such that $\pr\brk{\T\ism T}>0$.
Then for any $\eps>0$ there is $\omega>0$ such that
	$$\lim_{n\ra\infty}\pr\brk{\abs{\pr\brk{\vec T=T}-\frac1nH_{T,\omega}}>\eps}=0.$$
\end{lemma}
\begin{proof}
We observe that $\pr\brk{\vec T=T}$ is a number that depends on $T$ but not on $n$.
Hence, we assume that $\pr\brk{\vec T= T}\geq-\ln\eps>0$.
Furthermore, if $T_*$ is the isomorphism class of a rooted sub-tree of $T$,
then $\pr\brk{\vec T= T_*}\geq\pr\brk{\vec T= T}$.

The proof is by induction on the sum over the lengths of the color lists of the vertices in $T$.
In the case that $T$ consists of a single vertex $v$ of type $(i,\cbc i)$ for some $i\in\brk k$, the assertion
readily follows from \Lem~\ref{Prop_empirical}.

As for the inductive step,  pick and fix one representative $T_0\in T$.
If we remove the root $v_0$ from $T_0$, then we obtain a decorated forest $T_0-v_0$.
Each tree $T'$ in this forest contains precisely one neighbor of the root of $T_0$, which we designate as the root of $T'$.
Let $\cV$ be the set of all isomorphism classes of rooted decorated trees $T'$ obtained in this way.
Furthermore, for each 
$\hat T\in\cV$ let $y(\hat T)$  be the number of components of the forest $T_0-v_0$ that belong to the isomorphism class $\hat T$.
Let $(i_0,\ell_0)$ be the type of the root.

We are going to show that for $v=n'$ and for $\omega=\omega(T,\eps)$ sufficiently large we have
	$$|\pr\brk{T(v,\omega|\G,\SIGMA)\ism T_0}-\pr\brk{\vec T= T}|<\eps.$$
To this end, consider the graph $\widetilde\G$ obtained by removing $v$.
By \Lem~\ref{Prop_empirical} the number of vertices $w$ of $\widetilde\G$ with $L(w,\omega|\widetilde\G,\SIGMA)=\cbc j$ is $n(q_j+o_\omega(1))$ \whp\ for all $j$,
	where $o_\omega(1)$ signifies a term that tends to $0$ in the limit of large $\omega$.
Let $\cA$ be the event that this is indeed the case.
Moreover, let $\cB$ be the following event.
	\begin{itemize}
	\item $\vec\sigma(v)=i_0$.
	\item for each color $j\not\in\ell_0$, vertex $v$ has a neighbor $w$ in $\widetilde\G$ such that $L(w,\omega|\widetilde\G,\SIGMA)=\cbc j$.
	\item  $v$ does not have a neighbor $w$ with $L(w,\omega|\widetilde\G,\SIGMA)=\cbc h$ for any $h\in\ell_0$.
	\end{itemize}
Then
	\begin{eqnarray*}
	\pr\brk{\cB|\cA}&=&\frac1k\prod_{j\not\in\ell_0}\pr\brk{\Bin(n(q_j^*+o_\omega(1)),p')>0}\prod_{j\in\ell_0\setminus\cbc{i_0}}
			\pr\brk{\Bin(n(q_j^*+o_\omega(1)),p')=0}\\
		&\sim&\frac1k\prod_{j\not\in\ell_0}\pr\brk{\Po(np'(q_j^*+o_\omega(1)))>0}
			\prod_{j\in\ell_0\setminus\cbc{i_0}}\pr\brk{\Po(np'(q_j^*+o_\omega(1)))=0}
		= q^*_{i_0,\ell_0}+o_\omega(1).
	\end{eqnarray*}
Since $\pr\brk{\cA}\sim1$, we find
	\begin{equation}				\label{eqtreeStat8}
	\pr\brk\cB=q_{i_0,\ell_0}^*+o_\omega(1).
	\end{equation}

Furthermore, for each tree $T'\in\cV$ let $\widetilde Q(T')$ be the set of all vertices $w$ of $\widetilde\G$ such that $T(w,\omega|\widetilde\G,\SIGMA)\ism T'$.
In addition, let $\widetilde Q_\emptyset$ be the set of all vertices $w$ of $\widetilde\G$ that satisfy none of the following conditions:
	\begin{itemize}
	\item $w\in\bigcup_{T'\in\cV}Q(T')$.
	\item $w\in\SIGMA^{-1}(i_0)$.
	\item $L(w,\omega|\widetilde\G,\SIGMA)=\cbc j$ for some $j\in\brk k$.
	\end{itemize}
Further, let $q(T')=\pr\brk{T'\in\vec T}$ and let
	$$q_\emptyset=(1-q^*)(1-1/k)-\sum_{T'\in\cV}q(T').$$
Let $\cQ$ be the event that 
	$|\widetilde Q(T')|/n= 
			q(T')+o_\omega(1)$  for all $T'\in\cV$
	and that
	$|\widetilde Q_\emptyset|/n=q_\emptyset+o_\omega(1).$
Then $$\pr\brk{\cQ}\sim1$$ by induction.
Further, let $\cY$ be the event that for each $T'\in\cV$ we have 
	$y(T')=|\partial v\cap \widetilde Q(T')|$
and $\partial v\cap \widetilde Q_\emptyset=\emptyset$.
Then
	\begin{eqnarray}\nonumber
	\pr\brk{\cY|\cB}&\sim&\pr\brk{\cY|\cB,\cQ}=
		(1-p')^{n(q_\emptyset+o_\omega(1))}\prod_{T'\in\cV}\pr\brk{\Bin(n(q(T')+o_\omega(1)),p')=y(T')}\\
		&=&o_\omega(1)+\exp(-np'q_\emptyset)\prod_{T'\in\cV}\pr\brk{\Po(np'q(T'))=y(T')}\nonumber\\
		&=&o_\omega(1)+\exp(-d'q_\emptyset)\prod_{T'\in\cV}\pr\brk{\Po(d'q(T'))=y(T')}
			=o_\omega(1)+\pr\brk{T\in\T_{i_0,\ell_0}}. 
			\label{eqtreeStat111}
	\end{eqnarray}
The last equality sign follows from the fact that in tree $\T_{i_0,\ell_0}$,
the root has a Poisson number of children of possible ``shape'' $T'$.
Combining~(\ref{eqtreeStat8}) and~(\ref{eqtreeStat111}),  we find that
	\begin{equation}\label{eqtreeStat112}
	\pr\brk{\cB\cap\cY}=\pr\brk{T\in\T}+o_\omega(1).
	\end{equation}

Let $\cR$ be the event that $v$ is acyclic.
By \Lem s~\ref{Lemma_acyclic} 
we have $\pr\brk\cR\sim1$.
Furthermore, given $\cR$, we have $T(v,\omega|\G,\SIGMA)\in T$ iff
the event $\cB\cap\cY$ occurs.
Thus, (\ref{eqtreeStat112}) implies that
	\begin{equation}\label{eqtreeStat113}
	\pr\brk{T(v,\omega|\G,\SIGMA)\in T}=\pr\brk{\cB\cap\cY}+o(1)=\pr\brk{T\in\T}+o_\omega(1).
	\end{equation}
Moreover, (\ref{eqtreeStat113}) shows that
	\begin{equation}\label{eqtreeStat114}
	\frac1n\Erw[H_{T,\omega}]=\pr\brk{T\in\T}+o_\omega(1).
	\end{equation}
Finally, because the event $T(v,\omega|\G,\SIGMA)\in T$ is governed by the vertices at distance at most $|T|+\omega$
from $v$, \Lem~\ref{Lemma_conc} implies together with (\ref{eqtreeStat114}) that for any $\eps>0$ there is $\omega$ such that
	$$\pr\brk{|H_{T,\omega}-\pr\brk{T\in\T}|<\eps n}=1-o(1).$$
This completes the induction.
\end{proof}

\begin{lemma}\label{Lemma_omegaStable}
For any $\eps>0$ there is $\omega>0$ such \whp\ all but $\eps n$ vertices $v$ satisfy
	$T(v|\G,\SIGMA)=T(v,\omega|\G,\SIGMA).$
\end{lemma}
\begin{proof}
\Lem~\ref{Lemma_ab} implies that if $T(v|\G,\SIGMA)=T(v,\omega+2|\G,\SIGMA)$, unless
$T'(v,0|\G,\SIGMA)$ contains at least $\omega$ vertices.
Furthermore, \Lem~\ref{Lemma_T''} implies that for any fixed $\eps>0$ there is $\omega=\omega(\eps)$ such that this holds
for no more than $\eps n$ vertices \whp
\end{proof}

Finally, 
\Prop~\ref{Prop_empiricalTrees} is immediate from \Lem s~\ref{Lemma_treeStat} and~\ref{Lemma_omegaStable} and
\Prop~\ref{Prop_cluster} follows \Prop s~\ref{Lemma_WPupper} and~\ref{Prop_empiricalTrees}.

\subsection{Proof of \Lem~\ref{Lemma_T''}}\label{Sec_T''}
Set $\theta=\lceil\ln\ln n\rceil$.
Moreover, for a set $S\subset V$ let $C_S$ denote the $\sigma$-core of the subgraph of $G(n,p',\sigma)$ obtained by removing the vertices in $S$.
Further, for any vertex $w\in S$ let $\Lambda(w,S)$ be the set of colors $j\in\brk k$ such that in $G(n,p',\sigma)$ vertex
			$w$ does not have a neighbor in $\sigma^{-1}(j)\cap C_S$.
In addition, let us call $S$ {\em wobbly} in $G(n,p',\sigma)$ if the following conditions are satisfied.
\begin{description}
\item[W1] $|S|=\theta$.
\item[W2] We have $|\Lambda(w,S)|\geq2$ for all $w\in S$.
\item[W3] The subgraph of $G(n,p',\sigma)$ induced on $S$ has a spanning tree $T$	such that
			$$\Lambda(u,S)\cap \Lambda(w,S)\neq\emptyset\quad\mbox{for each edge $\cbc{u,w}$ of $T$}.$$
\end{description}
Assume that $T''(v)$ contains at least $\theta$ vertices.
If $T=(S,E_T)$ is a sub-tree on $\theta$ vertices contained in $T''(v)$, then $S$ is wobbly.
Therefore, it suffices to prove that the total number $W$ of vertices that are contained in a wobbly set $S$ satisfies
	\begin{equation}\label{eqCoreCondition0}
	\Erw[W]\leq\sum_{S\subset V:\abs S=\theta}\theta\cdot\pr\brk{S\mbox{ is wobbly}}=o(n).
	\end{equation}

To prove~(\ref{eqCoreCondition0}), we need a bit of notation.
For a set $S$ let $\cE_S$ be the event that
	\begin{equation}\label{eqCoreCondition1}
	|C_S\cap\sigma^{-1}(i)|\geq\frac nk(1-k^{-2/3})\quad\mbox{for all }i\in\brk k.
	\end{equation}
Then \Prop~\ref{Prop_core} implies that for any set $S$ of size on $\theta$ we have
	\begin{equation}\label{eqCoreCondition2}
	\pr\brk{\cE_S}\geq1-\exp(-\Omega(n)).
	\end{equation}

Further, for a vertex $w\in S$ and a set $J_w\subset\brk k\setminus\cbc{\sigma(w)}$ let $\cL(w,J_w)$ be the event that $\Lambda(w,S)\supset J_w$.
Crucially, the core $C_S$ of the subgraph of $G(n,p',\sigma)$ obtained by removing $S$ is independent of the edges between $S$ and $C_S$.
Therefore, $w$ is adjacent to a vertex $x$ in $C_S$ with $\sigma(x)\neq\sigma(w)$ with probability $p'$, independently for all such vertices $x$.
Consequently,
	\begin{equation}\label{eqCoreCondition3}
	\pr\brk{\cL(w,J_w)|\cE_S}\leq \prod_{j\in J}(1-p')^{\frac nk(1-k^{-2/3})}\leq k^{-1.99|J|}.
	\end{equation}
Moreover, due to the independence of the edges in $G(n,p',\sigma)$, the events $\cL(w,J_w)$ are independent for all $w\in S$.

Let $S\subset V$ be a set of size $\theta$.
Let us call a vertex $w\in S$ {\em rich} if $|\Lambda(w,S)|\geq\sqrt k$.
Further, let $R_S$ be the set of rich vertices in $S$.
To estimate the probability that $S$ is wobbly, we consider the following events.
\begin{itemize}
\item  Let $\cA_S$ be the event that $|R_S|\geq k^{-1/3}\theta$ 
		and that $G(n,p',\sigma)$ contains a tree $T$ with vertex set $S$.
\item Let $\cA_S'$ be the event that 
	and that $G(n,p',\sigma)$ contains a tree $T$ with vertex set $S$
		such that $$\sum_{w\in R_S}|N_T(w)|\geq\theta/2.$$
		(In words, the sum of the degrees of the rich vertices in $T$ is at least $\theta/2$.)
\item Let $\cA_S''$ be the event that $G(n,p',\sigma)$ contains a tree $T$ with vertex set $S$
		such that $$\sum_{w\in R_S}|N_T(w)|<\theta/2.$$
\item Let $\cW_S$ be the event that condition {\bf W2} is satisfied.
\item For a given tree $T$ with vertex set $S$ let $\cW_{S,T}'$ be the event that condition {\bf W3} is satisfied.
\end{itemize}
If $S$ is wobbly, then the event $\cA_S\cup (\cW_S\cap\cA_S')\cup(\cW_S\cap\cW_S'\cap\cA_S'')$ occurs.
Therefore,
	\begin{equation}\label{eqwobbly1}
	\pr\brk{S\mbox{ is wobbly}}\leq\pr\brk{\cA_S}+\pr\brk{\cW_S\cap\cA_S'\setminus\cA_S}+\pr\brk{\cW_S\cap \cW_S'\cap\cA_S''\setminus(\cA_S\cup\cA_S')}.
	\end{equation}
In the following, we are going to estimate the three probabilities on the r.h.s.\ separately.

With respect to the probability of $\cA_S$, (\ref{eqCoreCondition2}) and (\ref{eqCoreCondition3}) yield
	\begin{eqnarray*}
	\pr\brk{|R_S|\geq k^{-1/3}\theta}&\leq&\pr\brk{\neg\cE_S}+\pr\brk{\exists R\subset S,|R|=\lceil k^{-1/3}\theta\rceil:\forall w\in R:|\Lambda(w,S)|\geq\sqrt k|\cE_S}\\
		&\leq&\exp(-\Omega(n))+\bink{\theta}{k^{-1/3}\theta}\brk{\bink k{\sqrt k} k^{-1.9\sqrt k}}^{k^{-1/3}\theta}
			\leq\exp(-\sqrt k \theta).
	\end{eqnarray*}
Furthermore, by Cayley's formula there are $\theta^{\theta-2}$ possible trees with vertex set $S$.
Since any two vertices in $S$ are connected in $G(n,p',\sigma)$ with probability at most $p'$, and because edges occur independently,
we obtain
	\begin{equation}\label{eqwobbly2}
	\pr\brk{\cA_S}\leq \theta^{\theta-2}{p'}^{\theta-1}\cdot\pr\brk{|R_S|\geq k^{-1/3}\theta}\leq \theta^{\theta-2}{p'}^{\theta-1}\exp(-\sqrt k \theta).
	\end{equation}

To bound the probability of $\cW_S\cap\cA_S'\setminus\cA_S$,
let $R\subset S$ and $t\geq\theta/2$.
Moreover, let $e(S)$ denote the total number of edges spanned by $S$ in $G(n,p',\sigma)$,
and let $e(R,S)$ denote the number of edges that joint a vertex in $R$ with another vertex in $S$.
Let $\cA_S'(R,t)$ be the event  $e(S)\geq\theta-1$ and $e(R,S)=t$.
If $\cA_S'\setminus\cA_S$ occurs, then there exist $R\subset S$, $|R|\leq r=\lfloor k^{-1/3}\theta\rfloor$, and $t\geq\theta/4$ such that $\cA_S'(R,t)$ occurs.
Therefore, by the union bound,
	\begin{eqnarray}\label{eqwobbly3}
	\pr\brk{\cW_S\cap\cA_S'\setminus\cA_S}&\leq&\sum_{R\subset S:|R|\leq r}\sum_{t\geq\theta/4}\pr\brk{\cW_S\cap\cA_S'(R,t)}.
	\end{eqnarray}
Further, 
because the event $\cW_S$ is independent of the subgraph of $G(n,p',\sigma)$ induced on $S$, (\ref{eqwobbly3}) yields
	\begin{eqnarray}\label{eqwobbly4}
	\pr\brk{\cW_S\cap\cA_S'\setminus\cA_S}&\leq&\pr\brk{\cW_S}\cdot\sum_{R\subset S:|R|\leq r}\sum_{t\geq\theta/4}\pr\brk{\cA_S'(R,t)}.
	\end{eqnarray}
Because any two vertices in $S$ are connected with probability at most $p'$ independently,
the random variable $e(R,S)$ is stochastically dominated by a binomial distribution $\Bin(r\theta,p')$.
Therefore,
	\begin{equation}\label{eqwobbly5}
	\pr\brk{e(R,S)=t}\leq\pr\brk{\Bin(r\theta,p')=t}\leq\bink{r\theta}{t}{p'}^t.
	\end{equation}
Similarly, we find
	\begin{equation}\label{eqwobbly6}
	\pr\brk{e(S)\geq\theta-1|e(R,S)=t}\leq\pr\brk{\Bin\bc{\bink{\theta}{2},p'}\geq\theta-t-1}\leq\bink{\theta^2/2}{\theta-t-1}{p'}^{\theta-t-1}.
	\end{equation}
Combining~(\ref{eqwobbly5}) and~Ê(\ref{eqwobbly6}), we get
	\begin{eqnarray}\label{eqwobbly7}
	\pr\brk{\cA_S'(R,t)}&\leq&\bink{r\theta}{t}\bink{\theta^2/2}{\theta-t-1}p^{\theta-1}
	\end{eqnarray}
Further, plugging~(\ref{eqwobbly7}) into~(\ref{eqwobbly4}), we obtain
	\begin{eqnarray}\nonumber
	\pr\brk{\cW_S\cap\cA_S'\setminus\cA_S}&\leq&\pr\brk{\cW_S}\cdot2^\theta p^{\theta-1}\sum_{t\geq\theta/4}\bink{r\theta}{t}\bink{\theta^2/2}{\theta-t-1}
		\leq2^{1+\theta} p^{\theta-1}\pr\brk{\cW_S}\bink{r\theta}{\theta/4}\bink{\theta^2/2}{3\theta/4-1}\\
		&\leq&2^{1+\theta} p^{\theta-1}\pr\brk{\cW_S}\bcfr{\eul r\theta}{\theta/4}^{\theta/4}\bcfr{\eul\theta^2/2}{3\theta/4}^{3\theta/4}
			\leq\theta^\theta p^{\theta-1}k^{-\theta/13}\pr\brk{\cW_S}.
				\label{eqwobbly8}
	\end{eqnarray}
Finally, if the event $\cW_S$ occurs, then for each $w\in S$ there is $j\in\brk k\setminus\cbc{\sigma(w)}$ such that $j\in\Lambda(w,S)$.
Thus, (\ref{eqCoreCondition2}) and~(\ref{eqCoreCondition3}) yield
	\begin{eqnarray}				\label{eqwobbly9}
	\pr\brk{\cW_S}&\leq&\pr\brk{\neg\cE_S}+\prod_{w\in S}\sum_{j\neq\sigma(w)}\pr\brk{\cL(w,\cbc j)|\cE_S}\leq\exp(-\Omega(n))+k^{-0.99\theta}
		\leq k^{-0.98\theta}.
	\end{eqnarray}
Combining~(\ref{eqwobbly8}) and~(\ref{eqwobbly9}), we arrive at
	\begin{eqnarray}
	\pr\brk{\cW_S\cap\cA_S'\setminus\cA_S}&\leq&\theta^\theta p^{\theta-1}k^{-1.02\theta}.
				\label{eqwobbly8}
	\end{eqnarray}

To bound the probability of $\cA_S''$, suppose that $T$ is a tree with vertex set $S$, let $U\subset S$ and denote by $\cA_S''(T,U)$ the event that the following statements are true.
	\begin{enumerate}
	\item[(i)] $T$ is contained as a subgraph in $G(n,p',\sigma)$.
	\item[(ii)] Let $s_0=\min S$ and consider $s_0$ the root of $T$.
			Then for each $u\in U$ the parent $P(u)$ satisfies $P(u)\not\in R_S$.
	\end{enumerate}
If the event $\cA_S''\setminus(\cA_S\cup\cA_S')$ occurs, then there exist a tree $T$ and a set $U$ of size $|U|\geq\theta/3$ such that
$\cA_S''(T,U)$ occurs.
Therefore,
	\begin{equation}\label{eqwobbly10}
	\pr\brk{\cW_S\cap\cW_{S,T}'\cap\cA_S''\setminus(\cA_S\cup\cA_S')}\leq\sum_T\sum_{U:|U|\geq\theta/3}\pr\brk{\cW_S\cap\cW_{S,T}'\cap\cA_S''(T,U)}.
	\end{equation}

Fix a tree $T$ on $S$ and a set $U\subset S$, $|U|\geq\theta/3$.
Since any two vertices are connected in $G(n,p',\sigma)$ with probability at most $p'$ independently,
the probability that (i) occurs is bounded by ${p'}^{\theta-1}$.
Furthermore, if (ii) occurs and $u\in U$, then $|\Lambda(P(u),S)|\leq\sqrt k$ because $P(u)$ is not rich.
In addition, {\bf W3} requires that $\Lambda(P(u),S)\cap\Lambda(u,S)\neq\emptyset$.
There are two ways how this can come about: first, it could be that $\Lambda(P(u),S)\cap\Lambda(u,S)\setminus\cbc{\sigma(u)}\neq\emptyset$.
Then the event $\cL(u,\cbc j)$ occurs for some $j\in \Lambda(P(u),S)\setminus\cbc{\sigma(u)}$.
Hence, due to~(\ref{eqCoreCondition3})
	\begin{equation}\label{eqwobbly11}
	\pr\brk{\Lambda(P(u),S)\cap\Lambda(u,S)\setminus\cbc{\sigma(u)}\neq\emptyset|\cE_S,|\Lambda(P(u),S)|\leq\sqrt k}\leq k^{-1.49}\quad\mbox{for any }u\in U.
	\end{equation}
Alternatively, it could be that $\vec\sigma(u)\in\Lambda(P(u),S)$.
Given that $\Lambda(P(u),S)$ has size at most $\sqrt k$, the probability of this event is bounded by $k^{-1/2}$ because $\vec\sigma(u)$ is random.
Additionally, by {\bf W2} there is another color $j\in\Lambda(u)$, $j\neq\sigma(u)$.
Hence, the event $\cL(u,\cbc j)$ occurs and~(\ref{eqCoreCondition3}) yields
	\begin{equation}\label{eqwobbly12}
	\pr\brk{\vec\sigma(u)\in\Lambda(P(u),S),\Lambda(u,S)\setminus\cbc{\sigma(u)}\neq\emptyset|\cE_S,|\Lambda(P(u),S)|\leq\sqrt k}
		\leq k^{-1.49}\quad\mbox{for any }u\in U.
	\end{equation}
Combining~(\ref{eqCoreCondition2}), (\ref{eqwobbly11}) and~(\ref{eqwobbly12}), we find
	\begin{equation}\label{eqwobbly13}
	\pr\brk{\forall u\in U:\Lambda(P(u),S)\cap\Lambda(u,S)\neq\emptyset\wedge|\Lambda(P(u),S)|\leq\sqrt k}\leq\exp(-\Omega(n))+k^{-1.48|U|}.
	\end{equation}
In addition, if $w\in S\setminus U$, then {\bf W2} requires that the event $\cL(w,\cbc j)$ occurs for some $j\neq\sigma(w)$
and~(\ref{eqCoreCondition3}) yields
	\begin{equation}\label{eqwobbly14}
	\pr\brk{\forall w\in S\setminus U:\exists j\in\brk k\setminus\cbc{\sigma(w)}:\cL(w,j)|\cE_S}
		\leq k^{-0.99|S\setminus U|}.
	\end{equation}
Combining~(\ref{eqwobbly13}) and~(\ref{eqwobbly14}), we obtain
	\begin{equation}\label{eqwobbly15}
	\pr\brk{\cW_S\cap\cW'_{S,T}\cap\cA_S''(T,U)|T\subset G(n,p',\sigma)}
		\leq \exp(-\Omega(n))+k^{-0.99(\theta-|U|)}\cdot k^{-1.48|U|}\leq k^{-1.1\theta}.
	\end{equation}
Further, the probability that $T$ is contained in $G(n,p',\sigma)$ is bounded by ${p'}^{\theta-1}$.
Thus, (\ref{eqwobbly15}) implies
	\begin{equation}\label{eqwobbly16}
	\pr\brk{\cW_S\cap\cW'_{S,T}\cap\cA_S''(T,U)}
		\leq k^{-1.1\theta}{p'}^{\theta-1}.
	\end{equation}
Finally, combining~(\ref{eqwobbly10}) and~(\ref{eqwobbly16}) and using Cayley's formula, we obtain
	\begin{equation}\label{eqwobbly17}
	\pr\brk{\cW_S\cap \cW_S'\cap\cA_S''\setminus(\cA_S\cup\cA_S')}
		\leq 2^\theta\theta^{\theta-2}k^{-1.1\theta}{p'}^{\theta-1}\leq\theta^{\theta-2}{p'}^{\theta-1}k^{-1.09\theta}.
	\end{equation}

Plugging~(\ref{eqwobbly2}), (\ref{eqwobbly8}) and~(\ref{eqwobbly17}) into~(\ref{eqwobbly1}), we see that
	$$\theta\pr\brk{S\mbox{ is wobbly}}\leq2\theta^{\theta+1} p^{\theta-1}k^{-1.02\theta}.$$
Hence, (\ref{eqCoreCondition0}) yields
	\begin{eqnarray*}
	\Erw\brk{W}&\leq&2\theta^{\theta+1} {p'}^{\theta-1}k^{-1.02\theta}\cdot\bink n\theta
		\leq2\bcfr{\eul n}{\theta}^\theta\theta^{\theta+1}{p'}^{\theta-1}k^{-1.02\theta}\\
		&\leq&n(3np')^\theta k^{-1.02\theta}\leq n(7k\ln k)^\theta k^{-1.02\theta}=o(n),
	\end{eqnarray*}
as desired.

\subsection{Proof of \Lem~\ref{Lemma_conc}.}\label{Sec_Lemma_conc}
The following large deviations inequality known as {\em Warnke's inequality} facilitates the proof of \Lem~\ref{Lemma_conc}.

\begin{lemma}[\cite{Lutz}]\label{Lemma_Lutz}
Let $X_1,\ldots,X_N$ be independent random variables with values in a finite set $\Lambda$.
Assume that $f:\Lambda^N\ra\RR$ is a function, that $\Gamma\subset\Lambda^N$ is an event and that $c,c'>0$ are numbers such that 
the following is true.
	\begin{equation}\label{eqTL}
	\parbox{12cm}{If $x,x'\in\Lambda^N$ are such that there is $k\in\brk N$ such that $x_i=x_i'$ for all $i\neq k$, then
			$$|f(x)-f(x')|\leq\left\{\begin{array}{cl}
				c&\mbox{ if }x\in\Gamma,\\
				c'&\mbox{ if }x\not\in\Gamma.
				\end{array}\right.$$}
	\end{equation}
Then for any $\gamma\in(0,1]$
and any $t>0$ we have
	$$\pr\brk{|f(X_1,\ldots,X_N)-\Erw[f(X_1,\ldots,X_N)]|>t}
		\leq2\exp\bc{-\frac{t^2}{2N(c+\gamma (c'-c))^2}}+\frac{2 N}\gamma\pr\brk{(X_1,\ldots,X_N)\not\in\Gamma}
			.$$
\end{lemma}

\begin{proof}[Proof of \Lem~\ref{Lemma_conc}]
The proof is based on \Lem~\ref{Lemma_Lutz}.
Of course, we can view $(\G,\SIGMA)$ as chosen from a product space $X_2,\ldots,X_N$ with $N=2n'$ where
$X_i$ is a $0/1$ vector of length $i-1$ whose components are independent $\Be(p')$ variables
for $2\leq i\leq n'$ and where $X_i\in\brk k$ is uniformly distributed for $i>\bink{n'}2$ (``vertex exposure'').
Let $\Gamma$ be the event that $|N^\omega(v)|\leq \lambda=n^{0.01}$ for all vertices $v$.
Then by \Lem~\ref{Lemma_acyclic} we have
	\begin{eqnarray}\label{eqLemma_conc1}
	\pr\brk{\Gamma}&\geq& 1-\exp(-\Omega(\ln^2n)).
	\end{eqnarray}
Furthermore, let $\cG'$ be the graph obtained from $G$ by removing all edges $e$ that
are incident with a vertex $v$ such that $|N_{\G}^\omega(v)|>\lambda$
and let 
	$$S'=\sum_v S_v(\G',\SIGMA)=\abs{\cbc{v\in\brk{n'}:(N_{\G'}^\omega(v),\SIGMA|_{N_{\G'}^\omega(v)},v)\in\cS}}.$$
If $\Gamma$ occurs, then $S=S'$.
Hence, (\ref{eqLemma_conc1}) implies that
	\begin{eqnarray}\label{eqLemma_conc2}
	\Erw[S']&=&\Erw[S]+o(1).
	\end{eqnarray}

Moreover, the random variable $S'=f(X_2,\ldots,X_N)$ satisfies~(\ref{eqTL}) with $c=\lambda$ and $c'=n'$.
Indeed, altering either the color of one vertex $u$ or its set of neighbors can only affect those vertices $v$
that are at distance at most $\omega$ from $u$, and in $\G'$ there are no more than $\lambda$ such vertices.
Thus, \Lem~\ref{Lemma_Lutz} applied with, say, $t=n^{2/3}$ and $\gamma=1/n$ and~(\ref{eqLemma_conc1}) yield
	\begin{eqnarray}\label{eqLemma_conc3}
	\pr\brk{|S'-\Erw[S']|>t}\leq\exp(-\Omega(\ln^2n))=o(1).
	\end{eqnarray}
Finally, the assertion follows from~(\ref{eqLemma_conc2}) and~(\ref{eqLemma_conc3}).
\end{proof}

\medskip
\noindent
{\sc Acknowledgment.}
We thank Guilhem Semerjian for helpful discussions and explanations regarding the articles~\cite{pnas,LenkaFlorent} and
Nick Wormald for pointing us to~\cite[\Thm~3.8]{McDiarmid}.

\end{document}